\documentclass[12pt,oneside]{report}
\usepackage{hyperref}
\hypersetup{colorlinks=true,linkcolor=MidnightBlue,citecolor=Bittersweet}
\usepackage{amsmath}
\usepackage{amssymb}
\usepackage{amscd}
\usepackage{amsthm}
\usepackage[dvipsnames]{xcolor}
\usepackage{braket}
\usepackage{enumitem}
\usepackage{bbm}
\usepackage{graphicx}
\usepackage{chngcntr}
\usepackage{slashed}

\numberwithin{equation}{chapter}
\newtheoremstyle{plain}
{2ex} 
{2ex} 
{\rm} 
{} 
{\bfseries} 
{. } 
{0.1em} 
{} 
\theoremstyle{plain}
\newtheorem{defi}{Definition}[chapter]
\newtheorem*{defi*}{Definition}
\newtheorem{theo}[defi]{Theorem}
\newtheorem{prop}[defi]{Proposition}
\newtheorem*{prop*}{Proposition}
\newtheorem{lemm}[defi]{Lemma}  
\newtheorem*{lemm*}{Lemma}    
\newtheorem{coro}[defi]{Corollary}
\newtheorem{conj}[defi]{Conjecture}
\newtheorem{rmk}[defi]{Remark}
\def\be{\begin{equation}}
\def\ee{\end{equation}}

\def\A{\mathcal{A}}

\def\b{\mathfrak{b}}

\def\de{\partial}
\def\der{\bar{\partial}}
\def\dia{\diamond}
\def\Diff{\text{Diff}_{\h}}
\def\Diffs{\text{Diff}_{\h,\sigma}}
\def\DRn{\mathcal{D}\!\,(\gln)}
\def\dom{\text{Dom}}

\def\E{\operatorname{E}}

\def\F{\operatorname{F}}
\def\Fn{\F_{\n_+}}
\def\Ft{\tilde{\F}}
\def\Ftn{\Ft_{\n_+}}

\def\g{\mathfrak{g}}
\def\gl{\mathfrak{gl}}
\def\gln{{\mathbf{gl}}_n}

\def\hn{\h(n)}
\def\h{{\mathbf{h}}}
\def\hf{\mathfrak{h}}

\def\k{\mathfrak{k}}
\def\K{\mathbb{K}}

\def\lb{\vec{\lambda}}

\def\mathZ{\mathbb{Z}_{\geq 0}}
\def\Mc{\mathcal{M}}
\def\M{\operatorname{M}}
\def\Mb{\bar{\M}}

\def\n{\mathfrak{n}}
\def\N{\mathbb{N}}
\def\nb{\{1,\dots,n\}}

\def\PPsi{\hat{\Psi}}

\def\q{\check{\operatorname{q}}}
\def\QQ{\operatorname{Q}}

\def\R{\operatorname{R}}
\def\RR{\hat{\R}}

\def\sl{\mathfrak{sl}}
\def\S{\mathcal{S}}

\def\teL{\operatorname{L}}

\def\th{\tilde{h}}

\def\U{\operatorname{U}}
\def\Ub{\bar{\U}}
\def\Ubh{\Ub(\hf)}

\def\Un{\Ub(n)}

\def\V{\operatorname{V}}

\def\Wb{\bar{\mathcal{W}}}
\def\Wc{\mathcal{W}}
\def\W{\text{W}}

\def\Z{x}
\def\Zc{\mathcal{Z}}

\newcommand{\quom}[2]{\text{$#1$} \big/ \raisebox{-.3em}{\text{$#2$}}}
\newcommand{\qquom}[3]{\raisebox{-.3em}{\text{$#1$}} \setminus \text{$#2$} \ \big/ \ \raisebox{-.3em}{\text{$#3$}}}

\newcommand{\chim}[2]{\chi_{#1 \setminus #2}}
\newcommand{\Diffv}[1]{\text{Diff}_{\h,#1}}
\newcommand{\mtop}[2]{\genfrac{}{}{0pt}{}{#1}{#2}}

\newcommand{\vv}[1]{\raisebox{-.4em}{\Big\vert}_{\hspace{-0.15em}\raisebox{0.3em}{\text{$#1$}}}}

\begin{document}

\begin{titlepage}
\noindent
\thispagestyle{empty}
\begin{minipage}{15cm}
\begin{flushleft}
\vspace{-2cm}
   \includegraphics[height=2cm]{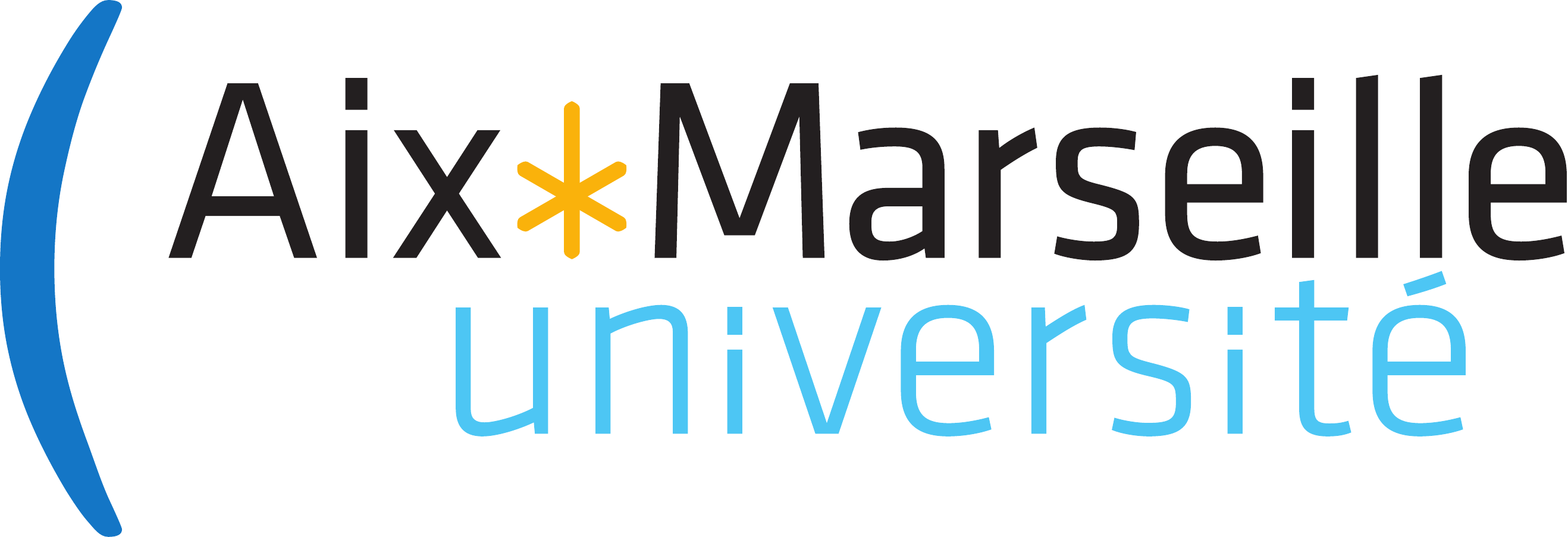}
   \hfill
   \includegraphics[height=2cm]{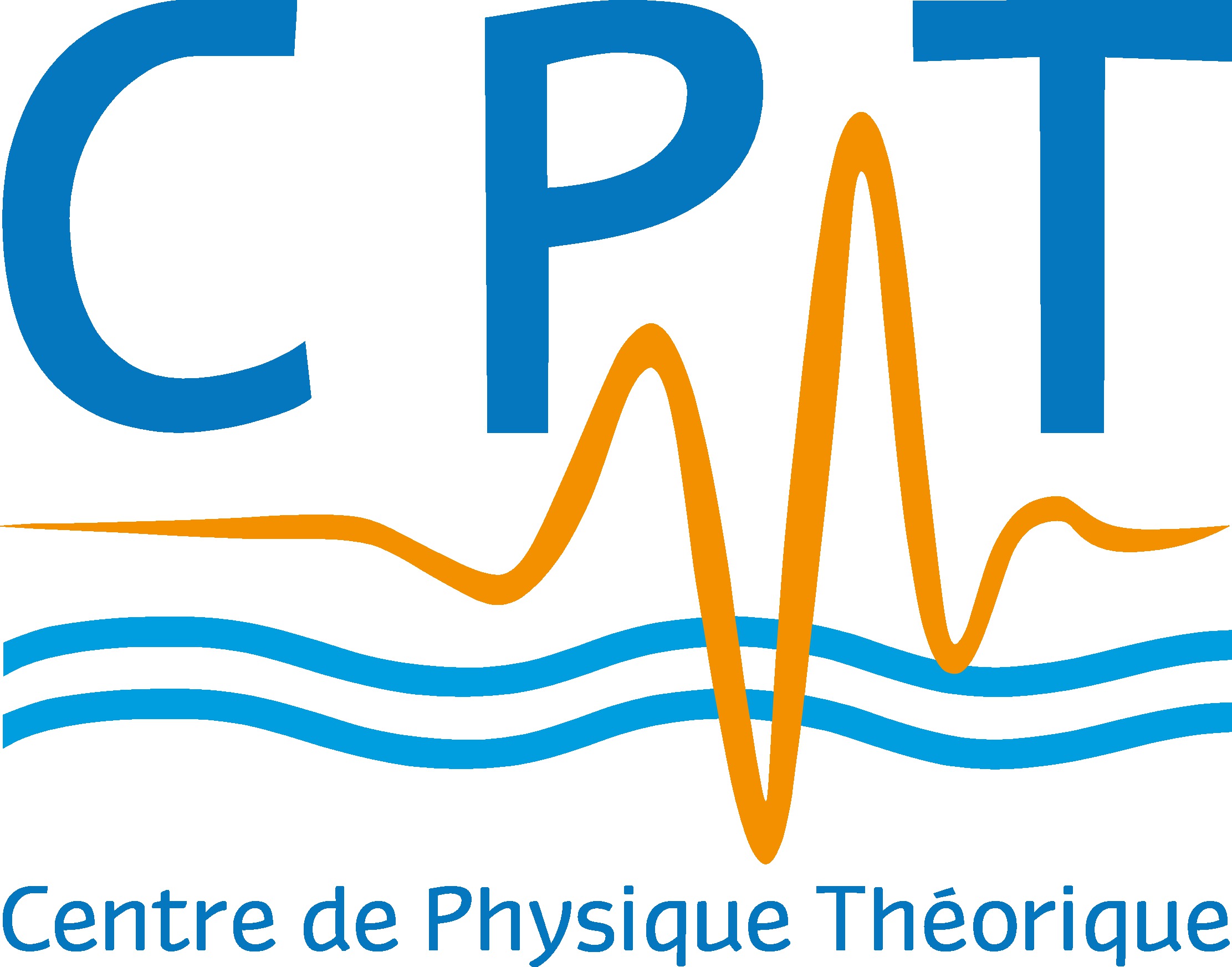}\\
\end{flushleft}
\begin{flushright}
\end{flushright}
\end{minipage}
\centering
\begin{center}
\begin{minipage}{15cm}
\centering
\textsc{Universit\'e d'Aix-Marseille \\
\vspace{0.5ex}
\'Ecole Doctorale 352\\
\vspace{0.5ex}
Facult\'e des Sciences\\
\vspace{0.5ex}
Centre de Physique Th\'eorique de Marseille\\}
\vspace{8ex}
\begin{bf}
\textsc{Th\`ese pr\'esent\'ee pour obtenir le grade universitaire de docteur\\}
\end{bf}
\vspace{4ex}
\begin{it}
Discipline :\\
\vspace{0.5ex}
\end{it}
\textsc{PHYSIQUE ET SCIENCES DE LA MATI\`ERE}

\vspace{2ex}

\begin{it}
Sp\'ecialit\'e : \\
\vspace{0.5ex}
\end{it}
Physique Th\'eorique et Math\'ematique \\
\vspace{8ex}
\begin{large}
Basile HERLEMONT
\end{large}
\end{minipage}

\vspace{10ex}

\begin{minipage}{15cm}
\textwidth 18cm
\begin{bf}
\begin{Large}
\begin{center}
Calcul diff\'erentiel sur des espaces $\h$-deform\'es\\[.2cm]
Differential calculus on $\h$-deformed spaces
\end{center}
\end{Large}
\end{bf}

\vspace{10ex}

\centering
Soutenue le 16 novembre 2017 devant le jury compos\'e de:\\
\vspace{8ex}
\begin{tabular}{lllll}
M.	        & 	Robert	        &	COQUEREAUX    &	CPT     	&	Examinateur \\
M.	        & 	Sergey	        &	KHOROSHKIN    &	ITEP     	&	Rapporteur \\
M.	        &	Oleg        &	OGIEVETSKY	&	CPT           &	Directeur de th\`ese \\
M.	        &	Vladimir	        &	RUBTSOV	&	Universit\'e d'Angers            &	Rapporteur \\  
\end{tabular}
\end{minipage}
\end{center}
\end{titlepage}
\newpage
\chapter*{Introduction}
\vskip -1cm
Let $\k$ be a reductive Lie algebra and $\k = \n_- \oplus \mathfrak{h} \oplus \n_+$ its Gauss decomposition. 
Any irreducible $\k$-module is generated by the action of the nilpotent subalgebra $\n_-$ of $\k$ on a highest weight vector and the space of highest weight vectors, 
that is, vectors annihilated by $\n_+$, is one-dimensional.
Let now $\g$ be a Lie algebra whose representations we want to understand. Assume that $\g$ contains $\k$ and the adjoint action of $\k$ on $\g$ is completely reducible. Any finite-dimensional irreducible $\g$-module $\Mc$ is a direct sum of irreducible $\k$-modules. 
It follows that $\Mc$ is generated by the action of $\n_-$ on the space $V_{\n_+}(\Mc)=\{v \in \Mc\, \vert\, \n_+ v=0\}$ of highest weight vectors for $\k$. J. Mickelsson 
\cite{M} introduced the  algebra $\S(\g,\k)$, which was later called Mickelsson algebra, which acts irreducibly 
on the space $V_{\n_+}(\Mc)$. In general, the algebra $\S(\g,\k)$ is not finitely generated. 
However, after passing to the field of fractions of $\U(\hf)$, we obtain the algebra $\Zc(\g,\k)$, called reduction algebra, whose structure is more transparent .
The algebra $\Zc(\g,\k)$ was later described as a localization of the double coset space
$$\qquom{\n_- \U(\g)}{\U(\g)}{\U(\g)\n_+} \ ,$$
equipped with the associative multiplication defined with the help of the ``extremal projector" of $\k$, introduced by Asherova-Smirnov-Tolstoy \cite{AST}. 
One may work with smaller rings of fractions than the entire field of fractions of $\U(\hf)$. There is minimal localization, see details in {\bf Section \ref{stepalgebras}},
which will be denoted $\mathcal{R}^\g_\k$ and also called reduction algebra. The reduction algebra can be defined in a more general situation, 
for an associative algebra $\A$ instead of $\U(\g)$; we then denote it by $\mathcal{R}^\A_\k$. 

In this thesis, we study the properties and generalizations of the algebra $\Diff(n,N)$,
the reduction algebra of $\W_{nN} \otimes \U(\gl_n)$ with respect to the diagonal embedding of $\U(\gl_n)$, where $\W_{nN}$ is the Weyl algebra in $nN$ variables.
The algebra $\Diff(n,N)$ plays an important role in the theory of diagonal reduction algebras (we refer to \cite{KO2,KO3,KO5,KO6} for generalities on diagonal reduction algebras). Similarly to the ring of $q$-differential operators \cite{WZ}, the algebra $\Diff(n,N)$ can be described in the R-matrix formalism. The R-matrix, needed here, is a solution of the so-called \emph{dynamical Yang--Baxter equation}.

\vskip .1cm
Let $\h$ be a finite dimensional abelian Lie algebra and $V$ a semi-simple $\h$-module. A (meromorphic) map $\RR : \h^\star \to \text{End}(V \otimes V)$ satisfies the dynamical Yang-Baxter equation if
$$\RR_{12}(\lambda) \RR_{23}(\lambda-h_1) \RR_{12}(\lambda) = \RR_{23}(\lambda-h_1) \RR_{12}(\lambda) \RR_{23}(\lambda-h_1)\ ,$$
where  $\RR_{12} = \mathbbm{1} \otimes \RR$ and $\RR_{23} = \RR \otimes \mathbbm{1}$. The operator $\RR_{23}(\lambda-h_1)$ is defined by
$$\RR_{23}(\lambda-h_1) v_1 \otimes v_2 \otimes v_3 = v_1 \otimes \left( \RR(\lambda-\mu_1) v_2 \otimes v_3 \right)$$
where $\mu_1 \in \h^\star$ is the weight of $v_1$ (we refer to \cite{F,GN,ES} for different aspects of the dynamical Yang--Baxter equation and its solutions).
The Yang--Baxter equation appears as a consistency condition for some $1+1$ dimensional quantum systems, in the theory of integrable models etc.
The R-matrices gives rise to different quadratic algebras \cite{Ka,Ma,CP}.

The rings of differential operators on $q$-spaces are quadratic algebras see \cite{WZ}. They have applications in the $q$-differential calculus, 
construction of $q$-Minkowski space etc (see \cite{SWZ,SWZ2,OSWZ,OZ,O}). These rings result from the pairing of two quantum planes. 
In a similar fashion, the ring $\Diff(n,N)$ results from the pairing of two $\hf$-deformed spaces.
It turns out that the result is surprisingly different from that for $q$-deformed vector spaces and the pairing is not unique.

\vskip .1cm
{\bf Contents of the thesis.}
\textbf{Chapter \ref{ChapterLiestepAlge}} is the introduction to the theory of reduction algebras. In \textbf{Section \ref{stepalgebras}} we give the initial definition
of reduction algebras.  \textbf{Sections \ref{EndoVermaMod}, \ref{ExtremalProj}, \ref{UsesExtreProj}} are devoted to the extremal projector and the alternative definition of 
reduction algebras. In \textbf{Section \ref{ExampleRedAlge}} we discuss two examples of reduction algebras.

\vskip .1cm
In \textbf{Chapters \ref{ChapterRinghdef}--\ref{Chaptergenesevcop}}, mainly based on \cite{HO1,HO2}, we present the results of this thesis. 
In \textbf{Chapter \ref{ChapterRinghdef}} we describe the center of $\Diff(n):=\Diff(n,1)$ and construct an isomorphism 
between certain localizations of $\Diff (n)$ and the Weyl algebra $\text{W}_n$ extended by $n$ indeterminates.
In \textbf{Chapter \ref{GeneralizedDiff}} we define and study general consistent pairings of $\h$-deformed coordinate rings. The consistency condition leads to an 
over-determined system of finite-difference equations. We give the general solution of the system. The resulting algebras $\Diffs(n)$ are labeled by a certain ``potential" $\sigma$ in a localization of $\U(\hf)$.
In \textbf{Chapter \ref{ChapterRepreDiff}} we generalize the results of \textbf{Chapter \ref{ChapterRinghdef}} to
the rings $\Diffs(n)$ and initiate the representation theory of $\Diffs(n)$.
In \textbf{Chapter \ref{Chaptergenesevcop}} we show that $\Diff(n,N)$  is an essentially unique algeba realizing the pairing of $\h$-deformed spaces 
if $N>1$. We describe the subspace of quadratic central elements of $\Diff(n,N)$ and present an action of the symmetric group $\mathbb{S}_n$ on the ring $\Diff(n,N)$ and on the diagonal reduction algebra $\DRn$, generalizing the formulas for $\Diff(n)$. 

\vskip .1cm
A possible continuation of our work is the study of the finite dimensional $\Diffs(n)$-modules and a generalization of results of \textbf{Chapter \ref{ChapterRinghdef}} 
to $\Diff(n,N)$, $N>1$.

\paragraph{Acknowledgements.}
I thank my Ph.D advisor Pr. Oleg Ogievetsky to have accepted to supervise me for this thesis. Not all of these three years have been easy. However, I am grateful to him to have tried to adapt to my personality. I also thank him to have introduced me to the study of a lot of subjects I enjoyed: the theory of reduction algebras, quantum groups, quadratic algebras, algebraic combinatorics and others.

\vskip .1cm
I thank Pr. Robert Coquereaux to have answered to a lot of my not so clear questions and to have accepted to be in the jury of my thesis. I am grateful to Pr. Sergey Khoroshkin, to have taken some of his time to answer my questions as well and to have accepted to be in the jury of my thesis. I also thank Pr. Vladimir Roubtsov to have accepted to be in the jury of my thesis.

\vskip .1cm
I ``gratefulnessly" thank my family, always present, and my friends. Thanks to Romain, Morgan \& cie, who always have been there, perfectly playing their friend's role maybe as much as cards. Thanks to Adrien for all the discussions about mathematics and other difficult matters. Finally, thanks to Irene who tried to support me during the calm and also the turbulent moments.

\tableofcontents
\chapter{Reduction algebras}\label{ChapterLiestepAlge}
\vskip -1cm
In \textbf{Section \ref{stepalgebras}} we introduce 
the step algebra, as well as the Mickelsson and reduction algebras of a Lie algebra $\g$ with respect to  its reductive Lie subalgebra $\k$ and motivate their study. The original definition \ref{defMickelssonalgebras} does not provide an efficient way to work with reduction algebras. 
However the reduction algebra $\mathcal{R}^\g_\k$ can be defined differently, 
with the help of the extremal projector, see \cite{AST,AST2,Zh3}. This description is useful in applications.

\vskip .1cm
In \textbf{Sections \ref{EndoVermaMod}} and \textbf{\ref{ExtremalProj}} we explain the existence of the extremal projector with the use 
of  the universal Verma module $\M_{\n_+} := \U(\k) / \U(\k)\n_+$. We prove its uniqueness and give its factorized form. In \textbf{Section \ref{UsesExtreProj}} 
we define the reduction algebra and Zhelobenko automorphisms. The material of \textbf{Sections \ref{EndoVermaMod}} and \textbf{\ref{ExtremalProj}} 
is based on \cite{M,Hom1,Hom2,Zh3,KO1}.

\vskip .1cm
In \textbf{Section \ref{ExampleRedAlge}}, we discuss two examples of reduction algebras. The first one is $\mathcal{R}^{\gl_3}_{\gl_2}$; we work out the details explicitly employing the extremal projector. 
The second example is the reduction algebra $\Diff(n,N)$ of the tensor product $\text{W}_{nN} \otimes \U(\gl_n)$, where $\text{W}_{nN}$ is the Weyl algebra, with respect to the diagonal embedding of $\U(\gl_n)$. Here we choose to use the particular properties of the reduction algebras minimizing the direct calculations with the extremal projector. The algebra $\Diff(n,N)$ and its generalizations are studied in \textbf{Chapters \ref{ChapterRinghdef}-\ref{Chaptergenesevcop}}.

\vskip .1cm
We assume some knowledge on Lie algebras and refer to the books \cite{Hum}, \cite{Se} and \cite{J}.

\subsubsection*{Notation}

\noindent$\bullet$ $\mathbb{Z}$ is the set of integers : $\dots,-2,-1,0,1,2,\dots$ 

\noindent$\bullet$ $\mathZ$ is the set of non-negative integers : $0,1,2,\dots$ 

\noindent$\bullet$ $\N$ is the set of natural numbers : $1,2,\dots$ 

\noindent$\bullet$ $\mathbb{C}$ is the field of complex numbers.

\noindent$\bullet$ $\K$ is an algebraically closed field of characteristic zero.

\noindent$\bullet$ The entries of the product of matrices:  $(AB)^i_j = \sum_a A^i_a B^a_j$.

\noindent$\bullet$ $A^t$ is the transpose of $A$. 

\noindent$\bullet$ $\mathbb{S}_n$ is the symmetric group on $n$ letters. The symbol $s_i$ stands for the transposition $(i,i+1)$. 

\noindent$\bullet$ The symbol $\otimes$, if not otherwise specified, stands for the tensor product over $\mathbb{C}$.

\subsubsection*{Background in Lie algebra theory}\label{ReducLieAlgebra}

\noindent{\bf The universal enveloping algebra of a Lie algebra.}
Let $\g$ be a complex Lie algebra and $\U(\g)$ its universal enveloping algebra. The diagonal map
$ \g \longrightarrow \g \oplus \g$, $a \longmapsto  (a,a)$, 
extends to the embedding of enveloping algebras 
$\U(\g) \longrightarrow \U(\g) \otimes \U(\g) $, also called diagonal.

\begin{theo}\label{PBWtheo} (The Poincar\'e$-$Birkhoff$-$Witt theorem)  
Let $g_1,\dots,g_m$ be a basis of a Lie algebra $\g$. The elements
\begin{equation}\label{basisPBWtheorem}
g_1^{k_1} \dots g_m^{k_m} \quad , \quad (k_1,\dots,k_m) \in \mathZ^m\ ,
\end{equation}
form a basis of $\U(\g)$. 
\end{theo}

\begin{defi} (The PBW property) Let $R$ be a ring over a commutative $\mathbb{C}$-ring $U$. The ring $R$ is said to have the PBW property 
with respect to a family of elements $r_1,\dots,r_m \in R$ if the monomials
$$r_1^{k_1} \dots r_m^{k_m} \quad , \quad (k_1,\dots,k_m) \in \mathZ^m\ ,$$
form a $U$-basis of $R$. 
\end{defi}

\noindent{\bf Reductive Lie algebras.}
Let $\k$ be a complex finite-dimensional reductive Lie algebra. The word ``reductive" means that the adjoint action of $\k$ is completely reducible. A reductive Lie algebra is a direct sum of a semisimple Lie algebra and an abelian Lie algebra. We fix a 
triangular decomposition of $\k$,
\begin{equation}\label{CartandecomporeducLiealg}
\k = \n_+ \oplus \hf \oplus \n_- \ .
\end{equation}
Here $\hf$ is the Cartan subalgebra and $\n_+,\n_-$ are two opposite nilpotent Lie subalgebras.
Let
\begin{equation}\label{borelsubalgereduLiealg}
\b_- = \n_- \oplus \hf \qquad \text{and} \qquad \b_+ = \n_+ \oplus \hf \ ,
\end{equation}
be the corresponding  Borel subalgebras of $\k$.

Let $\Delta$ be the set of roots, $\Delta^+$ the subset of positive roots and $\Pi$ the set of simple roots. Let
$$\k_\alpha := \left\{ x \in \k \ : \ [h,x] = \alpha(h)x \ , \ \forall h \in \hf \right\}$$
be the subspace corresponding to the root $\alpha$. We have
$$\n_- = \bigoplus_{\alpha \in \Delta^+} \k_{-\alpha} \qquad \text{and} \qquad \n_+ = \bigoplus_{\alpha \in \Delta^+} \k_{\alpha} \ .$$
We choose a basis $(e_{\alpha},e_{-\alpha})$ in $\k_\alpha\oplus\k_{-\alpha}$ such that 
the elements $(e_{\alpha},h_{\alpha},e_{-\alpha})$, where  $h_{\alpha}:=[e_{\alpha},e_{-\alpha}]$ is a coroot corresponding to $\alpha$, form an $\sl_2$-triple, that is $[h_{\alpha},e_{\pm\alpha}]=\pm 2e_{\pm\alpha}$.

\vskip .1cm
\noindent
We define the Chevalley anti-involution $\epsilon$ of $\k$ by
\begin{equation}\label{formchevalleyantiinv}
\epsilon(e_\alpha) = e_{-\alpha} \ , \ \epsilon(e_{-\alpha}) = e_\alpha \ , \  \alpha \in \Delta \ , \ \text{and}\ \ \epsilon(h)=h\ ,\ h\in\hf.
\end{equation}
A linear order $\prec$ on the set $\Delta^+_\hf$ of positive roots of $\g$ is said to be convex if for any $\alpha, \beta \in \Delta^+_\hf(\g)$, 
$\alpha\prec\beta$, such that $\alpha+\beta \in \Delta^+_\hf$ we have
\begin{equation}\label{equaconvexorder}
\alpha \prec \beta \Longrightarrow 
\alpha \prec \alpha+\beta \prec \beta \ .
\end{equation}
A positive part of the root system of any reductive Lie algebra admits a convex order. 

\vskip .1cm
\noindent{\bf Weyl group.}
Let $\Pi = \{\alpha_1,\dots,\alpha_n\}$ and $(a_{ij})_{1\leq i,j \leq n}$ the Cartan matrix of $\k$.. The Weyl group $W$ of $\k$
is generated by the simple reflections $\sigma_i$ corresponding to the roots $\alpha_i$.
The reflections $\sigma_1,\dots,\sigma_n$ satisfy the following braid relations
\begin{equation}\label{braidrela}
\underset{m_{ij}}{\underbrace{\sigma_i \sigma_j \sigma_i\dots}}  = \underset{m_{ij}}{\underbrace{\sigma_j \sigma_i \sigma_j\dots}}   \quad , \quad i \neq j \ ,
\end{equation}
where $m_{ij} = 2$ if $a_{ij}=0$; $m_{ij} = 3$ if $a_{ji} a_{ij}=1$; $m_{ij} = 4$ if $a_{ji} a_{ij}=2$; $m_{ij} = 6$ if $a_{ji} a_{ij}=3$. 

\vskip .1cm
\noindent
For any $\alpha \in \Delta$, we define a map $T_\alpha$ from $\k$ to $\k$  by
\begin{equation}\label{liftweylgroupgene}
T_\alpha = \text{exp}(\text{ad}_{e_\alpha}) \circ \text{exp}(-\text{ad}_{e_{-\alpha}}) \circ \text{exp}(\text{ad}_{e_\alpha}) \ .
\end{equation}
The restriction of $T_\alpha$ on $\hf$ coincides with the action of $W$ on $\hf$. We denote by the same symbol the extension of the 
automorphism  $T_\alpha$  to $\U(\k)$.

\vskip .1cm
\noindent
Let $T_i := T_{\alpha_i}$, $i \in \{1,\dots,n\}$. The maps $T_i$'s satisfy the same braid relations \eqref{braidrela} as the $\sigma_i$'s.

\section{Origins}\label{stepalgebras}

Let $\g$ and $\k$ be finite dimensional complex Lie algebras. Let $\k \subset \g$ such that the adjoint action of $\k$ in $\g$ is completely reducible. We say that $\k$ is reductive in $\g$. In particular, $\k$ is a reductive Lie algebra and admits a triangular decomposition \eqref{CartandecomporeducLiealg}.
\begin{defi}
For a set $I \subset \U(\g)$, the right normalizer of $I$ is $N(I) = \{ a \in \U(\g) : Ia \subset I \}$.
\end{defi}

\subsection{Motivation}\label{motivstepalgebra}

Let $\Mc$ be a finite dimensional $\g$-module, completely reducible as a $\k$-module.
Then the latter is fully characterized by the space of highest weight vectors $V_{\n_+}(\Mc)$.
We want to restore the $\g$-module $\Mc$ from the space $V_{\n_+}(\Mc):= \left\{ v \in \Mc : \n_+ . v = 0 \right\}$ equipped with some extra structure.
This was done by J. Mickelsson \cite{M} who defined the \emph{step algebra} of $\g$ with respect to $\k$. We use a simplified version of the step algebra 
which we call Mickelsson algebra. It acts on the space $V_{\n_+}(\Mc)$. This section is essentially based on \cite{M}, \cite{Hom1}, \cite{Hom2} and \cite{Zh3}.

\subsection{Definition of Mickelsson algebra}\label{ConstrstepAlgebras}

An element $x \in \U(\g)$ stabilizes $V_{\n_+}(\Mc)$ if $x.V_{\n_+}(\Mc) \subset V_{\n_+}(\Mc)$. Any element $x \in N(\U(\g)\n_+)$ stabilizes $V_{\n_+}(\Mc)$. This motivates the following definition.

\begin{defi}\label{defstepgebras}
The Mickelsson algebra $\mathcal{S}(\g,\k)$  is defined by
$$\mathcal{S}(\g,\k) := \quom{N(\U(\g)\n_+)}{\U(\g)\n_+} \ .$$
\end{defi}
In his paper \cite{M}, Mickelsson considered $\k$ to be semisimple. However, we see from Definition \ref{defstepgebras} that $\mathcal{S}(\g,\k) = \mathcal{S}(\g,\k')$ with $\k = \k' \oplus \mathfrak{a}$, $\mathfrak{a}$ the center of $\k$. 
\begin{prop*} [\cite{Hom2}] 
If $\Mc$ is an irreducible $\g$-module then $\mathcal{S}(\g,\k) v =V_{\n_+}(\Mc)$ for all $v \in V_{\n_+}(\Mc) \setminus \{0\}$.
\end{prop*}

\noindent
We are now interested in the generating sets of elements of $\mathcal{S}(\g,\k)$.

\vskip .1cm
We can see that the subalgebra $\U(\hf)$ of $\U(\g)$ is embedded in the quotient $\mathcal{S}(\g,\k)$. We will use the same notation $\U(\hf)$ for its image in $\mathcal{S}(\g,\k)$. 

\vskip .1cm
Let $\mathfrak{p}$ be an ad$_{\k}$-invariant complement to $\k$ in $\g$. Let $g_1,\dots,g_B$ be a ``weight" basis of $\mathfrak{p}$, that is, $[h,g_i]=\beta_i(h)g_i$ with $(h,\beta_i) \in \hf \times \hf^\star$, $i=1,\dots,B$. 
For any $i=1,\dots,B$ there exist elements $s_i\in N(\U(\g)\n_+)$ such that $s_i \equiv f_i g_i$ $\text{mod}\ \n_-\U(\g)$ with $f_i \in \U(\hf)$. The image of such $s_i$ in 
$\mathcal{S}(\g,\k)$, denoted by the same letter $s_i$, is called an \emph{elementary step}. Elementary steps generate the step algebra, which is a subalgebra 
of $\mathcal{S}(\g,\k)$, and verify the following property.
\begin{prop}\label{almostPBWpropstepalgebras}
For any element $x \in \mathcal{S}(\g,\k)$, there exists $f \in \U(\hf)$ such that $f x$ is a linear combination over $\U(\hf)$ of the monomials
$$s_1^{k_1} \dots s_B^{k_B} \quad , \quad (k_1,\dots,k_B) \in \mathZ^B \ .$$
\end{prop}
This property can be seen as a ``weak PBW theorem" in $\mathcal{S}(\g,\k)$. Therefore, one can ask if there exists a localization of $\mathcal{S}(\g,\k)$ satisfying the PBW property. The first intention is to pass to the field of fractions $\U'(\hf)$ of $\U(\hf)$.
Define the corresponding localizations of $\U(\k)$ and $\U(\g)$,
$$\U'(\k) = \U(\k) \otimes_{\U(\hf)} \U'(\hf) \quad \text{and} \quad \U'(\g) = \U(\g) \otimes_{\U(\hf)} \U'(\hf)\ .$$

\begin{defi}\label{defMickelssonalgebras}
The reduction algebra $\mathcal{Z}(\g,\k)$ is the following extension of the Mickelsson algebra $\mathcal{S}(\g,\k)$:
$$\mathcal{Z}(\g,\k) = \quom{N(\U'(\g)\n_+)}{\U'(\g)\n_+} \ .$$
\end{defi}

The reduction algebra $\mathcal{Z}(\g,\k)$ is actually a localization of its associated step algebra and of the algebra $\mathcal{S}(\g,\k)$.

\begin{prop}\label{Locasteptensprod} We have
\begin{equation}\label{factoMickalge}
\mathcal{Z}(\g,\k) = \S(\g,\k) \otimes_{\U(\hf)} \U'(\hf) \ .
\end{equation}
\end{prop}

By Proposition \ref{almostPBWpropstepalgebras}, we can see that the localisation \eqref{factoMickalge} implies that any reduction algebra $\mathcal{Z}(\g,\k)$ satisfies the PBW property (see \cite{Zh3} for details). 

\vskip .1cm
One can ask if there exist other localizations of $\mathcal{S}(\g,\k)$ which also have the PBW property. The answer is yes and there is a minimal localization, 
by a multiplicative subset of $\U(\hf)$ generated by the elements $(h_\alpha + k)$ with $k \in \mathbb{Z}$ and $\alpha \in \Delta$. The localization of the algebra $\mathcal{S}(\g,\k)$ by this multiplicative set is denoted $\mathcal{R}^\g_\k$. 
In the sequel we mainly work with this localization.

\section{Endomorphism spaces}\label{EndoVermaMod}
\subsection{Weyl algebra $\text{W}_n$}\label{WeylAlgebraSec}
Let $\text{W}_n$ be the Weyl algebra of rank $n$, the algebra with generators $X^j,D_j$, $j=1,\dots,n$ and the defining relations
$$X^iX^j=X^jX^i\ ,\ D_iD_j=D_jD_i\ ,\ D_iX^j=\delta_i^j+X^jD_i\ ,\ \ i,j=1,\dots,n\ .$$
The Weyl algebra satisfies the PBW property. 
The algebra $\W_n$ is the algebra of polynomial differential operators on the ring $\K[X^1,\dots,X^n]$ and we have the natural isomorphism
$M:=\W_n/I\simeq\K[X^1,\dots,X^n]$ of $W_n$-modules. Here $I$ is the left ideal generated by the elements  $D_1,\dots,D_n$. The algebra $\text{End}(M)$ can be described as $\F_n = \K[X^1,\dots,X^n][[D_1,\dots,D_n]]$, the algebra of formal power series of the form
\begin{equation}\label{formendWeyl}
\sum_{k \in \mathZ^n} f_k D_1^{k_1} \dots D_n^{k_n}\ ,\ f_k \in \K[X^1,\dots,X^n]\ \ \text{for all} \ k=(k_1,\dots,k_n) \in \mathZ^n \ .
\end{equation}
Set $\text{deg}\left( (X^1)^{k_1} \dots (X^n)^{k_n}\right) 
:=k_1+\dots+k_n$. Let $M=\oplus_{d\in\mathZ} M_d$ be the corresponding decomposition of $M$ into the direct sum of homogeneous components. 
An endomorphism $\psi \in \text{End}(M)$ is said to be bounded if there exists $m \in \mathZ$ such that $\psi(p)\in \oplus_{j=d-m}^{d+m}M_j$
for all $p \in M_d$, $d\in\mathZ$. Let $\tilde{\F}_n$ denote the subalgebra of $\F_n$ consisting of all bounded endomorphisms.
For $d \in \mathbb{Z}$, let $\Ft_{n,d}$ be the subspace of $\F_n$ of elements of the form \eqref{formendWeyl} such that $\text{deg}(f_k) = d + \left(\sum_i k_i \right)$.
Then  
\begin{equation}\label{decompoF11}
\Ft_n = \bigoplus_{d \in \mathbb{Z}} \Ft_{n,d} \ .
\end{equation}
\noindent{\bf Vacuum projector.}
For any $i \in \{1,\dots,n\}$, let 
$$P_i = \sum_{k \in \mathZ} \frac{(-1)^k}{k!} (X^i)^k (D_i)^k \in \Ft_{n,0}\ .$$
We have $P_i X^i = D_i P_i = 0$. Therefore $P X^i = D_i P = 0$, $i=1,\dots,n$, for the product $P = \prod_{i=1}^n P_i$. The endomorphism $P$ projects in $\K[X^1,\dots,X^n]$ on $\K$ along the ideal generated by $X^1,\dots,X^n$. The element $P$ is called the vacuum projector of the Weyl algebra $\text{W}_n$ (see 14.9 of \cite{Zh2}).

\subsection{Endomorphism space of universal Verma module}\label{EndRedLieAlg}
In this section we describe a similar construction for a reductive Lie algebra $\k$. 

Let $\Ubh$ be the localization of $\U(\hf)$ with respect to the multiplicative set $S$ generated by the elements of the form
\begin{equation}\label{formeleloca}
h_\alpha + k \quad , \quad \alpha \in \Delta^+ \ , \ k \in \mathbb{Z}\ .
\end{equation}
Let $\Ub(\k)$ be the corresponding localization of $\U(\k)$ :
\begin{equation}\label{factolocalk11}
\Ub(\k) = \U(\k) \otimes_{\U(\hf)} \Ubh \ .
\end{equation}
The localization $\Mb_{\n_+}$ of the universal Verma module $\M_{\n_+}$ is defined by
$$\Mb_{\n_+} := \quom{\Ub(\k)}{\Ub(\k)\n_+} \ .$$
The PBW theorem implies the folowing isomorphism of $\Ub(\k)$-modules:
\begin{equation}\label{isoVermamod}
\Mb_{\n_+} \simeq \Ub(\b_-) := \U(\b_-) \otimes_{\U(\hf)} \Ubh \ .
\end{equation}
The ring $\Ub(\k)$ is naturally a $\Ubh$-bimodule. Let  $e_+(r)= e_{\alpha_1}^{r_1} \dots e_{\alpha_L}^{r_L}$ and $e_-(r) = \epsilon(e_+(r))$, $r=(r_1,\dots,r_L) \in \mathbb{Z}_{\geq 0}^L$. Here $\Delta^+ = \{\alpha_1,\dots,\alpha_L\}$ is the set of positive roots. The PBW theorem 
implies that the elements
\begin{equation}\label{defgenek12}
e_-(r) e_+(s)\ ,\  r,s \in \mathbb{Z}_{\geq 0}^L
\end{equation}
form a basis of $\Ub(\k)$ considered as a one-sided $\Ubh$-module. 

The analogue $\Fn$ of the algebra $\F_n$ is formed by elements $\sum_{s \in \mathZ^L} f_s e_+(s)$ with $f_s\in \Ub(\b_-)$.
For any $d \in \mathbb{Z}$, let $\Ft_{\n_+,d}$ be the subspace of $\Fn$ of elements of the form 
\begin{equation}\label{lineardecompototalweightelem}
\sum f_{r,s} e_-(r)e_+(s) \ ,
\end{equation}
with $\sum_i (s_i-r_i) = d$. The subalgebra $\Ftn \subset \Fn$, analogous to $\Ft_n$, is
$$\Ftn = \bigoplus_{d \in \mathbb{Z}} \Ft_{\n_+,d} \ .$$
The algebra $\Ftn$ is called the Taylor extension of $\Ub(\k)$. Let $\Ftn^0 \subset \Ftn$ be the subalgebra of elements of the form \eqref{lineardecompototalweightelem} with $r=s$. Let $\E_{\n_+}^0 \subset \text{End}(\Mb_{\n_+})$ be the subalgebra of endomorphisms $\psi$ such that $\psi(ha) = h \psi(a)$, $(h,a) \in \hf \times \Mb_{\n_+}$. Similarly to the proof of Theorem 1 of \cite{Zh3}, one can prove the isomorphism of algebras
\begin{equation}\label{isomFtn0E0}
\Ftn^0 \simeq \E_{\n_+}^0 \ .
\end{equation}
Theorem 1 of \cite{Zh3} is proved 
over the field of fractions of $\U(\hf)$. Going through the proof of Theorem 1 of \cite{Zh3}, one can see that it is not necessary. The elements of $\U(\hf)$ which have to be inverted (to prove the theorem) are the coefficients $\mathbf{s}(e_-(r),e_-(s))$ for $r,s \in \mathZ^ n$ where $\mathbf{s}$ is the Shapovalov form fo $\k$. We know from part 2 of \cite{Sha} that these coefficients are products of elements of the form \eqref{formeleloca}.

\section{Extremal projector}\label{ExtremalProj}
The extremal projectors of reductive Lie algebras are widely used in the representation theory of Lie algebras, see \cite{T} for a survey. In this section we give the definition of the extremal projector, verify its existence, uniqueness and write down a certain factorization.
\subsection{Definition}
\begin{theo}\label{theoextremalprojector}
There exists a unique element $P \in \Ftn$, called the extremal projector of $\k$, which satisfies
$$P^2 = P \quad , \quad \epsilon(P) = P\ ,$$
where $\epsilon$ is the Chevalley anti-involution \eqref{formchevalleyantiinv}, and
\begin{equation}\label{extrprojdefequa1}
P \ \n_- = \n_+ \, P = 0\ ,
\end{equation}
\begin{equation}\label{extrprojdefequa2}
P \, \equiv \, 1 \ \text{mod} \ \Ftn \n_+ \quad \text{and} \quad P \, \equiv \, 1 \ \text{mod} \ \n_- \Ftn \ .
\end{equation}
\end{theo}
\begin{proof}
This is due to the isomorphism \eqref{isomFtn0E0}. Let $P$ be the projector of $\bar{\M}_{\n_+}$ on $\Ub(\hf)$ along the $\Ub(\hf)$-submodule $\n_- \Ub(\k)$. So $P\in \E_{\n_+}^0=\Ftn^0$, $P^2=P$ and $P \n_- = 0$. Since for any $(e_+,a) \in \n_+ \times \Mb_{\n_+}$ we have $e_+ P a = 0$, it implies that $e_+ P = 0$ and so $\n_+ P = 0$. The first property of \eqref{extrprojdefequa2} is checked by projection of $1$ in $\bar{M}_{\n_+}$. The element $\epsilon(P) \in \Ftn$ satisfies \eqref{extrprojdefequa1} and the second property of \eqref{extrprojdefequa2} so $\epsilon(P) = P \epsilon(P) = P$. This implies that $P$ satisfies the second property of \eqref{extrprojdefequa2} and of course that $\epsilon(P)=P$. Clearly, $P$ is uniquely defined.
\end{proof}

\subsection{$\sl_2$-triple}\label{ExtremalProjsl2}

\begin{defi}\label{extremalprojsl2}
For any $\alpha \in \Delta^+$ and $t \in \mathbb{Z}$, we define the element
\begin{equation}\label{equaextremalprojsl2}
P_\alpha(t) = \sum_{k=0}^{\infty} \frac{(-1)^k}{k!} \frac{1}{f_{\alpha,k}(t)} e_{-\alpha}^k e_\alpha^k
\end{equation}
of $\Ftn$ with $f_{\alpha,0}(t) = 1$ and $f_{\alpha,k}(t) = \prod_{l=1}^k (h_\alpha + t + l)$ for all $k \in \N$.
\end{defi}
One can verify that
$e_\alpha P_\alpha(t) = \frac{t-1}{h_\alpha+t-1} P_\alpha(t-1) e_\alpha$.
Clearly, $\epsilon (P_\alpha)=P_\alpha$, where $\epsilon$ is the Chevalley anti-involution \eqref{formchevalleyantiinv}, and and $P_\alpha\in\Ftn^0$, so
$ P_\alpha(t) e_{-\alpha} = \frac{t-1}{h_\alpha+t+1} e_{-\alpha} P_\alpha(t-1) $.
It implies that $e_\alpha P_\alpha(1) = P_\alpha(1) e_{-\alpha} = 0$, $P_\alpha(1) \, \equiv \, 1 \ \text{mod} \ \Ftn e_\alpha$ and $P_\alpha(1) \, \equiv \, 1 \ \text{mod} \ e_{-\alpha} \Ftn$. The uniqueness implies the following proposition.

\begin{prop}\label{propextremalprojsl2}
For any $\alpha \in \Delta^+$, the element $P_\alpha(1) \in \Ftn$ defined in \eqref{extremalprojsl2} is the extremal projector of the subalgebra of $\k$ generated by  
$\{e_{-\alpha},h_\alpha,e_\alpha\}$.
\end{prop}

\subsection{Factorization} 

Let $\rho$ be the half sum of the positive roots of $\k$ : $\rho = \displaystyle \frac{1}{2} \sum_{\alpha \in \Delta^+} \alpha$.

\begin{prop}\label{factoextrprojroots}
\cite{AST} The extremal projector $P$ of $\k$ admits the factorization in the ordered product:
\begin{equation}\label{factorizationofP}P = \prod_{i=1}^L P_{\alpha_i}(\rho(h_\alpha)) \ ,\end{equation}
for any choice of convex order $\{\alpha_1,\dots,\alpha_L\}$ on $\Delta^+$. 
\end{prop}

\section{Uses of extremal projector}\label{UsesExtreProj}

The extremal projector is extensively used in the theory of reduction algebras. It allows to give alternative definitions of the Mickelsson algebras \ref{defMickelssonalgebras} and to describe them in terms of generators and relations.

Assume that $\k$ is a Lie subalgebra of  a finite-dimensional Lie algebra $\g$ and the adjoint action of $\k$ is semisimple on $\g$. 
The algebra $\Ub(\k)$ is a free left $\Ub(\hf)$-module with a $\Ub(\hf)$-basis given by the monomials \eqref{defgenek12}.
Let $\mathfrak{p}$ be as in \textbf{Section \ref{ConstrstepAlgebras}}.
Let $\Ub(\g)$ be the localization of $\U(\k)$ with respect to the set of denominators \eqref{formeleloca}.
The algebra $\Ub(\g)$ is a free left $\Ub(\k)$-module with a $\Ub(\k)$-basis
\begin{equation}\label{Ukbasisofg}
g_1^{k_1} \dots g_B^{k_B} \quad , \quad (k_1,\dots,k_B) \in \mathZ^B\ .
\end{equation}
\subsection{Double coset algebra}\label{AlternativeDefReducAlge}
\begin{defi}\label{defireducalgedoubcos}
Let $\mathcal{R}^\g_\k$ be the double coset space $\qquom{\n_-\Ub(\g)}{\Ub(\g)}{\Ub(\g)\n_+}$
equipped with the product $\tilde{a} \dia \tilde{b} = aPb\ \text{mod}\ \n_-\Ub(\g)+\Ub(\g)\n_+$ for representatives $a,b$ of the 
cosets $\tilde{a},\tilde{b}\in \mathcal{R}^\g_\k$.
\end{defi}
\vspace{-.2cm}
We give without proof several remarks, see \cite{Zh1,Zh3,Kh,KO1}. The product $\dia$ is well defined due to the local nilpotency of the adjoint action 
of $\n_+$ on $\U(\g)$. In the sequel we often denote the image of an element $a$ in the double coset algebra by the same letter $a$.  
The algebra $\mathcal{R}^\g_\k$ is isomorphic to the localized algebra $\mathcal{S}(\g,\k) \otimes_{\U(\hf)} \bar{U}(\hf)$. The reduction algebra $\mathcal{R}^\g_\k$ is generated over $\Ub(\hf)$ by the images of the elements $g_1,\dots,g_B$. These generators satisfy ordering relations $$g_i \dia g_j = \sum_{k<l} f_{klij} \, g_k \dia g_l + \sum_{k} f_{kij} \, g_k + f_{ij} \ , \ i>j\ ,\ f_{klij}, f_{kij}, f_{ij} \in \Ub(\hf)\ .$$
 
The double coset algebra $\mathcal{R}^{\A}_\k$ can be defined in a more general situation, for an associative algebra $\A$ containing $\U(\k)$. The algebra $\A$ must satisfy some conditions with respect to $\k$: the so called $\k$-admissibility of $\A$ and the local highest weight condition.

\subsection{Zhelobenko automorphisms}\label{ZhelobenkoAuto}
The algebra $\mathcal{R}^\A_\k$ admits the action of Zhelobenko automorphisms. For details see \cite{Zh3}, \cite{KO1}.

\begin{defi}\label{Zhelobenkoautoms}
For any $\alpha_i \in \Pi$ and $x \in A$ let 
$$\operatorname{q}_i(x) = \sum_{k \geq 0} \frac{(-1)^k}{k!} \text{ad}_{e_{\alpha_i}}^k(x) e_{-{\alpha_i}} \frac{1}{f_{\alpha_i,k}} \ ,\ \ \text{where}\ \ f_{\alpha_i,k} = \displaystyle \prod_{l=1}^k (h_{\alpha_i} - l + 1) \ ,$$
and, see \eqref{liftweylgroupgene},
$\q_i = \operatorname{q}_i \circ \ T_i$.
\end{defi}
\noindent The operators $\q_i$ descend to automorphisms of $\mathcal{R}^\A_\k$ and satisfy the braid relations, see \eqref{braidrela},
\begin{equation}\label{propZhelobenkoautoms}
\underset{m_{ij}}{\underbrace{\q_i \q_j \q_i \dots}} = \underset{m_{ij}}{\underbrace{\q_j \q_i \q_j \dots}} \quad , \quad i \neq j \ .
\end{equation}

\section{Examples of reduction algebras}\label{ExampleRedAlge}
This section illustrates the machinery of reduction algebras. 

\vskip .1cm
Let $e_{ij}$, $i,j=1,\dots,n$ be the standard generators of the Lie algebra $\gl_n=\gl_n(\mathbb{C})$, 
$$[e_{ij},e_{kl}] = \delta_{kj} e_{il} - \delta_{il} e_{kj} \quad , \quad  i,j,k,l=1,\dots,n \ .$$

In the first example, we construct ``by hand" the reduction algebra of $\gl_3$ with respect to $\gl_2$, working 
explicitly with the extremal projector. The second example is the algebra of $\hf$-deformed differential operators $\Diff(n,N)$, 
the reduction algebra of the tensor product $\text{Diff}(n,N)\otimes \U(\gl_n)$ with respect to the diagonal embedding of $\U(\gl_n)$. Here $\text{Diff}(n,N)$ is the algebra of polynomial differential operators in $nN$ variables. 

\subsection{Reduction algebra of $\gl_3$ with respect to $\gl_2$}\label{FirstExempleRedAlg}
We use the following notation (its meaning will be explained in \textbf{Section \ref{GeneralizedDiff}}) for generators of $\gl_3$
$$x^1 = e_{13} \ , \ x^2 = e_{23} \ , \ e = e_{12} \ , \ \de_1 = e_{31} \ , \ \de_2 = e_{32} \ , \ f = e_{21}\ ,\ h_1 = e_{11} \ , \ h_2 = e_{22} \ , \ h_3 = e_{33}\ .$$
 We denote also: $h_{ij} = h_i-h_j$, $i,j=1,2,3$, $\n_- = \mathbb{C}f$, $\hf = \mathbb{C} h_{12}$ and $\n_+ = \mathbb{C}e$. We describe the reduction algebra 
 $\mathcal{R}^{\gl_3}_{\gl_2}$ with $\gl_2$ generated by $e,f$ and $h=h_{12}$. The algebra $\mathcal{R}^{\gl_3}_{\gl_2}$ is generated by $h_{3},x^1,x^2,\de_1,\de_2$ over $\Ub(\hf)$. The symbol $:\!a\!:$ stands for the image of $a \in \Ub(\gl_3)$ in the quotient by $\Ub(\gl_3) \n_+ + \n_- \Ub(\gl_3)$. The products $a \dia b = a P b$, $(a,b) \in \{x^1,x^2,\de_1,\de_2\}^2$ are given by
\begin{equation}\label{reducrela1}
\begin{array}{c}
\displaystyle{ x^2 \dia x^1 = : x^1 x^2 : \quad , \quad x^1 \dia x^2 = \frac{h+2}{h+1} :x^1 x^2:\ ,}\\[.5em]
\displaystyle{ \de_2 \dia \de_1 = \frac{h+2}{h+1} : \de_1 \de_2 : \quad , \quad \de_1 \dia \de_2 = :\de_1 \de_2:\ ,}\\[1em]
\displaystyle{ x^2 \dia \de_1 = : \de_1 x^2 : = \de_1 \dia x^2 \quad , \quad x^1 \dia \de_2 = :\de_2 x^1: = \de_2 \dia x^1\ ,}\\[.5em]
\displaystyle{ x^2 \dia \de_2 = : \de_2 x^2 : + h_{23} \quad , \quad \de_2 \dia x^2 = :\de_2 x^2: - \frac{1}{h+1} :\de_1 x^1:\ ,}\\[.5em]
\displaystyle{ x^1 \dia \de_1 = : \de_1 x^1 : - \frac{1}{h+1} :\de_2 x^2: + \frac{h}{h+1}(h_{13}+1) \quad , \quad \de_1 \dia x^1 = :\de_1 x^1: \ .}
\end{array}\end{equation}
For instance, we compute the second relation in the first line of \eqref{reducrela1}:
\begin{align*}
x^1 \dia x^2 & = :x^1 P x^2: = :x^1 \left(1-\frac{1}{h+2}fe+\dots\right) x^2: = :x^1 x^2: - \frac{1}{h+1} :x^1fex^2: \\[-.1em]
& = :x^1 x^2: - \frac{1}{h+1} :[x^1,f][e,x^2]: = :x^1 x^2: + \frac{1}{h+1} :x^2x^1: = \frac{h+2}{h+1} :x^1x^2: \ .
\end{align*}
The relations in the three first lines of  \eqref{reducrela1} lead to
$$
x^1 \dia x^2 = \frac{h+1}{h+2} x^2 \dia x^1 \quad , \quad \de_1 \dia \de_2 = \frac{h+2}{h+1} \de_2 \dia \de_1 \quad , \quad x^1 \dia \de_2 = \de_2 \dia x^1 \quad , \quad x^2 \dia \de_1 = \de_1 \dia x^2 \ .
$$
The relations two last lines of \eqref{reducrela1} form a linear system. It follows that
$$
\begin{array}{c}
\displaystyle{x^1 \dia \de_1 = \displaystyle \frac{h(h+2)}{(h+1)^2} \de_1 \dia x^1 - \frac{1}{h+1} \de_2\dia x^2 + \frac{h}{h+1}(h_{13}+1)\ ,}\\[.4em]
\displaystyle{x^2 \dia \de_2 = \displaystyle \frac{1}{h+1} \de_1 \dia x^1 + \de_2 \dia x^2 + h_{23} \ .}
\end{array}$$
The remaining relations are directly inherited from the weight relations in $\U(\gl_3)$:
$$\begin{array}{c}
\displaystyle{x^i \dia h_j = (h_j - \delta_{ij}) \dia x^i \ , \ \de_i \dia h_j = (h_j + \delta_{ij}) \dia \de_i \ , \ h_i \dia h_j = h_j \dia h_i \ , \ i,j=1,2,}\\[.4em]
\displaystyle{h_3\dia x^i=x^i\dia (h_3-1)\ ,\ h_3\dia \de_i=x^i\dia (h_3+1)\ ,\ h_3\dia h_i=h_i\dia h_3\ , \ i=1,2.}
\end{array}$$

\subsection{Reduction algebra $\Diff(n,N)$}\label{ReducAlgeDiff}
Let $\hf$ be the Cartan subalgebra of $\gl_n$ generated by the elements $h_i = e_{ii}$, $i=1,\dots,n$. We denote $\th_i := h_i - i$ and $\th_{ij} := \th_i-\th_j$, $i,j=1,\dots,n$. We write $\U(n)$ and $\Un$ instead of $\U(\hf)$ and $\U(\hf)$, see \textbf{Section \ref{EndRedLieAlg}}. 
Let $\text{Diff}(n,N)$ be the algebra of polynomial differential operators in $nN$ variables. It is generated by elements $\Z^{i\alpha}$ and $\de_{i\alpha}$, $i=1\dots,n$, $\alpha=1,\dots,N$ with relations
\begin{equation}\label{commutrelaDiff11}
[\Z^{i\alpha},\Z^{j\beta}] = [\de_{i\alpha},\de_{j\beta}] = 0 \quad \text{and} \quad [\de_{i\alpha},\Z^{j\beta}] = \delta^j_i \delta^\beta_\alpha \ .
\end{equation}
We describe the reduction algebra $\Diff(n,N) := \mathcal{R}^{\A}_{\gl_n}$ of the algebra 
$\A :=\text{Diff}(n,N)\otimes \U(\gl_n)$ with respect to the diagonally embedded $\U(\gl_n)\to \text{Diff}(n,N)\otimes \U(\gl_n)$, 
\begin{equation}\label{embeddingUgln}
e_{ij}\mapsto \psi'(e_{ij}):=\psi(e_{ij})\otimes 1+1\otimes e_{ij}\ ,\ \ \text{where}\ \  \psi(e_{ij}):= \sum_{\alpha=1}^N \Z^{i\alpha} \de_{j\alpha} \quad , \quad i,j=1,\dots,n \ .
\end{equation}
We consider the Cartan decomposition $\gl_n = \n_- \oplus \hf \oplus \n_+$ with
$$\n_- = \sum_{1 \leq i < j \leq n} \mathbb{C} e_{ji} \quad \text{and} \quad \n_+ = \sum_{1 \leq i < j \leq n} \mathbb{C} e_{ij} \ .$$
Over $\Un$, the generators of $\Diff(n,N)$ are
$\Z^{i\alpha} \otimes 1$ and $\de_{i\alpha} \otimes 1$,
which we denote by $\Z^{i\alpha}$ and $\de_{i\alpha}$ as well. In the sequel we sometimes omit the symbol 
$\dia$ for the product in $\mathcal{R}^{\A}_{\gl_n}$.

\noindent{\bf Zhelobenko automorphisms.} The action of the Zhelobenko operators $\q_i$, $i=1,\dots,n-1$ (see Definition \ref{Zhelobenkoautoms}) is given by 
(see \cite{KO5})
\begin{equation}\label{automq}
\begin{split}
\q_i(\Z^{i\alpha}) &= - \Z^{i+1,\alpha} \frac{\th_{i,i+1}}{\th_{i,i+1}-1} \quad , \quad \q_i(\Z^{i+1,\alpha}) = \Z^{i\alpha} \quad , \quad \q_i(\Z^{j\alpha}) = \Z^{j\alpha} \ , \ j \neq i,i+1\ , \\
\q_i(\de_{i\alpha}) &= \de_{i+1,\alpha} \frac{\th_{i,i+1}}{\th_{i,i+1}-1} \quad , \quad \q_i(\de_{i+1,\alpha}) = -\de_{i\alpha} \quad , \quad \q_i(\de_{j\alpha}) = \de_{j\alpha} \ , \ j \neq i,i+1\ , \\
\q_i(\th_j)&=\th_{s_i(j)}\ .
\end{split}
\end{equation}
\noindent{\bf Chevalley anti-involution.} We extend the Chevalley anti-involution of $\U(\gl_n)$
to $\A$ by
\begin{equation}\label{invoantiautoDiff11}
\epsilon(\Z^{i\alpha}) = \de_{i\alpha} \quad , \quad \epsilon(\de_{i\alpha}) = \Z^{i\alpha} \ . \end{equation}
Since $\epsilon(P) = P$, the anti-involution $\epsilon$ induces an anti-involution of $\Diff(n,N)$.

\vskip .2cm
\noindent{\bf Defining relations.}
In $\text{Diff}(n,N)$ we have the weight relations
\begin{equation}\label{weightrelainit1}
\Z^{j\alpha} \th_i = (\th_i-\delta^j_i) \Z^{j\alpha} \quad , \quad \de_{j\alpha} \th_i = (\th_i+\delta_{ji}) \de_{j\alpha} \quad , \quad i,j=1,\dots,n \ .
\end{equation}

\noindent{\bf Relations between $\Z$'s and between $\de$'s.} The element $\Z^{1\alpha}$ is a highest weight vector, 
\begin{equation}\label{highestweightx11}
[\psi'(e_{rs}),\Z^{1\alpha}] = 0 \quad , \quad 1 \leq r < s \leq n \ .
\end{equation}
It follows that $\Z^{1\alpha} \Z^{1\beta} = :\Z^{1\beta} \Z^{1\alpha}: = \Z^{1\beta} \Z^{1\alpha}$. Applying the automorphisms $\q_i$, we obtain $\Z^{i\alpha} \Z^{i\beta} = :\Z^{i\beta} \Z^{i\alpha}: = \Z^{i\beta} \Z^{i\alpha}$.
Due to \eqref{factorizationofP}, we have $\Z^{1\alpha} \Z^{2\beta} = : \Z^{1\alpha} P \Z^{2\beta}: =: \Z^{1\alpha} P_{12} \Z^{2\beta}:$ where $P_{12}$ is the extremal projector of the $\sl_2$-subalgebra $\{e_{12},e_{21},h_{12}\}$. As in \textbf{Section \ref{FirstExempleRedAlg}}, we find
$$\Z^{2\alpha} \Z^{1\beta} = \Z^{1\beta} \Z^{2\alpha} - \frac{1}{\th_{12}} \Z^{2\beta} \Z^{1\alpha}\ .$$
Applying the Zhelobenko automorphisms, 
we get
\begin{equation}\label{relaxxini0}
\Z^{j\alpha} \Z^{i\beta} = \Z^{i\beta} \Z^{j\alpha} - \frac{1}{\th_{ij}} \Z^{j\beta} \Z^{i\alpha}  \ , \ 1 \leq i < j \leq n \ .
\end{equation}
Using the equation \eqref{relaxxini0} and its image under the permutation of $\alpha$ and $\beta$, we find
\begin{equation}\label{relaxxini0b}\Z^{i\alpha} \Z^{j\beta} = \frac{\th_{ij}^2-1}{\th_{ij}^2} \Z^{j\beta} \Z^{i\alpha} + \frac{1}{\th_{ij}} \Z^{i\beta} \Z^{j\alpha}  \quad , \quad 1 \leq i < j \leq n\ .\end{equation}
We rewrite \eqref{relaxxini0} and \eqref{relaxxini0b} in the matrix form 
\begin{equation}\label{relaxxini}
\Z^{i\alpha}\Z^{j\beta}=\sum_{k,l}\RR^{ij}_{kl} \Z^{k\beta} \Z^{l\alpha}\ . 
\end{equation}
The meaning of the matrix $\RR$ will be explained in Section \ref{ChapterRinghdef}. 
Applying $\epsilon$ to \eqref{relaxxini}, we get
\begin{equation}\label{reladdinit1}
\de_{j\alpha} \de_{i\beta} = \sum_{k,l} \de_{l\beta} \de_{k\alpha} \RR^{ij}_{kl} \  .
\end{equation}
\noindent{\bf Relations between $\Z$'s and $\de$'s.} Due to the structure of the extremal projector $P$, 
$\Z^{i\alpha}\dia\de_{j\beta}=A^{ik}_{jl}:\Z^{l\alpha}\de_{k\beta}:$, where $A^{ik}_{jl}\in \Ub(n)$, $A^{ij}_{ji}=1$ and $A^{ik}_{jl}=0$ if $l<i$ or $j>k$. So 
the transition matrix from $:\Z^{l\alpha}\de_{k\beta}:$'s to $\Z^{i\alpha}\dia\de_{j\beta}$'s (similarly, from $:\de_{k\alpha}\Z^{l\beta}:$'s to $\de_{j\alpha}\dia\Z^{i\beta}$'s) is triangular with ones on the diagonal. Hence
the defining relations between the generators $\Z$ and $\de$ have the form
\begin{equation}\label{relaxndn3}
\Z^{i\alpha} \de_{j\beta} = \sum_{k,l} \hat{S}^{ki}_{lj} \de_{k\beta} \Z^{l\alpha} - \delta^i_j \delta^\alpha_\beta \sigma_{i\alpha} \ ,\ \ \hat{S}^{ki}_{lj}, \sigma_{i\alpha} \ \in \Ub(\hf) \ . 
\end{equation}
The constant term in the right hand side is due to the structure of relations \eqref{commutrelaDiff11}.

The ring $\Diff(n,N)$ has the PBW property over $\Ub(\hf)$ with respect to the set of generators  $\Z^{l\alpha},\de_{k\beta}$, see \cite{KO4}.
The PBW property is sufficient for the determination of $\hat{S}$ if $N>1$. 
Indeed, reordering the monomial $\Z^{i\alpha}\Z^{j\beta}\de_{k\beta}$, $\alpha\neq\beta$, in two ways, we get the system 
\begin{equation}\label{equamatrixRS}
\sum_{a,b}\hat{R}^{ij}_{ab}\hat{S}^{ab}_{ck}\sigma_{a\beta}=\delta^i_c\delta^j_k\sigma_{j\beta}\vert_{_{h_l\to h_l-\delta^l_i,l=1,\dots,n}}
\ ,\ i,j,c,k=1,\dots,n\ .
\end{equation}
This system uniquely defines $\hat{S}$ if $\hat{R}$ is invertible and $\sigma_{a\beta}$ are non-zero. However, for $N=1$, one cannot determine $\hat{S}$
using just the PBW property($\hat{R}$ also cannot be determined for $N=1$). The explicit calculation with the projector shows that for $N=1$ the system \eqref{equamatrixRS} is still valid.

\vskip .2cm
\noindent{\bf Conclusion.} 
The ring $\Diff(n,N)$ is generated over $\Un$ by the elements $\Z^{i\alpha}, \der_{i\alpha}$ with the defining commutation relations \eqref{weightrelainit1},  
 \eqref{relaxxini}, \eqref{reladdinit1} and \eqref{relaxndn3}.

\chapter{$\h$-deformed differential operators}\label{ChapterRinghdef}

\vskip -1cm
The ring $\Diff(n,N)$, discussed in \textbf{Section \ref{ExampleRedAlge}}, is formed by $N$ copies of the ring $\Diff(n)=\Diff(n,1)$. In this chapter we investigate the structure of the ring $\Diff(n)$. We recall its definition in \textbf{Section \ref{hdiff-sect}}.
Our first result is the description of the center of $\Diff(n)$: it is a ring of polynomials in $n$ generators. 

\vskip .1cm
The ring $\Diff(n)$ is a Noetherian Ore domain, see \cite{KO4}. It is natural to investigate its rings of fractions
and test the validity of the Gelfand--Kirillov-like conjecture \cite{GK} as it is done for the ring of $q$-differential operators in \cite{O}. 
The second result of this chapter consists in a construction of an isomorphism 
between certain localizations of $\Diff (n)$ and the Weyl algebra $\text{W}_n$ extended by $n$ indeterminates. 

\paragraph{Notation.}
Let $\h(n)$ be the abelian Lie algebra with generators $\th_i$, $i=1,\dots,n$, and  $\U(n)$ its universal enveloping algebra. Set $\th_{ij}=\th_i-\th_j\in \h(n)$. 
We define $\Ub(n)$ to be the ring of fractions of the commutative ring 
$\U(n)$ with respect to  the multiplicative set of denominators, generated by the elements
$(\th_{ij}+k)^{-1}$, $k \in \mathbb{Z}$, $i,j=1,\dots,n$, $i\neq j$.
Let
\begin{equation}\label{defphichi}
\psi_i:=\prod_{k:k>i}\th_{ik}\ , \psi_i':=\prod_{k:k<i}\th_{ik}\ \ \text{and}\ \ \chi_i:=\psi_i\psi_i'\ ,\ i=1,\dots,n\ .
\end{equation}
Let $\varepsilon_j$, $j=1,\dots,n$, be the elementary translations of the generators of $\U(n)$, 
$\varepsilon_j:\th_i\mapsto\th_i+\delta_i^j$. For an element $p\in \Ub(n)$ we denote  
$\varepsilon_j(p)$ by $p[\varepsilon_j]$. We shall use the finite difference operators $\Delta_j$ defined by
$$\Delta_j f:=f-f[-\varepsilon_j]\ .$$ 
We have 
$\Delta_i (fg) = f \Delta_i(g) + \Delta_i(f) g - \Delta_i(f) \Delta_i(g)$.

We denote by $e_k$, $k=0,\dots,n$, the elementary symmetric polynomials in the variables $\th_1,\dots,\th_n$,
and by $e(t)$ the generating function of the polynomials $e_k$,
\begin{equation}\label{eah}
e_k = \sum_{i_1<\dots<i_k} \th_{i_1} \dots \th_{i_k}\ ,\ 
e(t) = \sum_{k=0}^n e_k \, t^k = \prod_{i=1}^n\, (1+\th_i t) \ .
\end{equation}
We denote by $\RR\in\text{End}_{\Ub(n)}\left( \Ub(n)^n\otimes_{\Ub(n)}\Ub(n)^n\right)$ the standard solution of the dynamical Yang--Baxter equation
\begin{equation}
\sum_{a,b,u} {\RR}^{ij}_{ab} \RR^{bk}_{ur}[-\varepsilon_a]\RR^{au}_{mn} = \sum_{a,b,u} \RR^{jk}_{ab}[-\varepsilon_i]\RR^{ia}_{mu} \RR^{ub}_{nr}[-\varepsilon_m] 
\end{equation}
of type A. The non-zero components of the operator $\RR$ are
\begin{equation}\label{dynRcompb}
\RR_{ij}^{ij}=\frac{1}{\th_{ij}}\ ,\ \ i\not=j\ ,\qquad\text{and}\qquad \RR_{ji}^{ij}=\left\{
\begin{array}{cc}
 \dfrac{\th_{ij}^2-1}{\th_{ij}^2}\ ,& \, i<j,\\[0.5em]
 1\ ,&\, i\geq j\ .
\end{array}
\right.
\end{equation}
We shall need the following properties of $\RR$ : 
\begin{eqnarray}
\label{weze}
&\RR^{ij}_{kl}[\varepsilon_i+\varepsilon_j]=\RR^{ij}_{kl} \ ,\ i,j,k,l=1,\dots,n\ . &\\
\label{iceco}
&\RR^{ij}_{kl}=0\ \text{if}\ (i,j)\neq (k,l)\ \text{or}\ (l,k) \ .&\\
\label{Q5do}
&\RR^2=\text{Id} \ .&
\end{eqnarray}

We denote by $\PPsi\in\text{End}_{\Ub(n)}\left( \Ub(n)^n\otimes_{\Ub(n)}\Ub(n)^n\right)$ the dynamical version of the skew inverse of  $\RR$
(see e.g. \cite{O2}, section 4.1.2 for details of the $R$-matrix technique needed here),
defined by  
\begin{equation}\label{skewn}
\sum_{k,l}\,\PPsi_{jl}^{ik}\,\RR_{nk}^{ml}[\varepsilon_m]=\delta_{n}^i\delta_j^m.
\end{equation}
The non-zero components of the operator $\PPsi$ are, see \cite{KO5},
\begin{equation}\label{explpsi}
\PPsi^{ij}_{ij} = \QQ^+_i \QQ^-_j \frac{1}{\th_{ij}+1} \ , \qquad \PPsi^{ij}_{ji} = \left\{\begin{array}{cc}
1&,\ i<j\\[.5em]
\displaystyle{\frac{(\th_{ij}-1)^2}{\th_{ij}(\th_{ij}-2)} }&,\ i>j\end{array}\right.
\end{equation}
where
\begin{equation}\label{defqpm}
\QQ^{\pm}_i=\frac{\chi_i[\pm\varepsilon_i]}{\chi_i} \ .
\end{equation}

\section{Definition and properties of rings of $\h$-deformed differential operators}\label{hdiff-sect}

The ring $\Diff (n)$ of $\h$-deformed differential operators of type A is a $\Un$-bimodule with the generators $\Z^j$ and $\der_j$, $j=1,\dots,n$.
The ring $\Diff (n)$ is free as a one-sided $\Un$-module; the left and right $\Un$-module structures are related by
\begin{equation}\label{relsweights}
\th_i\Z^j=\Z^j(\th_i+\delta_i^j)\ ,\ \th_i\der_j=\der_j(\th_i-\delta_i^j)\ .
\end{equation}
The defining relations for the generators $\Z^j$ and $\der_j$, $j=1,\dots,n$, read, see 
\cite{KO5},
\begin{equation}\label{defhdiffA}
 \Z^i \Z^j=\sum_{k,l}\RR_{kl}^{ij}\Z^k \Z^l\ ,\qquad \der_i\der_j=\sum_{k,l}\RR_{ji}^{lk}\der_k\der_l\ ,\qquad 
 \Z^i\der_j=\sum_{k,l}\RR_{lj}^{ki}[\varepsilon_k]\der_k \Z^l-\delta_{j}^i\ ,
\end{equation}
or, in components, 
\begin{equation}\label{relshdiffcompx}
\Z^i \Z^j=\frac{\th_{ij}+1}{\th_{ij}}\Z^j \Z^i\ ,\ i<j\ ,\end{equation}
\begin{equation}\label{relshdiffcompd}
\der_i \der_j=\frac{\th_{ij}-1}{\th_{ij}}\,\der_j \der_i\ ,\ i<j\ ,\end{equation}
\begin{equation}\label{relshdiffcompdx}
\Z^i \der_j=\left\{
\begin{array}{ll}
\der_j \Z^i&,\ \ i<j\ , \\[.4em]
\displaystyle{\frac{\th_{ij}(\th_{ij}-2)}{(\th_{ij}-1)^2}}\,\der_j \Z^i&
,\  i>j\ ,\end{array}\right.
\end{equation}
\begin{equation}\label{relshdiffcompdx2} 
\Z^i \der_i = \sum_j\displaystyle{\frac{1}{1-\th_{ij}}}\,\der_j \Z^j-1
\end{equation}
which form a set of ordering relations of $\Diff(n)$ with respect to the order,
\begin{equation}\label{partialordergenerators}
\der_j < \der_i < \Z^j < \Z^i \quad \text{when} \quad 1 \leq i < j \leq n \ .
\end{equation}
We have seen in \textbf{Chapter \ref{ChapterLiestepAlge}} that $\Diff(n)$ satisfies the PBW property: the monomials
\begin{equation}\label{UhbasisDiff}
(\der_n)^{l_n} \dots (\der_1)^{l_1} (\Z^1)^{k_1} \dots (\Z^n)^{k_n} \quad , \quad k_1,\dots,k_n,l_1,\dots,l_n \in \mathbb{Z}_{\geq 0}
\end{equation}
form a $\Un$-basis of $\Diff(n)$. A $\Un$-basis is also formed by the monomials
\begin{equation}\label{UhbasisDiff211}
(\Z^1)^{k_1} \dots (\Z^n)^{k_n} (\der_n)^{l_n} \dots (\der_1)^{l_1} \quad , \quad k_1,\dots,k_n,l_1,\dots,l_n \in \mathbb{Z}_{\geq 0}\ .
\end{equation}
The relations \eqref{relsweights} are called the $\h(n)$-weight relations of $\Diff(n)$. An element $f \in \Diff(n)$ is said to have an $\h(n)$-weight $\omega \in \h(n)^\star$ if
\begin{equation}\label{hdesp1aw}
\th_i f = f \left(\th_i+\omega(\th_i)\right)\ ,\ i=1,\dots,n\ .
\end{equation}
The ring $\Diff(n)$ acquires the structure of a filtered algebra by the assignment
\begin{equation}\label{filtrationDiff}
\deg(\Z^i) = \deg(\der_i) = 1 \quad , \quad i=1,\dots,n .
\end{equation}
 The associated graded algebra is denoted gr$\Diff(n)$. 

\vskip .1cm
The ring $\Diff (n)$ admits Zhelobenko automorphisms $\q_i$, $i=1,\dots,n-1$, given by \eqref{automq}.
The operators $\q_i$, $i=1,\dots,n-1$, generate the action of the braid group, see \cite{KO1}.

\vskip .2cm
The  ring $\Diff (n)$ admits an involutive anti-automorphism $\epsilon$, defined by 
\begin{equation}\label{antiautom}
\epsilon(\der_i)=\varphi_i\Z^i\ ,\ \epsilon(\Z^i)=\der_i\varphi_i^{-1}\ ,
\ \text{where}\ \varphi_i:=\frac{\psi_i}{\psi_i[-\varepsilon_i]}=\prod_{k:k>i}\frac{\th_{ik}}{\th_{ik}-1}\ ,
\ i=1,\dots,n \ ,
\end{equation} 
The proof reduces to the formula 
$\frac{\varphi_i[-\epsilon_j]}{\varphi_i}=\frac{\th_{ij}^2-1}{\th_{ij}^2}\ \ \text{for}\ \ 1\leq i<j\leq n$. 
 
\begin{lemm}\label{qga}
Let $\Gamma_i:=\der_i\Z^i\ $ for $i=1,\dots,n$. Then 
\begin{itemize}
\item[(i)] $\Gamma_i\Z^j=\displaystyle{\frac{\th_{ij}+1}{\th_{ij}}}\Z^j\Gamma_i\ $ and 
$\ \Gamma_i\der_j=\displaystyle{\frac{\th_{ij}-1}{\th_{ij}}}\der_j\Gamma_i\ $ for $i\neq j$, $\ i,j=1,\dots,n$.
\item[(ii)] $\q_i(\Gamma_j)=\Gamma_{s_i(j)}$ for $i=1,\dots,n-1$ and $j=1,\dots,n$.
\item[(iii)] $\Gamma_i\Gamma_j=\Gamma_j\Gamma_i\ $ for $i,j=1,\dots,n$.
\end{itemize}
\end{lemm}
\begin{proof}
Formulas (i) and (ii) are obtained by a direct calculation; (iii) follows from (i).
\end{proof}

\vskip .2cm
We will use the following technical lemma whose proof consists in a direct calculation.
\begin{lemm}\label{lecommfamilies}
Let $\mathfrak{A}$ be an associative algebra. Assume that elements $\breve{h}_i$, $\breve{Z}_i,\breve{Z}_i\in\mathfrak{A}$, $i=1,\dots,n$, satisfy
$$\breve{h}_i\breve{h}_j=\breve{h}_j\breve{h}_i\ ,\ \breve{h}_i\breve{Z}^j=\breve{Z}^j(\breve{h}_i+\delta_i^j)\ ,\ \breve{h}_i\breve{Z}_j=\breve{Z}_j(\breve{h}_i-\delta_i^j)\ ,\ 
i,j=1,\dots, n\ .$$
Let $\breve{h}_{ij}:=\breve{h}_i-\breve{h}_j$ and 
$\breve{\psi}_i:=\prod_{k:k>i}\breve{h}_{ik}$, $\breve{\psi}_i':=\prod_{k:k<i}\breve{h}_{ik}$,  $i=1,\dots,n$.
Assume that the elements $\breve{h}_{ij}$ are invertible. Then 

\vskip .2cm
\noindent (i) the elements $\breve{Z}^i$ satisfy $$\breve{Z}^i\breve{Z}^j=\frac{\breve{h}_{ij}+1}{\breve{h}_{ij}}\breve{Z}^j\breve{Z}^i\ \ \text{for}\ \  i<j\ ,\ i,j=1,\dots,n$$ if and only if 
any of the two families $\{\breve{Z}^{\circ i}\}_{i=1}^n$ or $\{\breve{Z}'^{\circ i}\}_{i=1}^n$ is commutative, where 
\begin{equation}\label{twocommfamiliesa}  \breve{Z}^{\circ i}:=\psi_i\breve{Z}^i\ ,\ \breve{Z}'^{\circ i}:=\breve{Z}^i\psi_i'\ ;\end{equation}

\vskip .2cm
\noindent (ii) the elements $\breve{Z}_i$ satisfy $$\breve{Z}_i\breve{Z}_j=\frac{\breve{h}_{ij}-1}{\breve{h}_{ij}}\breve{Z}_j\breve{Z}_i\ \ \text{for}\ \  i<j\ ,\ i,j=1,\dots,n$$ if and only if 
any of the two families $\{\breve{Z}^{\circ}_i\}_{i=1}^n$ or $\{\breve{Z}'^{\circ}_i\}_{i=1}^n$ is commutative, where 
\begin{equation}\label{twocommfamiliesb}  \breve{Z}^{\circ}_i:=\psi_i\breve{Z}_i\ ,\ \breve{Z}'^{\circ}_i:=\breve{Z}_i\psi_i'\ .\end{equation}

\end{lemm}

\section{Zero weight  elements}\label{quadce}

\subsection{Quadratic central elements}

Define
\begin{equation}\label{defcentralelemDiff(n)11}
c_k:=\sum_j \frac{\partial e_k}{\partial\th_j}\Gamma_j-e_k\ .
\end{equation}
It follows from Lemma \ref{qga} that $\q_j(c_k)=c_k$ for all $j=1,\dots,n-1$ and $k=1,\dots,n$.

\begin{prop}\label{quceel}
The elements $c_k$, $k=1,\dots,n$, belong to the center of the ring $\Diff (n)$.
\end{prop}
\begin{proof}
We introduce the generating functions $e(t):=\sum_{k=0}^ne_kt^k=\prod_i(1+\th_i t)$
and 
$$c(t):=\sum_{k=1}^nc_kt^k=u(t)e(t)+1\ \ \text{with}\ \ u(t):=t\sum_i\frac{1}{1+\th_i t}\Gamma_i-1\ .$$
The expression $u(t)$ is introduced for convenience; the denominator $1+\th_i t$, not defined in the ring $\Diff(n)$, vanishes in the product $u(t)e(t)$.
We check that the polynomial $c(t)$ is central. We have 
\begin{equation}\label{quacen1}
\Z^je(t)=\frac{1+(\th_j-1)t}{1+\th_jt}e(t)\Z^j\ .
\end{equation}
$$\text{Next,}\ \  \Z^ju(t)=\biggl(\sum_{k:k\neq j}\frac{t}{1+\th_k t}\frac{\th_{kj}}{\th_{kj}+1}
\Gamma_k+
\frac{t}{1+(\th_j-1)t}\Bigl(\sum_k\frac{1}{1-\th_{jk}}\Gamma_k-1\Bigr)-1\biggr)\Z^j\ .\hspace{2cm}$$
The coefficient of $\Gamma_k$ in this expression is equal to
$\frac{t(1+t \th_j)}{(1+t\th_k)(1+(\th_j-1)t)}$
for both $k\neq j$ and $k=j$. Therefore, 
\begin{equation}\label{quacen2}
\Z^ju(t)=\frac{1+\th_jt}{1+(\th_j-1)t}u(t)\Z^j\ .
\end{equation}
Combining \eqref{quacen1} and \eqref{quacen2} we find that $c(t)$ commutes with $\Z^j$, $j=1,\dots,n$. For $\der_j$ one can either make a parallel calculation or use the anti-automorphism \eqref{antiautom}.
\end{proof}

\begin{lemm}\label{gammaandc}
(i) The matrix $V$, defined by $V^k_j:=\frac{\partial e_j}{\partial \th_k}$, is invertible. 
Its inverse is 
\begin{equation}\label{defV-1matrix11}
(V^{-1})^j_i=\frac{(-1)^{j-1}\th_i^{n-j}}{\chi_i}\ ,
\end{equation}
where the elements $\chi_i$ are defined in \eqref{defphichi}. 
 
\vskip .2cm\noindent 
(ii) We have
\begin{equation}\label{znachc}
\Gamma_j = \frac{\th_j^n-\th_j^nc(-\th_j^{-1})}{\chi_j} = \sum_k (V^{-1})_j^k c_k + \frac{\th_j^n}{\chi_j} \ .
\end{equation}
\end{lemm}
\noindent{\it{Proof.}} (i) See, e.g. \cite{OP}, Proposition 4.

\noindent
(ii) Since $c_k=\sum_jV_k^j\Gamma_j-e_k$, we have
$$\Gamma_j = \sum_k(V^{-1})_j^k(c_k+e_k)=\frac{1}{\chi_j}\sum_k(-1)^{k-1}\th_j^{n-k}(c_k+e_k)=-\frac{\th_j^n}{\chi_j}(c(-\th_j^{-1})+e(-\th_j^{-1})-1)\ .$$
Since $e(-\th_j^{-1})=0$, we obtain \eqref{znachc}. Another proof: evaluate $c(t)$ at $t=-\th_j^{-1}$.\hfill$\square$

\subsection{Subring of $\hn$-weight zero elements}
Let $\Diff^0(n)$ be the subring of $\Diff(n)$ formed by the elements of $\h(n)$-weight zero.

\begin{lemm}\label{zeroweightelementscenter}
$\Diff^0(n)$ is freely generated over $\Un$ by the (commutative) elements $c_1,\dots,c_n$, 
\begin{equation}\label{zeroweightsubring}
\Diff^0(n) = \Un[c_1,\dots,c_n] \ .
\end{equation}
\end{lemm}
\begin{proof}
Clearly, the elements $c_1,\dots,c_n$ are in $\Diff^0(n)$. For any $i \in \{1,\dots,n\}$ and $\alpha_i \in \mathZ$ the element $A = \der_i^{\alpha_i} (\Z^i)^{\alpha_i}$ is in $\Diff^0(n)$.  Lemma \ref{gammaandc} implies that $A$ is in the $\Un$-module generated by the elements $c_1,\dots,c_n$. By induction, there exists a tensor with components $T^{\alpha_1,\dots,\alpha_n}_{\beta_1,\dots,\beta_n} \in \Un$, $(\alpha_1,\dots,\alpha_n,\beta_1,\dots,\beta_n) \in \mathZ^{2n}$ such that for any $(\alpha_1,\dots,\alpha_n) \in \mathZ^n$, the element
$$B = \der_n^{\alpha_n} \dots \der_1^{\alpha_1} (\Z^1)^{\alpha_1} \dots (\Z^n)^{\alpha_n}$$
can be written as $B = T^{\alpha_1,\dots,\alpha_n}_{\beta_1,\dots,\beta_n} c_1^{\beta_1} \dots c_n^{\beta_n} $.
We see from equation \eqref{znachc} that
\begin{equation}\label{tensornonzero}
T^{\alpha_1,\dots,\alpha_n}_{\beta_1,\dots,\beta_n} \neq 0 \ \Longrightarrow \ \alpha_1+\dots+\alpha_n \geq \beta_1+\dots+\beta_n \ .
\end{equation}
The ring $\Diff(n)$ has the PBW property in the form \eqref{UhbasisDiff}. It implies together with \eqref{tensornonzero} that the tensor $T$ is invertible. This implies the validity of \eqref{zeroweightsubring}.
\end{proof}

\section{Isomorphism between rings of fractions}\label{isowe}

It follows from the results of \cite{KO4} that the ring $\Diff(n)$ has no zero divisors. Let $\text{S}_n$ be the multiplicative set generated by $\Z^j$ and $\th_{ij}+k$, $i,j=1,\dots,n$ and $k \in \mathbb{Z}$.

\subsection{Localizations of $\Diff(n)$ and $W_n$}\label{secisoloc1}
The (left) Ore condition for a multiplicative set $S$ of a domain $\mathcal{D}$ reads:  
$S \, r \cap \mathcal{D} \, s \neq \emptyset$, $(r,s) \in \mathcal{D} \times S$.
The set $\text{S}_\Z$ satisfies both left and right Ore conditions (see, e.g., \cite{A} for definitions): say, for the 
left Ore conditions 
we have to check only that for any $\Z^k$ and a monomial $m=\der_{i_1}\dots \der_{i_A}\Z^{j_1}\dots \Z^{j_B}$ there exist $\tilde{s}\in\text{S}_\Z$ and $\tilde{m}\in\Diff(n)$ such that 
$\tilde{s}m=\tilde{m}\Z^k$. The structure of the commutation relations \eqref{relshdiffcompx}-\eqref{relshdiffcompdx2}) shows that one can choose 
$\tilde{s}=(\Z^k)^{\nu}$ with sufficiently large $\nu$. Denote by $\text{S}_\Z^{-1}\Diff (n)$ the localization of the ring $\Diff (n)$ with respect to the set $\text{S}_\Z$. 

\vskip .2cm
\noindent{\bf The isomorphism.} Let $\text{W}_n$ be the Weyl algebra of rank $n$ (see \textbf{Section \ref{WeylAlgebraSec}}). Let $\text{T}$ be the multiplicative set generated by $X^jD_j-X^kD_k+\ell$, $1\leq j<k\leq n$, $\ell\in\mathbb{Z}$, and $X^j$, $j=1,\dots,n$. As above, one can check that the set $\text{T}$ is a left and right Ore set in $\W_n$.
Denote by $\text{T}^{-1}\text{W}_n$ the localization of $\text{W}_n$ relative to the set $\text{T}$.

Let $a_1,\dots,a_n$ be a family of commuting variables. We shall use the following notation : 
$$\begin{array}{c}\mathcal{H}_j:=D_jX^j\ ,\ \mathcal{H}_{jk}:=\mathcal{H}_j-\mathcal{H}_k\ ,\\[.1em]
\Psi'_j:=\prod_{k:k<j}\mathcal{H}_{jk}\ ,\ \Psi_j:=\prod_{k:k>j}\mathcal{H}_{jk}\ ,\\[.3em]
\mathbf{C}(t):=\sum_{k=1}^n a_kt^k\ ,\ \Upsilon_i:=\mathcal{H}_i^n \left(1-\mathbf{C}(-\mathcal{H}_i^{-1})\right)\ .
\end{array}$$
Here $\mathbf{C}$ is a polynomial of degree $n$ so the elements $\Upsilon_i$ are polynomials in $\mathcal{H}_i$, $i=1,\dots,n$.

\begin{theo}\label{isoloc}
The ring $\text{S}_n^{-1}\Diff (n)$ is isomorphic to the ring $\K[a_1,\dots,a_n]\otimes \text{T}^{-1}\text{W}_n$.
\end{theo}
\begin{proof}
The knowledge of the central elements (Proposition \ref{quceel}) allows to exhibit a generating set of the ring $\text{S}_n^{-1}\Diff (n)$ in which the required isomorphism is quite transparent. 

\vskip .2cm
In the localized ring $\text{S}_n^{-1}\Diff (n)$ we can use the set of generators $\{\th_i,\Z^i,\Gamma_i\}_{i=1}^n$ instead of $\{\th_i,\Z^i,\der_i\}_{i=1}^n$. 
By Lemma \ref{gammaandc} (ii), the set $\{\th_i,\Z^i,c_i\}_{i=1}^n$ is also a generating set. Finally, $\b_{\text{D}}:=\{\th_i,\Z^{\circ i},c_i\}_{i=1}^n$, 
where $\Z'^{\circ i}:=\Z^i\psi_i'\ ,\ i=1,\dots,n$, is a generating set of the localized ring 
$\text{S}_n^{-1}\Diff (n)$ as well. It follows from Lemma \ref{lecommfamilies} that the family $ \{\Z'^{\circ i}\}_{i=1}^n$ is commutative. The complete set of the 
defining relations for the generators from the set $\b_{\text{D}}$ reads 
\begin{equation}\label{defrebd}\begin{array}{l}\th_i\th_j=\th_j\th_i\ ,\ \th_i \Z'^{\circ j}=\Z'^{\circ j}(\th_i+\delta_i^j)\ ,\ \Z'^{\circ i}\Z'^{\circ j}=\Z'^{\circ j}\Z'^{\circ i}\ ,\ i,j=1,\dots,n\ ,\\[.1em]
c_i\ \text{are central}\ ,\ i=1,\dots,n\ .\end{array}\end{equation}

In the localized ring $\K[a_1,\dots,a_n]\otimes \text{T}^{-1}\text{W}_n$ we can pass to the set of generators $\b_{\text{W}}:=\{\mathcal{H}_i,X^i,a_i\}_{i=1}^n$ with the defining relations
$$\begin{array}{l}\mathcal{H}_i\mathcal{H}_j=\mathcal{H}_j\mathcal{H}_i\ ,\ \mathcal{H}_i X^j=X^j(\mathcal{H}_i+\delta_i^j)\ ,\ X^{i}X^{j}=X^j X^i\ ,\ i,j=1,\dots,n\ ,\\[.1em]
a_i\ \text{are central}\ ,\ i=1,\dots,n\ .\end{array}$$
Therefore we have the isomorphism
$\mu\colon \K[a_1,\dots,a_n]\otimes \text{T}^{-1}\text{W}_n\to \text{S}_n^{-1}\Diff (n)$
given on our sets $\b_{\text{D}}$ and $\b_{\text{W}}$ of generators by
\begin{equation}\label{fromwtodiffnb}
\mu\colon  X^i\mapsto \Z'^{\circ i}\ ,\ \mathcal{H}_i\mapsto \th_i
\ ,\ a_i\mapsto c_i\ ,\ i=1,\dots,n\ .
\end{equation}
The proof is completed.
\end{proof}

\noindent We shall now rewrite the formulas for the isomorphism $\mu$ in terms of the original generators of the rings $\text{S}_n^{-1}\Diff (n)$ and 
$\K[a_1,\dots,a_n]\otimes \text{T}^{-1}\text{W}_n$. 

\begin{lemm}\label{lemmaiso}
We have
\begin{eqnarray}\label{fromwtodiff}
&\mu \colon  X^i\mapsto \Z^i\psi_i'\ ,\ D_i\mapsto (\psi_i')^{-1}\th_i(\Z^i)^{-1}\ ,\ a_i\mapsto c_i\ ,\ i=1,\dots,n\ ,&\\[.2em]
\label{fromdifftow}
&\mu^{-1}\colon  \th_i\mapsto \mathcal{H}_i\ ,\ \Z^i\mapsto X^i\frac{1}{\Psi_i'}\ ,\ \der_i\mapsto \frac{\Upsilon_i}{\Psi_i}
(X^i)^{-1}\ ,\ i=1,\dots,n \ .&
\end{eqnarray}
\end{lemm}
\begin{proof}
We shall comment only the last formula in \eqref{fromdifftow}. Lemma \ref{gammaandc} part (ii) implies that $\mu^{-1}(\chi_i\Gamma_i)=
\Upsilon_i$ and the formula for $\mu^{-1}(\der_i)$ follows since $\der_i=\Gamma_i(\Z^i)^{-1}$.
\end{proof}

\begin{prop}\label{isocenter}
The center of the ring $\Diff(n)$ is isomorphic to the polynomial ring $\K[t_1,\dots,t_n]$; the isomorphism is given by $t_j\mapsto c_j$, $j=1,\dots,n$.
\end{prop}
\begin{proof}
There are two ways to prove this proposition.
First, a central element of $\Diff(n)$ is of $\hn$-weight zero, so from Lemma \ref{zeroweightelementscenter} it is in $\Un[c_1,\dots,c_n]$. But the intersection of $\Un$ and the center of $\Diff(n)$ is $\K$. This shows that the center of $\Diff(n)$ is $\K[c_1,\dots,c_n]$.

The other way is to use the isomorphism $\mu$ and the triviality of the center of the Weyl algebra $\W_n$. Using the isomorphism $\mu$, we see that the center of $\text{S}_n^{-1} \Diff(n)$ is $\K[c_1,\dots,c_n]$ which is also the center of $\Diff(n)$.
\end{proof}

\chapter{Generalized $\h$-deformed differential operators}\label{GeneralizedDiff}

\vskip -1cm
In previous sections we constructed the ring $\Diff(n)$ and gave its R-matrix description. Similarly to the ring of $q$-differential operators, the ring $\Diff(n)$ realizes the consistent pairing of two $\h$-deformed coordinate rings in $n$ variables. The consistency condition is given by  oscillator like relations \eqref{relshdiffcompdx2}. In this chapter we define and study general consistent pairings of $\h$-deformed coordinate rings. 

\vskip .1cm
In \textbf{Section \ref{Motivations}}, we give the motivation for the study of these pairings. We define two coordinate rings in $nN$ variables with the help of the R-matrix \eqref{dynRcompb}. 
These two rings are the $\h$-deformed coordinate ring $\V(n,N)$ and its ``dual", the $\h$-deformed coordinate ring $\V^\star(n,N)$ of derivatives. The ring $\Diff(n)$ and the reduction algebra $AZ_n$ of \cite{Zh1} provides examples of different pairings. A general pairing between $\V(n,1)$ and $\V^\star(n,1)$ involves an $n$-tuple  of elements $\sigma_1,\dots,\sigma_n \in \Un$. The corresponding rings are denoted $\Diff(\sigma_1,\dots,\sigma_n)$.

\vskip .1cm
In \textbf{Section \ref{PBWPorpertySection}} we investigate the PBW property of $\Diff(\sigma_1,\dots,\sigma_n)$ which is the consistency condition for the pairing. 
We derive an over-determined system of finite-difference equations, called $\Delta$-system, for $\sigma_1,\dots,\sigma_n$ and give its general solution in the field of fractions $\K(n)$ of $\U(n)$ and then in $\Un$ (in \cite{HO2} we used another technique avoiding the field $\K(n)$). 

\section{Motivations and definitions}\label{Motivations}

\subsection{Coordinate rings of $\h$-deformed vector spaces}\label{corivvz}

Let $\F(n,N)$ be the ring with the generators $\Z^{i\alpha}$, $i=1,\dots,n$, $\alpha=1,\dots,N$, and $\th_i$, $i=1,\dots,n$, with the defining relations 
\begin{equation}\label{hdesp1a}
\th_i\th_j=\th_j\th_i\ ,\ i,j=1,\dots,n\ ,
\end{equation}
\begin{equation}\label{hdesp1nn}
\th_i\Z^{j\alpha}=\Z^{j\alpha}(\th_i+\delta_i^j)\ ,\ i,j=1,\dots,n\ ,\ \alpha=1,\dots,N\ .
\end{equation}
The ring $\U(n)$ is naturally a subring of $\F(n,N)$. Let $\bar{\F}(n,N):=\Ub(n)\otimes_{\U(n)}\F(n,N)$.
The coordinate ring $\V(n,N)$ of $N$ copies of the $\h$-deformed vector space is the factor-ring of $\bar{\F}(n,N)$ by the relations
\begin{equation}\label{hdesp2}
\Z^{i\alpha}\Z^{j\beta}=\sum_{k,l}\RR^{ij}_{kl} \Z^{k\beta} \Z^{l\alpha}\ ,\ i,j=1,\dots,n\ ,\ \alpha,\beta=1,\dots,N \ .
\end{equation}
The ring $\V(n,N)$ is the reduction algebra, with respect to $\gl_n$, of the semi-direct product of $\gl_n$ and the abelian Lie algebra $V\oplus V\oplus\dots\oplus V$
($N$ times) where $V$ is the (tautological) $n$-dimensional $\gl_n$-module. According to the general theory of reduction algebras \cite{Zh1,Zh3,KO1,KO4}, 
$\V(n,N)$ is a free left (or right) $\Ub(n)$-module; the ring $\V(n,N)$ has the following PBW property: 
\begin{equation}\label{hdesppbw}
\begin{array}{l}
\text{given an arbitrary order on the set $\Z^{i\alpha}$, $i=1,\dots,n$, $\alpha=1,\dots,N$, the set}\\[.2em]
\text{of all ordered monomials in $\Z^{i\alpha}$ is a basis of the left $\Ub(n)$-module $\V(n,N)$.}
\end{array}
\end{equation} 
Moreover, if $\{\RR_{ij}^{kl}\}_{i,j,k,l=1}^n$ is an arbitrary array 
of functions in $\th_i$, $i=1,\dots,n$, then the above PBW property implies that $\RR$ satisfies the dynamical Yang--Baxter equation when $N\geq 3$.

\vskip .1cm
Similarly, let $\F^*(n,N)$ be the ring with the generators $\der_{i\alpha}$, $i=1,\dots,n$, $\alpha=1,\dots,N$, and $\th_i$, $i=1,\dots,n$, with the defining relations \eqref{hdesp1a}
and 
\begin{equation}\label{hdesp3}
\th_i\der_{j\alpha}=\der_{j\alpha}(\th_i-\delta_i^j)\ ,\ i,j=1,\dots,n\ ,\ \alpha=1,\dots,N\ .
\end{equation}
Let $\bar{\F}^*(n,N):=\Ub(n)\otimes_{\U(n)}\F^*(n,N)$. The $\h(n)$-weights are defined by the same equation \eqref{hdesp1aw}.
The coordinate ring $\V^*(n,N)$ of $N$ copies of the ``dual" $\h$-deformed vector space is the factor-ring of $\bar{\F}^*(n,N)$ by the relations
\begin{equation}\label{hdesp4}
\der_{l\alpha}\der_{k\beta}=\sum_{i,j}\der_{j\beta}\der_{i\alpha}\RR^{ij}_{kl}\ ,\ k,l=1,\dots,n\ ,\ \alpha,\beta=1,\dots,N\ .
\end{equation}
Again, the ring $\V^*(n,N)$ is the reduction algebra, with respect to $\gl_n$, of the semi-direct product of $\gl_n$ and the abelian Lie algebra $V^*\oplus V^*\oplus\dots\oplus V^*$
($N$ times) where $V^*$ is the $\gl_n$-module, dual to $V$. The ring  
$\V^*(n,N)$ is a free left (or right) $\Ub(n)$-module; it has a similar to $\V(n,N)$ PBW property: 
\begin{equation}\label{hdesppbwd}
\begin{array}{l}\text{given an arbitrary order on the set $\der_{i\alpha}$ , $i=1,\dots,n$, $\alpha=1,\dots,N$, the set}\\[.2em] 
\text{of all ordered monomials in $\der_{i\alpha}$ is a basis of the left $\Ub(n)$-module $\V^*(n,N)$.}
\end{array}
\end{equation} 
Again, the PBW property of the algebra defined by the relations \eqref{hdesp4}, 
together with the weight prescriptions \eqref{hdesp3}, implies that $\RR$ satisfies the dynamical Yang--Baxter equation 
when $N\geq 3$.
The matrix algebras with the defining relations of the type  \eqref{hdesp2} appear in the study of the
chiral zero modes of the Wess--Zumino--Novikov--Witten model \cite{FHIOPT,HIOPT}.

\vskip .1cm
For $N=1$ we shall write $\V(n)$ and $\V^*(n)$ instead of $\V(n,1)$ and $\V^*(n,1)$. 

\subsection{Two examples}\label{twoexse}
Before presenting the main question we consider two examples. 

\noindent{\bf 1.} The reduction algebra $\Diff(n)$ is defined in \textbf{Section \ref{ReducAlgeDiff}}. It is generated, over $\Ub(n)$, by $\Z^i$ and $\der_i$, $i=1,\dots,n$. The $\h(n)$-weights of the generators are given by \eqref{hdesp1nn} and \eqref{hdesp3}. The remaining set of the defining relations, over $\Ub(n)$, consists of \eqref{hdesp2}, \eqref{hdesp4} (with $N=1$) and
\begin{equation}\label{stdi2}\Z^{i}\der_{j}=\sum_{k,l}\der_{k} \RR_{lj}^{ki} \Z^{l}-\delta_{j}^i\sigma_i^{(\text{Diff})} \ ,\ \ \text{where $\sigma_i^{(\text{Diff})}=1$, $i=1,\dots,n$.}
\end{equation}
\noindent{\bf 2.} This example is a generalization of the reduction algebra $\mathcal{R}^{\gl_3}_{\gl_2}$ we constructed in \textbf{Section \ref{FirstExempleRedAlg}}. Identifying each $n\times n$ matrix $a$ with the larger matrix $\left(\begin{array}{cc}a&0\\0&0\end{array}\right)$ gives an embedding of $\gl_n$ into ${\bf gl}_{n+1}$.
The resulting reduction algebra $\text{R}^{\gl_{n+1}}_{\gl_n}$ was denoted by $AZ_n$ in \cite{Zh1}. 
It is generated, over $\Ub(n)$, by $x^i$ and $y_i$, the images of the standard generators $e_{i,n+1}$ and $e_{n+1,i}$, $i=1,\dots,n$, of $\U(\gl_{n+1})$ and $\th_{n+1}=h_{n+1}-(n+1)$ where $h_{n+1}$ is the image of $e_{n+1,n+1}$ in $AZ_n$. Let 
\[\der_{i}:=y_{i}\frac{\psi_i}{\psi_i[-\varepsilon_i]}\ ,\]
where the elements $\psi_i$ are defined in \eqref{defphichi} (they depend on $\th_1,\dots,\th_n$ only). 
The $\h(n)$-weights of the generators are given by \eqref{hdesp1nn} and \eqref{hdesp3} while
\[\th_{n+1}\Z^i=\Z^i(\th_{n+1}-1)\ ,\ \th_{n+1}\der_i=\der_i(\th_{n+1}+1)\ ,\ i=1,\dots,n\ .\]
The set of the remaining defining relations consists of \eqref{hdesp2}, \eqref{hdesp4} (with $N=1$) and 
\begin{equation}\label{stdi2a}\Z^{i}\der_{j}=\sum_{k,l}\der_{k} \RR_{lj}^{ki} \Z^{l}-\delta_{j}^i \sigma_i^{(AZ)} \ ,\end{equation}
where 
\begin{equation}\label{stdisiDiff}\sigma_i^{(AZ)}=-\th_i+\th_{n+1}+1\ , \ i=1,\dots,n\ .\end{equation}

The algebra $AZ_n$ was used in \cite{Hom1} for the study of Harish-Chandra modules and in \cite{Zh1} for the construction of the Gelfand--Tsetlin bases \cite{GT}.

The algebra $AZ_n$ has a central element 
$c_1:=\th_1+\dots+\th_n+\th_{n+1} $.
In the factor-algebra $\overline{AZ}_n$ of $AZ_n$ by the ideal, generated by the element $c_1$, the relation \eqref{stdi2a} is replaced by
\begin{equation}\label{stdi2bb}\Z^{i}\der_{j}=\sum_{k,l}\der_{k} \RR_{lj}^{ki} \Z^{l}-\delta_{j}^i \sigma_i^{(\overline{AZ})} \ ,
\end{equation}
with 
\begin{equation}\label{stdisiAZ} \sigma_i^{(\overline{AZ})}=-\th_i-\sum_{k=1}^n \th_k+1\ ,\ i=1,\dots,n.\end{equation}

\subsection{Main question}\label{maquare}

Both rings, $\Diff(n)$ and $\overline{AZ}_n$, satisfy the PBW property. The only difference between these rings is in the form of the zero-order terms $\sigma_i^{(\text{Diff})}$ and $\sigma_i^{(\overline{AZ})}$ in the cross-commutation relations \eqref{stdi2} and \eqref{stdi2bb} (compare to the ring of
$q$-differential operators \cite{WZ} where the zero-order term is essentially - up to redefinitions - unique). 
It is therefore natural to investigate possible generalizations of the rings $\Diff(n)$ and $\overline{AZ}_n$. More precisely, given $n$ elements $\sigma_1,\dots,\sigma_n$ of 
$\Ub(n)$, we let $\Diff(\sigma_1,\dots,\sigma_n)$ be the ring, over 
$\Ub(n)$, with the generators $\Z^{i}$ and $\der_{i}$, $i=1,\dots,n$, subject to the defining relations \eqref{hdesp2}, \eqref{hdesp4} (with $N=1$) and the oscillator-like 
relations 
\begin{equation}\label{stdi2age}
\Z^{i}\der_{j}=\sum_{k,l}\der_{k} \RR_{lj}^{ki} \Z^{l}-\delta_{j}^i \sigma_i\ .
\end{equation}
The weight prescriptions for the generators are given by \eqref{hdesp1nn} and \eqref{hdesp3}. The diagonal form of the zero-order term (the Kronecker symbol $\delta_{j}^i$ in the right hand side of \eqref{stdi2age}) is dictated by the $\h(n)$-weight considerations.

\vskip .2cm
We shall study conditions under which the ring $\Diff(\sigma_1,\dots,\sigma_n)$ satisfies the PBW property. More specifically, since the rings $\V(n)$
and $\V^*(n)$ both satisfy the PBW property, our aim is to study conditions under which $\Diff(\sigma_1,\dots,\sigma_n)$ is isomorphic, as a 
$\Ub(n)$-module, to $\V^*(n)\otimes_{\Ub(n)} \V(n)$. 
Similarly to $\Diff(n)$, the assignment \eqref{filtrationDiff} defines a structure of a filtered algebra on $\Diff(\sigma_1,\dots,\sigma_n)$. The associated graded algebra 
is the homogeneous algebra $\Diff(0,\dots,0)$. This homogeneous algebra has the desired PBW property because it is the reduction algebra, 
with respect to $\gl_n$, of the semi-direct product of $\gl_n$ and the abelian Lie algebra $V\oplus V^*$. 

\vskip .2cm
The standard argument shows that the ring $\Diff(\sigma_1,\dots,\sigma_n)$ can be viewed as
a deformation of the homogeneous ring $\Diff(0,\dots,0)$: we may replace $\Z^i$ by $\hbar \Z^i$ and consider $\hbar$ as the deformation parameter; if $\hbar\neq 0$, 
the renormalization $\sigma\mapsto\hbar\sigma$ establishes the isomorphism of 
the rings $\Diff(\hbar\sigma_1,\dots,\hbar\sigma_n)$ and $\Diff(\sigma_1,\dots,\sigma_n)$. Thus our aim is to study the conditions  under which this deformation is flat. 

\section{Poincar\'e$-$Birkhoff$-$Witt property}\label{PBWPorpertySection}
The explicit defining relations for the ring $\Diff(\sigma_1,\dots,\sigma_n)$ are \eqref{relshdiffcompx}, \eqref{relshdiffcompd}, \eqref{relshdiffcompdx} and
\begin{equation}\label{hdesp10c} 
\Z^i \der_i = \sum_j\displaystyle{\frac{1}{1-\th_{ij}}}\,\der_j \Z^j - \sigma_i\ ,i=1,\dots,n\ .
\end{equation}
It turns out that the PBW property is equivalent to a system of finite-difference equations for the elements $\sigma_1,\dots,\sigma_n\in\Ub(n)$.

\vskip .2cm
\begin{prop}\label{equasigmasb}
The ring $\Diff(\sigma_1,\dots,\sigma_n)$ satisfies the PBW property if and only if the elements $\sigma_1,\dots,\sigma_n\in\Ub(n)$ satisfy the following linear system of finite-difference equations
\begin{equation}\label{eqsigib}
\th_{ij} \Delta_j \sigma_i = \sigma_i - \sigma_j\ ,\ i,j=1,\dots,n \ .
\end{equation}
\end{prop}
\noindent{\it Proof.} We can consider \eqref{relsweights}, \eqref{relshdiffcompx}, \eqref{relshdiffcompd}, \eqref{relshdiffcompdx}, \eqref{hdesp10c} as the set of ordering relations and use the diamond lemma, see \cite{Bo,B,Ne}, for the investigation of the PBW property. The relations \eqref{relshdiffcompx}, \eqref{relshdiffcompd} and \eqref{relshdiffcompdx} are compatible with the $\h(n)$-weights of the generators 
$\Z^i$ and $\der_i$, $i=1,\dots,n$, so we have to check the possible ambiguities involving the generators $\Z^i$ and $\der_i$, $i=1,\dots,n$, only. 
The properties \eqref{hdesppbw} and \eqref{hdesppbwd} show that the ambiguities of the forms $\Z\Z\Z$ and $\der\der\der$ are resolvable. It remains to check the ambiguities  
\begin{equation}\label{lefam}
\Z^i \der_j \der_k\ \ \text{and}\ \ \Z^j \Z^k \der_i\  .
\end{equation}
It follows from the properties \eqref{hdesppbw} and \eqref{hdesppbwd} that the choice of the order for the generators with indices $j$ and $k$ in \eqref{lefam} is irrelevant.
Besides, it can be verified directly that the ring $\Diff(\sigma_1,\dots,\sigma_n)$, with arbitrary $\sigma_1,\dots,\sigma_n\in\Ub(n)$ admits an involutive 
anti-automorphism $\epsilon$, defined by 
\begin{equation}\label{antiautomb}
\!\epsilon(\th_i)=\th_i\ ,\ \epsilon(\der_i)=\varphi_i\Z^i\ ,\ \epsilon(\Z^i)=\der_i\varphi_i^{-1}\ ,
\end{equation} 
where $\displaystyle{\varphi_i:=\frac{\psi_i}{\psi_i[-\varepsilon_i]}=\prod_{k:k>i}\frac{\th_{ik}}{\th_{ik}-1}\ ,\ i=1,\dots,n}$.
By using the anti-automorphism $\epsilon$ we reduce the check of the ambiguity $\Z^j \Z^k \der_i$ to the check of the ambiguity $\Z^i \der_j \der_k$.

\vskip .2cm
Since the associated graded algebra with respect to the filtration \eqref{filtrationDiff} has the PBW property, we have, in the check of the ambiguity 
$\Z^i \der_j \der_k$, to track only those ordered terms whose degree is smaller than 3. 
We use the symbol $u\;\rule[-3.5mm]{.28mm}{7mm}_{\;\text{l.d.t.}}$ to denote the part of the ordered expression for $u$ containing these lower degree terms.

\vskip .2cm
\noindent{\bf Check of the ambiguity $\Z^i \der_j \der_k$.} We calculate, for $i,j,k=1,\dots,n$, 
\begin{equation}\label{fiwa}\left( x^{i}\der_{j}\right) \der_{k}\;\rule[-3.5mm]{.28mm}{7mm}_{\;\text{l.d.t.}}=\left( \sum_{u,v}\RR^{ui}_{vj}[\varepsilon_u]\der_{u}x^{v}-
\delta^i_j\sigma_{i}\right)\der_{k}\;\rule[-3.5mm]{.28mm}{7mm}_{\;\text{l.d.t.}}=
-\sum_{u}\RR^{ui}_{kj}[\varepsilon_u]\der_{u}\sigma_{k}-\delta^i_j\sigma_{i}\der_{k}\ ,\end{equation} 
\begin{equation}\label{sewa}\begin{array}{rcl}
x^{i}\left( \der_{j} \der_{k}\right)\;\rule[-3.5mm]{.28mm}{7mm}_{\;\text{l.d.t.}}&=&\displaystyle{ x^{i}\sum_{a,b}\RR^{ab}_{kj}
\der_{b}\der_{a}\;\rule[-3.5mm]{.28mm}{7mm}_{\;\text{l.d.t.}}=
\sum_{a,b}\RR^{ab}_{kj}[-\varepsilon_i]\left( \sum_{c,d}\RR^{ci}_{db}[\varepsilon_c]\der_{c}x^{d}-\delta^i_b
\sigma_{i}\right)\der_{a}\;\rule[-3.5mm]{.28mm}{7mm}_{\;\text{l.d.t.}} }\\[2em]
&=&\displaystyle{ 
-\sum_{a,b,c}\RR^{ab}_{kj}[-\varepsilon_i] \RR^{ci}_{ab}[\varepsilon_c]\der_{c}\sigma_{a}-\sum_{a}\RR^{ai}_{kj}[-\varepsilon_i] 
\sigma_{i}\der_{a} }\ .\end{array}\end{equation} 
Comparing the resulting expressions in \eqref{fiwa} and \eqref{sewa} and collecting coefficients in $\der_u$, we find the necessary and sufficient condition for 
the resolvability of the ambiguity $\Z^i \der_j \der_k$:  
\begin{equation}\label{cocoon}
\RR^{ui}_{kj}[\varepsilon_u]\sigma_k[\varepsilon_u]+\delta^i_j\delta^u_k\sigma_i=\sum_{a,b} \RR^{ab}_{kj}[-\varepsilon_i]\RR^{ui}_{ab}[\varepsilon_u]
\sigma_a[\varepsilon_u]+\RR^{ui}_{kj}[-\varepsilon_i]\sigma_i\ ,\ i,k,j,u=1,\dots,n\ .
\end{equation} 
Shifting by $-\varepsilon_u$ and using the property \eqref{weze} together with the ice condition \eqref{iceco}, we rewrite \eqref{cocoon} in the form 
\begin{equation}\label{cocoon2}
\RR^{ui}_{kj}\left(\sigma_k-\sigma_i[-\varepsilon_u]\right)+\delta^i_j\delta^u_k\sigma_i[-\varepsilon_u]=\sum_{a,b} \RR^{ab}_{kj}\RR^{ui}_{ab}\sigma_a\ .
\end{equation} 
For $j=k$ the system \eqref{cocoon2} contains no equations. For $j\neq k$ we have two cases:

\noindent (i) $u=j$ and $i=k$. This part of the system \eqref{cocoon2} reads explicitly (see \eqref{dynRcompb})
$\sigma_k-\sigma_k[-\varepsilon_j]=\frac{1}{\th_{kj}}\left(\sigma_k-\sigma_j\right)$.
This is the system \eqref{eqsigib}.

\noindent (ii) $u=k$ and $i=j$. This part of the system \eqref{cocoon2} reads explicitly 
$\frac{1}{\th_{kj}}\left(\sigma_k-\sigma_j[-\varepsilon_k]\right)+\sigma_j[-\varepsilon_k]=\frac{1}{\th_{kj}^2}\sigma_k+\frac{\th_{kj}^2-1}{\th_{kj}^2}\sigma_j$, 
which reproduces the same system \eqref{eqsigib}.\hfill$\square$

\vskip .1cm
In the next proposition, we rewrite the system \eqref{eqsigib} in a matrix form.

\begin{prop}\label{{propHecketyperelasigmaR}}
Let $\Sigma$ and $X$ be the diagonal $n \times n$ matrices of coefficients $\Sigma^i_j = \delta^i_j \sigma_i$ and $X^i_j = \delta^i_j \Z^i$. In the localized ring $S_n^{-1} \Diff(n)$, the system \eqref{eqsigib}, can be rewritten in the form
$$\left( \RR_{12} \Sigma_1 + X_1 \Sigma_2 X_1^{-1} \right) \left( \RR_{12} - \mathbbm{1} \right) = 0 \ .$$
\end{prop}
\begin{proof}
It is a rewriting of the system of equations \eqref{cocoon2}.
\end{proof}

\subsection{$\Delta$-system}

We have considered the algebra $\Diff(\sigma_1,\dots,\sigma_n)$ with  elements $(\sigma_1,\dots,\sigma_n)$ in $\Un$. In this section, we rather study the system of equations \eqref{eqsigib} in the field of fractions $\K(n)$ of $\U(n)$.

\vskip .1cm
For $f \in \K(n)$, $i \in \{1,\dots,n\}$, and $\lambda_i \in \K$, $f \vert_{\th_i = \lambda_i}$ denotes the evaluation of $f$ at  
$\th_i=\lambda_i$. Define $\text{Dom}(f)$ to be the subset of $\K^n$ of elements $(\lambda_1,\dots,\lambda_n)$ on which values of $f$ are finite. 

\vskip .1cm 
The system \eqref{eqsigib} is closely related to the following linear system of finite-difference equations for one element $\sigma\in \K(n)$:  
\begin{equation}\label{eqsigibfopo2}
\Delta_i\Delta_j \left( \th_{ij}\sigma\right) =0\ ,\ i,j=1,\dots,n \ ,
\end{equation}
or, equivalently, 
\begin{equation}\label{eqsigibfopo}
\th_{ij} \Delta_j\Delta_i \sigma =\Delta_i \sigma -\Delta_j \sigma\ ,\ i,j=1,\dots,n \ .
\end{equation}
We call this system the $\Delta$-system in $\K(n)$. 

\vskip .3cm
{\bf Complete symmetric polynomials}
\begin{defi}\label{vespcoh}
Let $\mathcal{H}=\oplus_{j=0}^\infty\K H_j$ and $\mathcal{H}'=\oplus_{j=1}^\infty\K H_j$ where $H_j$ are the complete symmetric polynomials in the variables 
$\th_1,\dots,\th_n$,
\begin{equation}\label{homogeneouspoly11}
H_j=\sum_{1 \leq i_1\leq \dots\leq i_j\leq n}\th_{i_1}\dots \th_{i_j}\ .
\end{equation}  
\end{defi}

\begin{lemm}\label{idesy0}
Let $k \in \mathbb{Z}_{\geq 0}$. We have
\begin{equation}\label{equidesy0}
\sum_{j=1}^n\frac{\th_j^k}{\chi_j} = \left\{\begin{array}{ll}0\ ,& k=0,1,\dots,n-2\ ,\\[.4em]
H_{k-n+1}\ ,&k \geq n-1 \ .\end{array}\right.
\end{equation}
\end{lemm}
\begin{proof}
For any $k \in \mathbb{Z}_{\geq 0}$ and $t$ an auxiliary indeterminate, we define
$P_k = \sum_{j=1}^n \frac{\th_j^k}{\chi_j}$ and $P(t) = t \sum_{k=0}^{\infty} P_k t^k $.
We have
$$P(t) = t \sum_{k=0}^{+\infty} \sum_{j=1}^n \frac{(\th_j t)^k}{\chi_j} = t \sum_{j=1}^n \frac{\sum_{k=0}^{+\infty} (\th_j t)^k}{\chi_j} = \sum_{j=1}^n \frac{t}{1-\th_j t}\frac{1}{\chi_j} = \sum_{j=1}^n \frac{1}{t^{-1}-\th_j }\frac{1}{\chi_j} \ .$$
The last term expression is the partial fraction decomposition of $\prod_{j=1}^n (t^{-1}-\th_j)^{-1}$ with respect to $t^{-1}$. Therefore, 
$P(t) = \prod_{j=1}^n \frac{1}{t^{-1}-\th_j} = t \sum_{k=n-1}^{+\infty} H_{k-n+1} t^k$
which concludes the proof.
\end{proof}
{\bf Subspace $\mathcal{W}$ of $\Un$}
\begin{defi}\label{vespmj}
Let $\Wc_j$, respectively, $\Wb_j$, $j=1,\dots,n$, be vector spaces of elements of $\K(n)$ of the form
\begin{equation}
\frac{\pi(\th_j)}{\chi_j}\ \ \text{where} \ \pi(t) \in \K[t] \ ,\ \ \text{respectively}\ ,\ \pi(t) \in \K(t)\ , 
\end{equation}
and $\chi_j$ is defined in \eqref{defphichi}. Let $\Wc=\sum\Wc_j$ and $\Wb=\sum\Wb_j$. 
\end{defi}

\begin{prop}\label{WbandWclink}
We have
$\Wc = \Wb \cap \Un$.
\end{prop}
\begin{proof}
Fix $j=1,\dots,n$. Consider $f\in\Wb$ as a rational function in one variable $\th_j$, with values in the field of rational functions in other variables.
Possible poles of $f$ are at $\th_j=\th_k$, $k>1$, or at $\th_j=c$, $c\in\K$. For $f\in  \Wb \cap \Un$ the latter possibility is excluded,
whence the result.
\end{proof}

\begin{lemm}\label{idesy}
(i) Select $j\in\{1,\dots,n\}$. Then we have a direct sum decomposition 
\begin{equation}\label{cosyinw2}
\mathcal{W} = \bigoplus_{k:k\neq j}\mathcal{W}_j\oplus\mathcal{H}\ .
\end{equation}
(ii) The space $\mathcal{H}$ is a subspace of $\mathcal{W}$. Moreover,
\begin{equation}\label{cosyinw}
\mathcal{H}=\mathcal{W}\cap \U(n)\ .
\end{equation}
(iii) The symmetric group $\mathbb{S}_n$ acts on the ring $\Un$ and on the space $\mathcal{W}$ by permutations of the variables $\th_1,\dots,\th_n$ and
\begin{equation}\label{cosyinwz}
\mathcal{H}=\mathcal{W}^{\, \mathbb{S}_n}\ ,
\end{equation}
where $\mathcal{W}^{\, \mathbb{S}_n}$ denotes the subspace of $\mathbb{S}_n$-invariants in $\mathcal{W}$.
\end{lemm}
\begin{proof}
(i) Select $j\in\{1,\dots,n\}$ and suppose that
\begin{equation}\label{interWj}
\sum_{k \neq j} \frac{\pi_k(\th_k)}{\chi_k} + p= \sum_{k \neq j} \frac{\tilde{\pi}_k(\th_k)}{\chi_k} + \tilde{p}
\end{equation}
with $\pi_k(t), \tilde{\pi}_k(t) \in \K[t], k \neq j$ and $p, \tilde{p} \in \mathcal{H}$. The equation \eqref{interWj} is equivalent to
\begin{equation}\label{equaproofidesy1}
\sum_{k \neq j} \frac{\pi_k(\th_k)-\tilde{\pi}_k(\th_k)}{\chi_k} + (p-\tilde{p}) = 0 \ .
\end{equation}
The left hand side of \eqref{equaproofidesy1} is a partial fraction decomposition with respect to $\th_j$. It implies that $\pi_k(\th_k) = \tilde{\pi}_k(\th_k)$ for all $k \neq j$ and $p = \tilde{p}$. This implies the directness of the sum
$\bigoplus_{k:k\neq j}\mathcal{W}_j\oplus\mathcal{H}$
which by Lemma \ref{idesy0} is equal to $\mathcal{W}$. The propositions (ii) and (iii) can be deduced from (i).
\end{proof}

Let $X$ be an auxiliary indeterminate. We have a linear map of vector spaces
\begin{equation}\label{mappoln}\K[t]^n \to \mathcal{W}\end{equation}
defined by
\begin{equation}\label{formlinearmapK[t]W12}
(\pi_1(t),\dots,\pi_n(t))\mapsto \sum_{j=1}^n\, \frac{\pi_j(\th_j)}{\chi_j} \ .
\end{equation}
This map is surjective by definition of $\mathcal{W}$. It follows from Corollary \ref{coroidesy} that its kernel is the vector subspace of $\K[t]^n$ spanned by $n$-tuples $(t^j,\dots,t^j)$ for $j=0,1,\dots,n-2$. Lemma \ref{idesy0} implies that the image of the diagonal in $\K[t]^n$, formed by $n$-tuples $(\pi(t),\dots,\pi(t))$, is the space $\mathcal{H}$. 
We have the similar surjective map $\K(t)^n \to \Wb$, given by the same formula \eqref{formlinearmapK[t]W12}. 
The next corollary says that it has the same kernel as the map \eqref{mappoln}.

\begin{coro}\label{coroidesy}
For any $(\pi_1(t),\dots,\pi_n(t)) \in \K(t)^n$,
\begin{equation}\label{kernelunivratiofunct12}
\sum_{j=1}^n \, \frac{\pi_j(\th_j)}{\chi_j} = 0
\end{equation}
if and only if all $\pi_j(t)$ are equal to the same polynomial of degree less or equal to $n-2$.
\end{coro}
\begin{proof}
We use Proposition \ref{WbandWclink}, Lemma \ref{idesy} (iii) and Lemma \ref{idesy0}. 
\end{proof}

\paragraph{Potential.}
We give a general solution of the system \eqref{eqsigib}. We need the following lemma.

\begin{lemm}\label{lemDDf=0}
For any $i,j = 1,\dots,n$, $i\neq j$, and any $f \in \K(n)$,
$$\Delta_i \Delta_j (f) = 0 \quad \Longleftrightarrow \quad f = f_i + f_j$$
where $f_i, f_j \in \K(n)$ are respectively independent of $\th_i$ and $\th_j$.
\end{lemm}
\begin{proof}
If $\Delta_i \Delta_j (f) = 0$ then $g = \Delta_j (f)$ is independent of $\th_i$. Consider $\mu_i \in \K$ such that the evaluation $f \vert_{\th_i=\mu_i}$ is well defined, we write $f_i = f \vert_{\th_i=\mu_i}$. Then $g = g\vert_{\th_i=\mu_i} = \Delta_j (f \vert_{\th_i=\mu_i}) = \Delta_j (f_i)$ so
$\Delta_j (f - f_i) = g-g = 0$
and writing $f_j = f - f_i$ gives the result.
\end{proof}

If $f\in\Un$ and $\Delta_i \Delta_j (f) = 0$ then one can choose $f_i, f_j \in \Un$, see Lemma 17 of \cite{HO2}. The proof uses the partial decomposition in $\Un$ following  from the formula
\begin{equation}\label{partifracdeco112}
\frac{1}{(\th_{il}+k)(\th_{jl}+k')} = \frac{1}{\th_{ij}+k-k'} \left( \frac{1}{\th_{jl}+k'} - \frac{1}{\th_{il}+k} \right) \ , \ i,j,l=1,\dots,n \ , \ k,k' \in \mathbb{Z} \ .
\end{equation}

\begin{lemm}\label{exiopo}
Let $\sigma_1,\dots,\sigma_k$, $k\leq n$, be a $k$-tuple of elements in $\K(n)$ such that 
\begin{equation}\label{sepr4}
\Delta_i(\sigma_j)=\Delta_j (\sigma_i)\ ,\ i,j=1,\dots,k\ .\end{equation}
Assume that $\sigma_i$ belongs to the image of $\Delta_i\colon\K(n)\to\K(n)$ for all $i=1,\dots,k$. Then there exists a potential $f\in \K(n)$ such that 
\begin{equation}\label{sepr5}
\sigma_i = \Delta_i (f)\ ,\ i=1,\dots,k\ .
\end{equation}  
\end{lemm}
\begin{proof} For $k=1$ there is nothing to prove. Let now $k>1$. We use the induction in $k$. By the induction hypothesis, there exist elements $F,G \in \K(n)$ 
such that
$$\sigma_i = \Delta_i(F)\ \text{for}\ i=1,3,\dots,n\ \text{and}\ \sigma_j = \Delta_j(G)\ \text{for}\ j=2,3,\dots,n\ .$$
Then 
\[\Delta_l(F)=\Delta_l(G) \ \text{for}\ l=3,\dots,n\ \text{and}\ \Delta_1\Delta_2(G)=\Delta_2\Delta_1(F)\ .\]
The element $F-G$ does not depend on $\th_l$, $l=3,\dots,n$, and $\Delta_1\Delta_2(F-G)=0$. According to Lemma \ref{lemDDf=0}, there exist 
two elements $u$ and $v$ such that $u$ depends only on $\th_1$,  $v$ depends only on $\th_2$, and $F-G=u-v$. Then $f:=F+v=G+u$
is the desired potential.
\end{proof}

Lemma \ref{exiopo} holds for $\sigma_i\in\Un$ instead of $\K(n)$, see Lemma 19 of \cite{HO2}.
 
\begin{prop}\label{propotentialb}
Assume that the elements $\sigma_1,\dots,\sigma_n \in \K(n)$ satisfy the system \eqref{eqsigib}. Then there exists 
an element $\sigma \in \K(n)$ such that 
\begin{equation}\label{potentialb}
\sigma_i = \Delta_i \sigma\ , \ i=1,\dots,n \ .
\end{equation}
\end{prop}
\begin{proof}
The system \eqref{eqsigib} is equivalent to
$\sigma_j = \Delta_j \left[ (\th_{ji}+1) \sigma_i \right]$, $i \neq j \in \{1,\dots,n\}$.
So $\sigma_j$ is in the image of $\Delta_j$. Also, the system \eqref{eqsigib} implies that
$\Delta_i \sigma_j = \frac{\sigma_i-\sigma_j}{\th_{ij}} = \frac{\sigma_j-\sigma_i}{\th_{ji}} = \Delta_j \sigma_i$, $i \neq j \in \{1,\dots,n\}$.
Lemma \ref{exiopo} concludes the proof. 
\end{proof}

By Propositions \ref{equasigmasb} and \ref{propotentialb}, if the algebra $\Diff(\sigma_1,\dots,\sigma_n)$ has the PBW property then an element 
$\sigma$ such that $\sigma_i = \Delta_i \sigma , i=1,\dots,n$, exists.
We shall call the element $\sigma$ the ``potential"
and write $\Diffs(n) := \Diff(\sigma_1,\dots,\sigma_n)$.

\subsection{Solutions of $\Delta$-system in $\K(n)$}
In this section, we use the notation $\lb$ for any $(\lambda_1,\dots,\lambda_n) \in \K^n$ and
$$\chi_k := \prod_{\mtop{s=1}{s\neq k}}^n \th_{ks} \quad , \quad \chim{k}{i} := \prod_{\mtop{s=1}{s\neq k,i}}^n \th_{ks} \quad \text{and} \quad \chim{k}{i,j} := \prod_{\mtop{s=1}{s\neq k,i,j}}^n \th_{ks}$$
for any $i,j \in \nb$. If not otherwise specified, all sums and products run from $1$ to $n$.

\vskip .1cm
The solutions $\sigma \in \K(n)$ of the $\Delta$-system are given in Theorem \ref{theosolutionsystemequation}. The proof is done by induction on $n$. First, we prove three lemmas which link a solution $\sigma \in \K(n)$ of the $\Delta$-system to its evaluations $\sigma |_{\th_i=\lambda_i}, \sigma |_{\th_j=\lambda_j}$ and $\sigma |_{\mtop{\th_i=\lambda_i}{\th_j=\lambda_j}}$ with $(\lambda_1,\dots,\lambda_n) \in \text{Dom}(\sigma)$. After proving Theorem \ref{theosolutionsystemequation}, the Corollary \ref{corosolsystrestrict} gives all the solutions of the $\Delta$-system which are in $\Un$.

\vskip .1cm
\begin{lemm}\label{sigmainduction}
Let $\sigma \in \K(n)$ be a solution of the $\Delta$-system. For any $\lb \in \dom(\sigma)$ and $i \in \nb$, the element $\sigma |_{\th_i=\lambda_i} \in \K(n)$ is a solution of the $\Delta$-system in $\K(n-1)$.
\end{lemm}
\begin{proof}
Since for any $f \in \K(n)$ and $k \in \nb \setminus \{i\}$ : $\left( \Delta_k f \right)\vert_{\th_i=\lambda_i} =  \Delta_k \left(f\vert_{\th_i=\lambda_i} \right)$.
\end{proof}

\begin{lemm}\label{sigmaforminduction}
For any $\sigma \in \K(n)$, $\lb \in \dom(\sigma)$ and $i \neq j$ in $\nb$, the equation $\Delta_i \Delta_j (\th_{ij}\sigma) = 0$ is equivalent to
$$\sigma = \frac{1}{\th_{ij}} \left[ (\lambda_i-\th_j) \sigma |_{\th_i=\lambda_i} + (\th_i-\lambda_j) \sigma |_{\th_j=\lambda_j} - (\lambda_i-\lambda_j) \sigma \vv{\mtop{\th_i=\lambda_i}{\th_j=\lambda_j}} \right] \ .$$
\end{lemm}
\begin{proof}
By successive evaluations at $\th_i=\lambda_i$ and $\th_j=\lambda_j$, we have
$$\begin{array}{rcl}
\Delta_i \Delta_j (\th_{ij} \sigma) = 0 &\Longleftrightarrow & \Delta_j (\th_{ij} \sigma) = \left. \left[ \Delta_j (\th_{ij} \sigma) \right] \right\vert_{\th_i=\lambda_i} = \Delta_j \left[(\lambda_i-\th_j) \left. \sigma  \right\vert_{\th_i=\lambda_i} \right] \\
& \Longleftrightarrow & \Delta_j \left[\th_{ij} \sigma - (\lambda_i-\th_j) \left. \sigma  \right\vert_{\th_i=\lambda_i} \right] = 0 \\
& \Longleftrightarrow & \th_{ij} \sigma - (\lambda_i-\th_j) \left. \sigma \right\vert_{\th_i=\lambda_i} = \left. \left[\th_{ij} \sigma - (\lambda_i-\th_j) \left. \sigma  \right\vert_{\th_i=\lambda_i} \right] \right\vert_{\th_j=\lambda_j} \\
& \Longleftrightarrow & \th_{ij} \sigma - (\lambda_i-\th_j) \left. \sigma  \right\vert_{\th_i=\lambda_i} = (\th_i-\lambda_j) \left. \sigma \right\vert_{\th_j=\lambda_j} - (\lambda_i-\lambda_j) \sigma \vv{\mtop{\th_i=\lambda_i}{\th_j=\lambda_j}} 
\end{array}$$
which gives the result.
\end{proof}
\begin{lemm}\label{sigmaforminduction2}
We fix $i \neq j$ in $\nb$. Let $\sigma \in \K(n)$ be such that 
$$\sigma|_{\th_i=\lambda_i} = \sum_{k \neq i} \frac{\alpha_k(\th_k)}{\chim{k}{i}} \ \  \text{and} \ \  \sigma|_{\th_j=\lambda_j} = \sum_{k \neq j} \frac{\beta_k(\th_k)}{\chim{k}{j}}\ \ \text{for some}\ \ \alpha_k(t),\beta_k(t)\in\K(t) \ .$$
Then one can rewrite $\sigma|_{\th_j=\lambda_j} $ in the form 
\begin{equation}\label{formsigmaeval2}
\sigma|_{\th_j=\lambda_j} = \sum_{k \neq j} \frac{\beta'_k(\th_k)}{\chim{k}{j}}\ ,
\end{equation}
where $ \beta'_k(t)\in\K(t) $ are such that 
\begin{equation}\label{formsigmaeval3}
\sigma\vv{\mtop{\th_i=\lambda_i}{\th_j=\lambda_j}} = \sum_{k \neq i,j} \frac{\gamma_k(\th_k)}{\chim{k}{i,j}} \ \ \text{with} \ \ \gamma_k(\th_k) = \frac{\alpha_k(\th_k)-\alpha_j(\lambda_j)}{\th_k-\lambda_j} =  \frac{\beta'_k(\th_k)-\beta'_i(\lambda_i)}{\th_k-\lambda_i} \ ,k\neq i,j\ .
\end{equation}
\end{lemm}
\begin{proof}
From Lemma \ref{idesy0}, we have
$$\sigma|_{\th_i=\lambda_i} = \sum_{k \neq i,j} \frac{\alpha_k(\th_k)-\alpha_j(\th_j)}{\chim{k}{i}} \quad \text{and} \quad \sigma|_{\th_j=\lambda_j} = \sum_{k \neq i,j} \frac{\beta_k(\th_k)-\beta_i(\th_i)}{\chim{k}{j}} \ .$$
It follows that
$$\sigma\vv{\mtop{\th_i=\lambda_i}{\th_j=\lambda_j}} = \sum_{k \neq i,j} \left[ \frac{\alpha_k(\th_k)-\alpha_j(\lambda_j)}{\th_k-\lambda_j}\right] \frac{1}{\chim{k}{i,j}}  = \sum_{k \neq i,j} \left[ \frac{\beta_k(\th_k)-\beta_i(\lambda_i)}{\th_k-\lambda_i}\right] \frac{1}{\chim{k}{i,j}} \ .$$
From Corollary \ref{coroidesy}, there exists a polynomial $p$ of degree less or equal to $n-4$ such that
$$\frac{\alpha_k(\th_k)-\alpha_j(\lambda_j)}{\th_k-\lambda_j}
 = \frac{\beta_k(\th_k)-\beta_i(\lambda_i)}{\th_k-\lambda_i} + p(\th_k) \quad , \quad k \neq i,j \ .$$
By defining $\beta'_k(\th_k) = \beta_k(\th_k) + (\th_k-\lambda_i) p(\th_k)$, $k \neq j$, we get
\begin{equation}\label{elementpibar11}
\frac{\alpha_k(\th_k)-\alpha_j(\lambda_j)}{\th_k-\lambda_j}
 = \frac{\beta'_k(\th_k)-\beta'_i(\lambda_i)}{\th_k-\lambda_i} \quad , \quad k \neq i,j \ .
\end{equation}
By defining $\gamma_k(\th_k)$ as the element \eqref{elementpibar11}, $k \neq i,j$, we get \eqref{formsigmaeval3}. The polynomial $(t-\lambda_i)p(t)$ is of degree at most $n-3$ so Corollary \ref{coroidesy} implies equation \eqref{formsigmaeval2}.
\end{proof}

\begin{theo}\label{theosolutionsystemequation}
The space of solutions of the $\Delta$-system in $\K(n)$ is $\Wb$.
\end{theo}
\begin{proof}
We prove by induction on $n$ that a solution of the $\Delta$-system in $\K(n)$ is in $\Wb$. The $\Delta$-system in $\K(1)$ is trivial.
Consider a solution $\sigma \in \K(n)$ of the $\Delta$-system. We fix $i$ and $j$, $i \neq j$, in $\nb$. Lemma \ref{sigmaforminduction} implies that for any $\lb \in \dom(\sigma)$, we have
$$\sigma = \frac{1}{\th_{ij}} \left[ (\lambda_i-\th_j) \sigma |_{\th_i=\lambda_i} + (\th_i-\lambda_j) \sigma |_{\th_j=\lambda_j} - (\lambda_i-\lambda_j) \sigma \vv{\mtop{\th_i=\lambda_i}{\th_j=\lambda_j}} \right] \ .$$
From Lemma \ref{sigmainduction}, $\sigma |_{\th_i=\lambda_i}$ and $\sigma |_{\th_j=\lambda_j}$ are solutions of $\Delta$-systems in $\K(n-1)$. By the induction hypothesis and Lemma \ref{sigmaforminduction2}, we have
$$\sigma |_{\th_i=\lambda_i} = \sum_{k \neq i} \frac{\alpha_k(\th_k)}{\chim{k}{i}} \quad , \quad \sigma |_{\th_j=\lambda_j} = \sum_{k \neq j} \frac{\beta'_k(\th_k)}{\chim{k}{j}} \quad \text{and} \quad \sigma \vv{\mtop{\th_i=\lambda_i}{\th_j=\lambda_j}} = \sum_{k \neq i,j} \frac{\gamma_k(\th_k)}{\chim{k}{i,j}}$$
for some $\alpha_k(t),\beta'_k(t),\gamma_k(t)\in\K(t)$ such that
$$\gamma_k(\th_k) = \frac{\alpha_k(\th_k)-\alpha_j(\lambda_j)}{\th_k-\lambda_j} \ \text{and} \ \beta'_k(\th_k) = \frac{\th_k-\lambda_i}{\th_k-\lambda_j} \left(\alpha_k(\th_k)-\alpha_j(\lambda_j)\right) + \beta'_i(\lambda_i) \ , \ k \neq i,j \ .$$
It implies that
\begin{align*}
\sigma & = \frac{1}{\th_{ij}} \left[ (\lambda_i-\th_j) \sum_{k \neq i} \frac{\alpha_k(\th_k)}{\chim{k}{i}} + (\th_i-\lambda_j) \sum_{k \neq j} \frac{\beta'_k(\th_k)}{\chim{k}{j}} - (\lambda_i-\lambda_j) \sum_{k \neq i,j} \frac{\gamma_k(\th_k)}{\chim{k}{i,j}} \right] \\
& = \frac{1}{\th_{ij}} \left[ (\lambda_i-\th_j) \sum_{k \neq i} \frac{\alpha_k(\th_k)}{\chim{k}{i}} +(\th_i-\lambda_j) \frac{\beta'_i(\th_i)}{\chim{i}{j}} \right. \\
 & ~~~~~~~~~~~~~~~~ + (\th_i-\lambda_j) \sum_{k \neq i,j} \frac{1}{\chim{k}{j}} \left( \frac{\th_k-\lambda_i}{\th_k-\lambda_j} \left(\alpha_k(\th_k)-\alpha_j(\lambda_j)\right) + \beta'_i(\lambda_i) \right) \\
 & ~~~~~~~~~~~~~~~~~~~~~~~~~~~~~~~~ \left. - (\lambda_i-\lambda_j) \sum_{k \neq i,j} \frac{1}{\chim{k}{i,j}} \left( \frac{\alpha_k(\th_k)-\alpha_j(\lambda_j)}{\th_k-\lambda_j} \right) \right] \\
\quad & = \frac{1}{\th_{ij}} \left[ (\lambda_i-\th_j) \frac{\alpha_j(\th_j)}{\chim{j}{i}} +(\th_i-\lambda_j) \frac{\beta'_i(\th_i)}{\chim{i}{j}} + (\th_i-\lambda_j) \beta'_i(\lambda_i) \sum_{k \neq i,j} \frac{1}{\chim{k}{j}} \right. \\
 & ~~~~~~~~~~~~~~~~ + \sum_{k \neq i,j} \frac{\alpha_k(\th_k)}{\chi_k} \left( (\lambda_i-\th_j)\th_{ki}+(\th_i-\lambda_j)\th_{kj} \frac{\th_k-\lambda_i}{\th_k-\lambda_j} - \th_{ki} \th_{kj} \frac{\lambda_i-\lambda_j}{\th_k-\lambda_j} \right) \\
 & ~~~~~~~~~~~~~~~~~~~~~~~~~~~~~~~~ + \left. \left( \sum_{k \neq i,j} \left( \frac{\lambda_i-\lambda_j}{\th_k-\lambda_j} \th_{ki} - (\th_i-\lambda_j) \frac{\th_k-\lambda_i}{\th_k-\lambda_j} \right) \frac{1}{\chim{k}{j}} \right) \alpha_j(\lambda_j) \right] \ .
\end{align*}
Using Lemma \ref{idesy} and the equalities
$$(\lambda_i-\th_j)\th_{ki}+(\th_i-\lambda_j)\th_{kj} \frac{\th_k-\lambda_i}{\th_k-\lambda_j} - \th_{ki} \th_{kj} \frac{\lambda_i-\lambda_j}{\th_k-\lambda_j} = \th_{ij}(\th_k-\lambda_i) \ ,$$
$$\frac{\lambda_i-\lambda_j}{\th_k-\lambda_j} \th_{ki} - (\th_i-\lambda_j) \frac{\th_k-\lambda_i}{\th_k-\lambda_j} = \lambda_i-\th_i \ ,$$
we get
$$\sigma = \sum_{k \neq i} (\th_k-\lambda_i)\frac{\alpha_k(\th_k)}{\chi_k} + \frac{1}{\chi_i} \left[(\th_i-\lambda_j)\left(\beta'_i(\th_i)-\beta'_i(\lambda_i)\right) + (\th_i-\lambda_i) \alpha_j(\lambda_j) \right] \ .$$
Therefore $\sigma$ is in $\Wb$.

If $\sigma \in \Wb$ then for any $i \neq j$ in $\nb$, we have
$$\th_{ij} \sigma = \sum_{k=1}^n \frac{(\th_{ik}+\th_{kj}) \pi_k(\th_k)}{\chi_k} = - \underbrace{\left( \frac{\pi_j(\th_j)}{\chim{j}{i}} + \sum_{k \neq i,j} \frac{\pi_k(\th_k)}{\chim{k}{i}} \right)}_{\text{independent of}\ \th_i} + \underbrace{\left( \frac{\pi_i(\th_i)}{\chim{i}{j}} + \sum_{k \neq i,j} \frac{\pi_k(\th_k)}{\chim{k}{j}} \right)}_{\text{independent of}\ \th_j}$$
which concludes the proof.
\end{proof}

\begin{coro}\label{corosolsystrestrict}
In $\Un$, the subspace of solutions of \eqref{eqsigibfopo2} is $\Wc$ (of Definition \ref{vespmj}).
\end{coro}
\begin{proof}
It follows from Theorem \ref{theosolutionsystemequation} and Proposition \ref{WbandWclink}.
\end{proof}

\subsubsection{A characterization of polynomial potentials}
The polynomial potentials $\sigma\in\mathcal{W}$ can be characterized in different terms.  
The rings $\Diff(n)$ and $\overline{AZ}_n$ admit the action of Zhelobenko automorphisms $\q_1,\dots,\q_{n-1}$ \cite{Zh2,KO1}. Their action on the generators 
$\Z^i$ and $\der_i$, $i=1,\dots,n$, is given by equations \eqref{automq}.

\begin{lemm}\label{posoazhe} The ring $\Diffs(n)$ admits the action of Zhelobenko automorphisms if and only if $\sigma$ is a polynomial, 
$\sigma\in\mathcal{H}$.
\end{lemm}
\begin{proof}
By definition of $\Diffs(n)$, the element $\sigma \in \Un$ is in $\mathcal{W}$. The action of the Zhelobenko automorphisms on $\mathcal{W}$ is the action of the symmetric group $\mathbb{S}_n$ therefore by (iii) of Lemma \ref{idesy}, it implies that $\sigma$ is in $\mathcal{H}$.
\end{proof}

In the examples discussed in \textbf{Section \ref{twoexse}}, the ring $\Diff(n)$ corresponds to $\sigma=H_1$ and the ring $\overline{AZ}_n$ corresponds 
to $\sigma=-H_2=-\sum_{i,j:i\leq j}\th_i\th_j$, 
$\Delta_i H_2 =\th_i+\sum_{k=1}^n \th_k-1$.

\chapter{Rings \texorpdfstring{$\Diffs(n)$}{section1}}\label{ChapterRepreDiff}

\vskip -1.16cm
In \textbf{Chapter \ref{GeneralizedDiff}} we described the rings $\Diffs(n)$ satisfying the PBW property. The associated graded ring gr$\Diffs(n)$, relative to the filtration \eqref{filtrationDiff}, is a Noetherian Ore domain, see  \cite{KO4}. and so is the ring gr$\Diff(n)$. 
It follows, see \cite{GK}, that the ring $\Diffs(n)$ is a Noetherian Ore domain. We construct an analogue of the isomorphism $\mu$ of \textbf{Chapter \ref{ChapterRinghdef}} for the ring $\Diffs(n)$. 

\vskip .04cm
In \textbf{Section \ref{StructureofDiff}} we describe the analogue of the isomorphism $\mu$ \eqref{fromwtodiff}, determine the center of $\Diffs(n)$ and translate, with the help of $\mu$, the action of the symmetric group and, respectively, the braid group  to the ring $\Diffs(n)$ and to a localization of the ring $\text{W}_n$.

\vskip .04cm
In \textbf{Section \ref{InfiDimRepre}} we describe two families of infinite dimensional representations of $\Diffs(n)$. We then compute the 
central characters of the lowest weight representations.

\vskip .04cm
\textbf{Section \ref{FiniteDimModules}} is an introduction to finite dimensional $\Diffs(n)$-modules. We show that an irreducible $\Diffs(n)$-module has a highest weight $\vec{\lambda}$ and derive necessary and sufficient conditions for the irreducibility. We conjecture that irreducible $\Diffv{H_m}(n)$-modules exist if only if $m > n$
and check it  for $n = 1$ and $n=2$. Finally, we present examples of indecomposable $\Diffs(n)$-modules. 

\section{Structure of $\Diffs(n)$}\label{StructureofDiff}
\subsection{Quadratic central elements}

In \cite{HO1}, we have described the center of the ring $\Diff(n)$. The center of the ring $\Diffs(n)$, admits a similar description. Let
$\Gamma_i:=\der_i\Z^i\ \ \text{for}\ \ i=1,\dots,n$, as in Lemma \ref{qga} 
and
\begin{equation}\label{gefufoce}
c(t)=\sum_{i=1}^n \frac{e(t)}{1+\th_{i}t} \Gamma_i - \rho(t)=\sum_{k=1}^n c_k t^{k-1} \ ,
\end{equation}
where $t$ is an auxiliary variable and $\rho(t)$ a polynomial of degree $n-1$ with coefficients in $\Ub(n)$. 
\begin{prop}\label{lece} (i) Let $\sigma \in \bar{\mathcal{W}}$ and $\sigma_j=\Delta_j\sigma$, $j=1,\dots,n$. 
The elements $c_1,\dots,c_n$ are central in the ring $\Diffs(n)$ if and only if the polynomial $\rho(t)$ satisfies the system of finite-difference equations
\begin{equation}\label{corho}
\Delta_j \rho(t) = \frac{e(t)}{1+\th_j t} \sigma_j \ .
\end{equation}
(ii) For an arbitrary $\sigma \in \Wc$, the system \eqref{corho} admits a solution, unique modulo an element of $\K[t]$. 
The solutions of the system \eqref{corho} are of the form
\begin{equation}\label{soforho}
\rho(t)= \sum_{k=1}^n \frac{e(t)}{1+\th_k t} \frac{\pi_k(\th_k)}{\chi_k} + p(t)\ \ \text{where $p(t) \in \K[t]$ of degree $n-1$.}
\end{equation}
\end{prop}
\begin{proof}
(i) To analyze the relation $x^j c(t)-c(t)x^j=0$, we shall write the expression $x^j c(t)-c(t)x^j$ in the ordered form, in the order $\der\Z\Z$. 
The element
\begin{equation}\label{gefufoce2}
c_0(t)=\sum_i\frac{e(t)}{1+\th_{i}t }\Gamma_i
\end{equation}
is central in the homogeneous ring $\Diffv{0}(n)$, see the proof of Proposition \ref{quceel}. Hence we have to track only those ordered terms whose filtration degree is smaller than 3. As in \textbf{Chapter \ref{GeneralizedDiff}}, we use the symbol $u\;\rule[-3.5mm]{.28mm}{7mm}_{\;\text{l.d.t.}}$ to denote these lower degree terms in an expression $u$. We have
\[(x^j c(t)-c(t)x^j)\;\rule[-3.5mm]{.28mm}{7mm}_{\;\text{l.d.t.}}=\left(-\frac{e(t)}{1+\th_j t}\sigma_j-\rho(t)[-\varepsilon_j]+\rho(t)\right) x^j\ .\]
Thus the element $c(t)$ commutes with the generators $x^j$, $j=1,\dots,n$, if and only if the polynomial $\rho(t)$ satisfies the system \eqref{corho}. 
The use of the anti-automorphism \eqref{antiautomb} shows that the element $c(t)$ then commutes with the generators $\der_j$, $j=1,\dots,n$, as well.

\vskip .2cm
\noindent
(ii) Let $\sigma = \sum_{k=1}^n A_k$ where 
$A_k = \frac{\pi_k(\th_k)}{\chi_k}$ and $\rho_k(t) = \frac{e(t)}{1+\th_k t} A_k$.
For any $k \in \{1,\dots,n\}$,
$\Delta_k (\rho_k(t)) = \frac{e(t)}{1+\th_k t} \Delta_k (A_k)$
and for any $j \neq k$,
\begin{equation}\label{proofcorho1}
\Delta_j (A_k) = \Delta_j \left( \frac{\th_{kj} A_k}{\th_{kj}} \right) = \Delta_j \left( \frac{1}{\th_{kj}} \right) \th_{kj} A_k = \frac{A_k}{\th_{kj}+1} \ .
\end{equation}
According to the formula \eqref{proofcorho1}, we have
$$\begin{array}{rcl}
\Delta_j\rho_k(t) &=& \displaystyle{ 
\frac{e(t)}{(1+\th_k t)(1+\th_j t)} \Delta_j \left((1+\th_j t) A_k\right) } \\[1em]
&=& \displaystyle{ \frac{e(t)}{(1+\th_k t)(1+\th_j t)} \left( tA_k+(1+(\th_j-1)t) \Delta_j A_k \right) } = \displaystyle{ \frac{e(t)}{1+\th_j t}\Delta_j A_k \ .}
\end{array}\ ,$$
and \eqref{soforho} follows.
\end{proof}

\subsection{Center of $\Diffs(n)$}

Let $\sigma=\sum_k  \frac{\pi_k(\th_k)}{\chi_k}\in\Wc$ and 
$\rho(t) = \sum_{k=1}^n \rho_k(t)$ with $\rho_k(t) = \frac{e(t)}{1+\th_k t} \frac{\pi_k(\th_k)}{\chi_k}$, as in Proposition \ref{corho}.
Similarly to \textbf{Section \ref{quadce}}, we have the following lemmas.

\begin{lemm}\label{gammaandc2}
In $\Diffs(n)$, for any $j=1,\dots,n$, we have
$$\Gamma_j = \frac{1}{\chi_j} \left( \th_j^{n-1} c(-\th_j^{-1}) + \pi_j(\th_j) \right) = \sum_{k=1}^n (V^{-1})_j^k c_k + \frac{\pi_j(\th_j)}{\chi_j}$$
with $V^{-1}$ given by \eqref{defV-1matrix11} and $c(t)$ by \eqref{gefufoce}.
\end{lemm}
\begin{lemm}\label{zeroweightelementscenter2}
The ring $\Diffs^0(n)$ is freely generated over $\Un$ by the (commutative) elements $c_1,\dots,c_n$ defined by \eqref{gefufoce}, i.e.
\begin{equation}\label{zeroweightsubring2}
\Diffs^0(n) = \Un[c_1,\dots,c_n] \ .
\end{equation}
\end{lemm}
\begin{proof}
It repeats the proof of Lemma \ref{zeroweightelementscenter} with the use of Lemma \ref{gammaandc2} instead of \ref{gammaandc}.
\end{proof}
\begin{prop}
The center of the ring 
$\Diffs(n)$ is $\K[c_1,\dots,c_n]$. 
\end{prop}
\begin{proof}
It repeats the proof of Proposition \ref{isocenter} with the use of Lemma \ref{zeroweightelementscenter2} instead of \ref{zeroweightelementscenter}.
\end{proof}

\subsection{Ring of fractions of $\Diffs(n)$}

In \textbf{Section \ref{isowe}}, we have established an isomorphism between certain localizations of the ring $\Diff(n)$ and an extention of the Weyl algebra $\text{W}_n$. 
It turns out that when we pass to the analogous localization of the ring $\Diffs(n)$, we loose the information about the potential $\sigma$.

Similarly to \textbf{Section \ref{isowe}}, let $\text{S}_n$ be the multiplicative set generated by $\Z^j$, $1\leq j\leq n$, and $\th_{ij}+k$, $1\leq i<j\leq n$, $k \in \mathbb{Z}$. The set $\text{S}_n$ satisfies left and right Ore conditions. Denote by $\text{S}_n^{-1}\Diffs(n)$ the localization of the ring $\Diffs(n)$ with respect to $\text{S}_n$. Let $\{a_1,\dots,a_n\}$ be a set of commuting variables and $T^{-1}\text{W}_n$ the localization of the Weyl algebra defined in \textbf{Section \ref{WeylAlgebraSec}}. The following theorem is the analogue, for $\Diffs(n)$, of Theorem \ref{isoloc}.

\begin{prop}\label{isoloc2}
The ring $\text{S}_n^{-1}\Diffs(n)$ is isomorphic to the ring $\K[a_1,\dots,a_n] \otimes \text{T}^{-1}\text{W}_n$.
\end{prop}
By Proposition \ref{isoloc2}, the localizations $\text{S}_n^{-1}\Diffs(n)$ are isomorphic for all non-zero $\sigma \in \Wc$. Nevertheless, the rings $\Diffs(n)$ are not isomorphic as filtered rings over $\Un$.
Let $\Wc'=\bigoplus_{j\neq 1}\Wc_j \oplus\mathcal{H}'$, see Definition \ref{vespcoh}.
\begin{lemm}\label{leis2}
Let $\sigma,\sigma'\in\Wc'$.
The rings $\Diffs(n)$ and $\text{Diff}_{\h,\sigma'}(n)$ are isomorphic, as filtered rings over $\Ub(n)$ (the filtration is defined by \eqref{filtrationDiff}), 
if and only if 
$ \sigma'=\gamma\sigma\ \ \text{for some}\ \ \gamma\in\K^*$.
\end{lemm}
\begin{proof} Assume that $\iota\colon \Diffs(n)\to\text{Diff}_{\h,\sigma'}(n)$ is an isomorphism of filtered rings over $\Ub(n)$. To distinguish the generators, we denote 
the generators of the ring $\text{Diff}_{\h,\sigma'}(n)$ by $\Z'^{i}$ and $\der'_i$. 

\vskip .1cm
The $\varepsilon_i$-weight subspace $\mathfrak{E}_i$ of the ring $\Diffs(n)$ consists of elements of the form $\theta \Z^i$ where $\theta$ is a polynomial in the elements $\Gamma_j$, $j=1,\dots,n$, with coefficients in $\Ub(n)$. 
Since the space of the elements of $\mathfrak{E}_i$ with filtration degree $\leq 1$ 
is $\Ub(n)\Z^i$, we must have
\be\label{priso}\iota\colon \Z^i\mapsto \mu_i\Z'^i\ ,\ \der_i\mapsto\der'_i\nu_i\ee
with some invertible elements $\mu_i,\nu_i\in \Ub(n)$, $i=1,\dots,n$. Let $\gamma_i:=\mu_i\nu_i$, $i=1,\dots,n$. The defining relation \eqref{hdesp10c} and the corresponding relation for the ring $\text{Diff}_{\h,\sigma'}(n)$ shows that the formulas \eqref{priso} may define an isomorphism only if  
\be\label{noha1}\gamma_i=\gamma_j[\varepsilon_j]\ ,\ i,j=1,\dots,n\ ,\ee
and
\be\label{noha2}\gamma_i\sigma_i'=\sigma_i\ ,\ i=1,\dots,n\ .\ee
The condition \eqref{noha1} implies that $\gamma_i=\gamma$ for some $\gamma\in\K$. The condition \eqref{noha2} then becomes $\gamma\sigma_i'=\sigma_i$
and the assertion follows.  
\end{proof}

\subsection{Actions of braid and symmetric groups}\label{Actionbraidsymmetric}
We use the isomorphism $\mu$ to translate the action of the symmetric group from the Weyl algebra $\text{W}_n$ to the ring $\Diffs(n)$ and the action of the braid group from $\Diff(n)$ to $T^{-1} \text{W}_n$.

\vskip .2cm
\noindent{\bf Action of the symmetric group on $\Diffs(n)$.}
The symmetric group $\mathbb{S}_n$ acts by automorphisms on the algebra $\text{W}_n$, $\pi(X^j)=X^{\pi(j)}$, $\pi(D_j)=D_{\pi(j)}$ for $\pi \in \mathbb{S}_n$. The isomorphism $\mu$ translates this action to the action of $\mathbb{S}_n$ on the ring $\text{S}_n^{-1}\Diffs(n)$. It turns out that the subring $\Diffs(n)$ is preserved by this action. We present the formulas for the action of the generators $s_i$, $i=1,\dots,n-1$ of $\mathbb{S}_n$.
\begin{equation}\label{acofsn1}
\begin{array}{l}{\displaystyle 
s_i(\Z^i) =- \Z^{i+1} \th_{i,i+1} \ ,\ s_i(\Z^{i+1}) = \Z^i \frac{1}{\th_{i,i+1}}\ ,\ s_i(\Z^j) =\Z^{j}\ \text{for}\ j\neq i,i+1\ ,}\\[.3em]
{\displaystyle s_i(\der_i) = - \frac{1}{\th_{i,i+1}} \der_{i+1}\ ,\ s_i(\der_{i+1}) = \th_{i,i+1} \der_i\ ,\ s_i(\der_j) =\der_{j}\ \text{for}\ j\neq i,i+1\ ,}\\[.8em]
s_i(\th_j)=\th_{s_i(j)} \ .
\end{array}
\end{equation}
One can also check by direct calculations that the automorphisms $s_i$, $i=1,\dots,n-1$ satisfy the Artin relations.

The operators $s_i':=\epsilon s_i\epsilon$, $i=1,\dots,n-1$ where $\epsilon$ is the anti-automorphism \eqref{antiautomb}, generate the action of the  
symmetric group $\mathbb{S}_n$ by automorphisms as well. The action of the automorphism $s_i'$, $i=1,\dots,n-1$, involves only the element $\th_{i,i+1}$ (as the action of the 
automorphism $s_i$) and is given by 
\begin{equation}
\begin{array}{l}{\displaystyle 
s_i'(\Z^{i,\alpha}) = -  \frac{1}{\th_{i,i+1}}\Z^{i+1,\alpha} \ ,\ s_i'(\Z^{i+1,\alpha}) = \th_{i,i+1}\Z^{i,\alpha}\ ,\ s_i'(\Z^{j,\alpha}) = \Z^{j,\alpha}\ 
\text{for}\ j\neq i,i+1\ ,}\\[.2em]
{\displaystyle 
s_i'(\der_{i,\alpha}) = - \der_{i+1,\alpha} \th_{i,i+1} \ ,\ s_i'(\der_{i+1,\alpha}) =\der_{i,\alpha}\frac{1}{\th_{i,i+1}} \ ,\ 
s_i'(\der_{j,\alpha}) =\der_{j,\alpha}\ \text{for}\ j\neq i,i+1\ ,}\\[.8em]
s_i(\th_j)=\th_{s_i(j)}\ .
\end{array}
\end{equation}

\noindent{\bf Action of the braid group on $T^{-1}\text{W}_n$.}
The isomorphism $\mu$ can also be used to translate the action \eqref{automq} of the braid group by Zhelobenko operators to the 
action of the braid group by automorphisms on the ring $\K[a_1,\dots,a_n] \otimes \text{T}^{-1}\text{W}_n$. It turns out that this action preserves 
the subring $\text{T}^{-1}\text{W}_n$. Moreover, let $\text{T}_0$ be the multiplicative subset of $\text{T}$, generated by $X^jD_j-X^kD_k+\ell$, $1\leq j<k\leq n$, $\ell\in\mathbb{Z}$. Then the action of the operators $\q_i$, $i=1,\dots,n-1$, preserves the subring $\text{T}_0^{-1}\text{W}_n$. We present the formulas for the action of the operators $\q_i$, $i=1,\dots,n-1$:
$$\begin{array}{l}{\displaystyle 
\q_i(X^i) = \frac{1}{\mathcal{H}_{i,i+1}} X^{i+1}\ ,\ \q_i(X^{i+1}) = X^i \mathcal{H}_{i,i+1}\ ,\ \q_i(X^j)=X^j\ \text{for}\ j\neq i,i+1\ ,}\\[.1em] 
{\displaystyle \q_i(D_i) = D_{i+1} \mathcal{H}_{i,i+1}\ ,\ \q_i(D_{i+1}) = \frac{1}{\mathcal{H}_{i,i+1}} D_i\ ,\ \q_i(D_j)=D_j\ \text{for}\ j\neq i,i+1\ .}
\end{array}$$

\section{Infinite dimensional representations}\label{InfiDimRepre}

\subsection{Lowest weight representations}\label{infhighwe} 

The ring $\Diffs(n)$ has an $n$-parametric family of lowest weight representations,
introduced in \cite{HO1,HO2}. 
Let $\mathfrak{D}_n$ be 
the $\Ub(n)$-subring of $\Diffs(n)$ generated by $\{\der_i\}_{i=1}^n$. 
Let $\vec{\lambda}:=\{\lambda_1,\dots,\lambda_n\}\in\K^n$ be such that $\lambda_i-\lambda_j
\notin\mathbb{Z}$ for all $i,j=1,\dots,n$, $i\neq j$. Denote by $M_{\vec{\lambda}}$ the 
one-dimensional $\K$-vector space
with the basis vector $v$. The formulas  
\begin{equation}\label{dmoac}
\th_i.v= \lambda_iv\ ,\ \der_i. v = 0 \ ,\ i=1,\dots,n\ ,
\end{equation}
define the $\mathfrak{D}_n$-module structure on $M_{\vec{\lambda}}$. The lowest weight representation of lowest weight 
$\vec{\lambda}$ is  the induced representation $\text{Ind}_{\mathfrak{D}_n}^{\Diffs(n)}M_{\vec{\lambda}}$.
We describe the values of the central polynomial $c(t)$, see \eqref{gefufoce}, on the lowest weight representations.   

\vskip .2cm 
\begin{prop}\label{anosege}
The element $c(t)$ acts on $\text{Ind}_{\mathfrak{D}_n}^{\Diffs(n)}M_{\vec{\lambda}}$ by multiplication by the scalar 
\begin{equation}\label{reoacec}
\rho(t)[-\varepsilon]\ ,\ \text{where}\ \varepsilon=\varepsilon_1+\dots+\varepsilon_n\ .
\end{equation}
\end{prop}
\noindent{\it Proof.} In this proof, we use the following notation. For any element $A$ in $\Un[c_1,\dots,c_n]$, we denote $A[-\varepsilon_i] := \Z^i A (\Z^i)^{-1}$, $i=1,\dots,n$. For any element $B$ in $\mu^{-1} \left(\Un[c_1,\dots,c_n]\right)$, we denote $B[-\varepsilon_i] := X^i B (X^i)^{-1}$, $i=1,\dots,n$. For any $i=1,\dots,n$, we define $\bar{\Gamma}_i = \Z^i \der_i Q^-_i$ with $Q^-_i$ given by \eqref{defqpm}. From \eqref{gefufoce}, for any $i,k=1,\dots,n$ there exist $f_i(t),\bar{f}_i(t),\bar{\rho}(t) \in \Un[t]$, such that
\begin{equation}\label{c_k2sums}
c(t) \ = \ \sum_i f_i(t) \Gamma_i - \rho(t) \ = \ \sum_i \bar{f}_i(t) \bar{\Gamma}_i - \bar{\rho}(t) \quad , \quad k=1,\dots,n \ .
\end{equation}
Using the formulas \eqref{fromwtodiff} and \eqref{fromdifftow}, we get
\begin{equation}\label{gammaimu11}
\Gamma_i = \mu \circ \mu^{-1} \left( \der_i \Z^i \right) = \frac{\mu(\Upsilon_i)}{\chi_i} \quad , \quad i=1,\dots,n
\end{equation}
and for any $i=1,\dots,n$, we have
\begin{align*}
\bar{\Gamma}_i & = \mu \circ \mu^{-1} \left( \Z^i \der_i Q^-_i \right) = \mu \left( X^i \frac{\Upsilon_i}{\Psi_i \Psi'_i} (X^i)^{-1} \right) Q^-_i = \mu \left( \frac{\Upsilon_i[-\varepsilon_i]}{\Psi_i[-\varepsilon_i] \Psi'_i[-\varepsilon_i]} \right) Q^-_i \\
 & = \mu(\Upsilon_i[-\varepsilon_i]) \frac{Q^-_i}{\chi_i[-\varepsilon_i]} = \frac{\mu(\Upsilon_i [-\varepsilon_i])}{\chi_i} \ .
\end{align*}
So
\begin{equation}\label{calculGammas}
\bar{\Gamma}_i = \frac{\mu(\Upsilon_i) [-\varepsilon_i]}{\chi_i} \quad , \quad i=1,\dots,n \ .
\end{equation}
For any $i=1,\dots,n$, the element $\Upsilon_i$, defined in \textbf{Section \ref{isowe}}, verifies $\Upsilon_i[-\varepsilon_i] = \Upsilon_i[-\varepsilon]$. Also $\chi_i[-\varepsilon] = \chi_i$, $i=1,\dots,n$. The equations \eqref{gammaimu11} and \eqref{calculGammas}, then imply
$$\bar{\Gamma}_i = \Gamma_i [-\varepsilon] \quad , \quad i=1,\dots,n \ .$$
Using equation \eqref{c_k2sums} and that $\Z^1 \dots \Z^n c(t) = c(t) \Z^1 \dots \Z^n$, we have
$$c(t) \ = \ \sum_i f_i(t) \Gamma_i - \rho(t) \ = \ \sum_i \bar{f}_i(t)[\varepsilon] \Gamma_i - \bar{\rho}(t)[\varepsilon] \ .$$
Since the $\Gamma_i$'s are linearly independent in $\Diffs(n)$, we get
$$\bar{\rho}(t) = \rho(t)[-\varepsilon] \quad \text{and} \quad \bar{f}_i(t) = f_i(t)[-\varepsilon] \ , \ i=1,\dots,n \ .$$
By definition of the highest weight vector $v$, we have $\bar{\Gamma}_i v = 0$, $i=1,\dots,n$ and so
$$\hspace{2cm}c(t) v = \left( \sum_i \bar{f}_i(t) \bar{\Gamma}_i - \bar{\rho}(t) \right) v = - \bar{\rho}(t) v = -\rho(t) \vert_{\th_i\mapsto \lambda_i-1} v \ .\hspace{4cm}\square$$

\begin{rmk}\label{remhighestwereps}
Similarly to \eqref{dmoac}, we define highest weight representations $W_{\vec{\lambda}}:=\text{Ind}_{\mathfrak{X}_n}^{\Diffs(n)}M'_{\vec{\lambda}}$, induced from the one-dimensional representation $M'_{\vec{\lambda}}$ : 
\begin{equation}\label{dmoacbb}
\th_i.v= \lambda_iv\ ,\ \Z^i. v = 0 \ ,\ i=1,\dots,n\ ,
\end{equation}
of the $\Ub(n)$-subring $\mathfrak{X}_n$ of $\Diffs(n)$, generated by $\{\Z^i\}_{i=1}^n$.
\end{rmk}

\subsection{Another family of infinite dimensional representations}
The isomorphism \eqref{fromdifftow} allows  to construct another $2n$-parametric family of $\Diff(n)$-modules. 
Let $\vec{\gamma}:=\{\gamma_1,\dots,\gamma_n\}\in\K^n$ 
be such that $\gamma_i-\gamma_j\notin\mathbb{Z}$ if for all $i,j=1,\dots,n$, 
$i\neq j$. 
Let $V_{\vec{\gamma}}$ be the $\K$-vector space with the basis 
$$v_{\vec{j}}:=(X^1)^{j_1+\gamma_1}(X^2)^{j_2+\gamma_2}\dots (X^n)^{j_n+\gamma_n}\ ,\ 
\text{where}\ \vec{j}:=\{ j_1,\dots,j_n\}\ ,\ j_1,\dots,j_n\in\mathbb{Z}\ .$$ 
Under the conditions on $\vec{\gamma}$, $V_{\vec{\gamma}}$ is a $\text{T}^{-1}\text{W}_n$-module. Define the action of $a_k$ on the space $V_{\vec{\gamma}}$ by $a_k\colon v_{\vec{j}} \mapsto A_k v_{\vec{j}}$ with an arbitrary 
$\vec{A}:=\{A_1,\dots,A_n\}\in\K^n$. Then $V_{\vec{\gamma}}$ becomes an 
$\K[a_1,\dots,a_n] \otimes \text{T}^{-1}\text{W}_n$-module and therefore $\Diff(n)$-module which we denote by 
$V_{\vec{\gamma},\vec{A}}$. The central operator $c_k$ acts on $V_{\vec{\gamma},\vec{A}}$ 
by multiplication by the scalar $A_k$.

\section{Finite dimensional $\Diffs(n)$-modules}\label{FiniteDimModules}

\begin{defi}
We call an element $v\neq 0$ of a $\Diffs(n)$-module a lowest weight vector of weight $\vec{\lambda}=(\lambda_1,\dots,\lambda_n)$ if it satisfies \eqref{dmoacbb}.
\end{defi}

\subsection{Irreducible modules}\label{Irred}
Since $\K$ is algebraically closed, any finite dimensional $\Diffs(n)$-module 
$\Mc$ over $\K$ contains a common eigenvector for all $\th_i$, $i=1,\dots,n$. Acting on it by the raising operators $\Z^i$, $\th_i \Z^i = \Z^i (\th_i + 1)$, we shall eventually construct a highest weight vector of some weight $\vec{\lambda}=(\lambda_1,\dots,\lambda_n)$ with $\lambda_i\neq\lambda_j$ for $i\neq j$ since $\th_{ij}+a$, $a\in\mathbb{Z}$, are invertible. 
Clearly $\mathfrak{D}_nv$ is a submodule in $\Mc$ and the action of $\hf$ on it is semisimple. If $\Mc$ is irreducible then the subspace of highest  weight vectors is one-dimensional. We conclude that 
an irreducible module $\Mc$ is a quotient of some $W_{\vec{\lambda}}$, see Remark \ref{remhighestwereps}, by its maximal non-trivial submodule. We call this quotient $L_{\vec{\lambda}}$.

In Section \ref{Irred} the symbol $\Mc$ stands for a quotient of $W_{\vec{\lambda}}$ and $v\in\Mc$ is the image of the generating highest weight vector of $W_{\vec{\lambda}}$. 
We associate to $\Mc$ the sequence $d_{\Mc}=(d_1,\dots,d_n)$ where $d_i$ is the smallest integer such that $\der_i^{d_i+1} v= 0$ and $\der_i^{d_i} v \neq 0$.
\begin{prop}\label{lemmanonzeroirre}
The following conditions are equivalent :

\noindent (i) 
for any $i \in \{1,\dots,n\}$ and $k_i \in \{1,\dots,d_i\}$, 
\begin{equation}\label{equapropnonzeroirre}
\Z^i\,\der_i^{k_i}v= a_{ik_i} \, \der_i^{k_i-1}v \quad \text{with} \quad a_{ik_i} \in \K^\star ;
\end{equation}

\noindent(ii) $\Mc$ is irreducible.
\end{prop}
\begin{proof}
(ii) $\Rightarrow$ (i). The relations \eqref{relshdiffcompdx} and \eqref{hdesp10c} imply that
$\Z^j\, \der_i^{k_i} v = 0$, $j \neq i$, and $\Z^i\, \der_i^{k_i}v= a_{ik_i} \, \der_i^{k_i-1}v$.
If $a_{ik_i}= 0$ then $\der_i^{k_i} v$ is another highest weight vector of $\Mc$ and \eqref{equapropnonzeroirre} follows. 

(i) $\Rightarrow$ (ii). If $\Mc$ is reducible, it contains another highest weight vector $w$ of the form 
$$w =  \der_n^{k_n} \dots \der_1^{k_1} v\ ,\ (0,\dots,0)\neq (k_1,\dots,k_n) \in \mathZ^n\ .$$
By induction, \eqref{equapropnonzeroirre} implies that
$(\Z^i)^{k_i} \der_i^{k_i} v= b_{ik_i} v$ with $b_{ik_i}\neq 0$. Then the relations \eqref{relshdiffcompdx} imply that $(\Z^n)^{k_n} \dots (\Z^1)^{k_1}.w =cv$, 
$c \neq 0$, so $w$ cannot be a highest weight vector.
\end{proof}
\begin{coro}\label{corobasis}
The module $L_{\vec{\lambda}}$ admits a basis 
$$\der_n^{k_n} \dots \der_1^{k_1} v \ , \ 0 \leq k_i \leq d_i \ , \ \ \text{where}\ \ d_{L_{\vec{\lambda}}}=(d_1,\dots,d_n)\ .$$
\end{coro}
We express the irreducibility condition in terms of the highest weight $\vec{\lambda}$. 
\begin{lemm}\label{actiondxmLWVprop}
We have
\begin{equation}\label{actiondxmLWV}
\Z^i \der_i^m v = - \der_i^{m-1} \left( \sum_{k=0}^{m-1} \sigma_i[-k\varepsilon_i] \right) v\ ,\ i \in \{1,\dots,n\}\ ,\ m>1\ .
\end{equation}
\end{lemm}
\noindent{\it Proof.} We prove it by induction on $m$. For $m=0$, it is $\Z^i \ket{}=0$. For any $m\geq 0$, the relations \eqref{hdesp10c} imply that 
$\Z^i \der_i^{m+1} v = 
 \left( \der_i \Z^i-\sigma_i \right)\der_i^m v$. By the induction hypothesis, 
$$\hspace{1.2cm}\Z^i \der_i^{m+1} v = - \der_i^m \left( \sum_{k=0}^{m-1} \sigma_i[-k\varepsilon_i] \right) v - \der_i^m \sigma_i[-m\varepsilon_i] v = - (\der_i)^m \left( \sum_{k=0}^m \sigma_i[-k\varepsilon_i] \right) v \ .\hspace{1.2cm}\square$$

For brevity, we write $f \vert_{\lb} := f(\lambda_1,\dots,\lambda_n)$ for $f \in \Un$.

\begin{theo}\label{evalsigmairredmodule}
The module $\Mc$ is irreducible if and only if for any $i = 1,\dots,n$ 
\begin{equation}\label{condforirrll}
 \left(\sigma-\sigma[-(d_i+1) \varepsilon_i]\right) \vert_{\lb}=0 \ \ \text{and}\  \ 
\left(\sigma-\sigma[-k_i \varepsilon_i]\right)\vert_{\lb}\neq 0 \ , \ 1 \leq k_i \leq d_i  \ .\end{equation}
\end{theo}
\noindent{\it Proof.} This follows from Proposition \eqref{lemmanonzeroirre} and Lemma \eqref{actiondxmLWVprop} since 
$$\hspace{2.8cm}\sum_{k=0}^{m-1} \sigma_i[-k\varepsilon_i] = \sum_{k=0}^{m-1} \left( \sigma[-k\varepsilon_i] - \sigma[-(k+1)\varepsilon_i] \right) = \sigma - \sigma[- m \varepsilon_i] \ .\hspace{2.8cm}\square$$

\paragraph{Examples.} 
{\bf 1.} $\sigma = H_1$, see \eqref{homogeneouspoly11}. The algebra $\Diffv{H_1}(n) = \Diff(n)$ has no finite-dimensional modules since
for any $(d_1,\dots,d_n) \in \mathZ^n$ and any $\vec{\lambda}$, 
$$ \left( H_1-H_1[-(d_i+1) \varepsilon_i]\right)\vert_{\lb} = (d_i+1) \neq 0 \quad , \quad i = 1,\dots,n\ .$$

\noindent{\bf 2.} $\sigma = H_2$. The algebra $\Diffv{H_2}(n)$ has no finite-dimensional modules if $n\geq 2$.  
Indeed, for any $(d_1,\dots,d_n) \in \mathZ^n$ and any $\vec{\lambda}$ the first condition in \eqref{condforirrll} is equivalent to
$$\sum_{k \leq l} \lambda_k \lambda_l = \sum_{k \leq l} (\lambda_k - (d_i+1)\delta^i_k) (\lambda_l - (d_i+1)\delta^i_l)
\ \Longleftrightarrow\  \sum_{l=1}^n \lambda_l = d_i + 1 - \lambda_i\ ,\ 1\leq i\leq n\ . $$
If $n \geq 2$ then for any $i \neq j$, we have $\lambda_i - \lambda_j = d_i - d_j \in \mathbb{Z}$
which is forbidden. However, the algebra $\Diffv{H_2}(1)$ is isomorphic to $\U(\sl_2)$ and therefore admits finite-dimensional modules.

\begin{conj}
$\Diffv{H_m}(n)$ has finite dimensional modules if and only if $m > n$.
\end{conj} 

\subsection{Indecomposable $\Diffs(n)$-modules}\label{Reducibility.}
In this section, $\K = \mathbb{C}$. We give three examples of indecomposable $\Diffs(n)$-modules.

\vskip .1cm
\noindent{\bf Non-diagonalizable action of $\U(n)$.}
Set $\Z^1 \mapsto 0, \Z^2 \mapsto 0, \der_1 \mapsto 0, \der_2 \mapsto 0$ and
\begin{equation}\label{formh12nondiag}
\th_1 \mapsto  \frac{1}{2} \left( \begin{array}{cc}
1-i & 1 \\
0 & 1-i
\end{array} \right) \qquad \text{and} \qquad
\th_2 \mapsto  \frac{1}{2} \left( \begin{array}{cc}
1+i & 1 \\
0 & 1+i
\end{array} \right) \ .
\end{equation}
This defines an indecomposable two dimensional $\Diffv{H_4}(2)$-module since
$$\sigma_1 = H_4(\th_1,\th_2) - H_4(\th_1-\mathbbm{1},\th_2) \mapsto 0 \ \ \text{and} \ \ \sigma_2 = H_4(\th_1,\th_2) - H_4(\th_1,\th_2-\mathbbm{1})\mapsto 0\ .$$

\vskip .1cm
\noindent{\bf Diagonalizable action of $\U(n)$.} Set $\Z^2 \mapsto 0, \der_1 \mapsto 0, \der_2 \mapsto 0$ and
\begin{equation}\label{actinh1h2module211}
\Z^1 \mapsto \left( \begin{array}{cc}
0 & 0 \\
1 & 0
\end{array}\right) \quad , \quad
\th_1 \mapsto \left( \begin{array}{cc}
0 & 0 \\
0 & 1
\end{array} \right) \quad \text{and} \quad
\th_2 \mapsto \omega \left( \begin{array}{cc}
1 & 0 \\
0 & 1
\end{array} \right)
\end{equation}
with $\omega = \displaystyle e^{\frac{i \pi}{3}}$. 
This defines an indecomposable two dimensional $\Diffv{H_6}(2)$-module since
$$\sigma_1 = H_6(\th_1,\th_2) - H_6(\th_1-\mathbbm{1},\th_2) \mapsto 0 \ \ \text{and} \ \ \sigma_2 = H_6(\th_1,\th_2) - H_6(\th_1,\th_2-\mathbbm{1}) \mapsto 0 \ .$$

\vskip .1cm
\noindent{\bf Non polynomial $\sigma$.}  
Let $\vec{\lambda}=\{\lambda_1,\dots,\lambda_n\}\in\mathbb{C}^n$ be such that $\lambda_k - \lambda_l \notin \mathbb{Z}$ if $k \neq l$. 
Define polynomials $A_1(X) = \prod_{s=0}^{p-1} (X - \lambda_1+s)$  
and $A_j(X) = (X - \lambda_j)(X-\lambda_j+1)$, $j>1$. 
Set $\der_i \mapsto 0$, $i=1,\dots,n$, $\Z^j \mapsto 0$, $\th_j \mapsto \lambda_j \mathbbm{1}$, $j=2,\dots,n$, and 
\begin{equation}\label{formhidiagodimp}
\th_i \mapsto \left( \begin{array}{cccc}
\lambda_1 & & & \\
 & \lambda_1 - 1 & & \\
 & & \ddots & \\
 & & & \lambda_1 - p + 1
\end{array} \right) \quad \text{and} \quad 
\Z^i \mapsto \left( \begin{array}{cccc}
0 & 1 & & \\
 & \ddots & \ddots & \\
 & & \ddots & 1 \\
 & & & 0
\end{array} \right) \ ,
\end{equation}
so $\th_i \Z^i - \Z^i (\th_i+1)\mapsto 0$. 
This defines an indecomposable $p$-dimensional $\Diffv{\sigma}(2)$-module with with
$\sigma = \sum_{l=1}^n \frac{A_l(\th_l)}{\chi_l}$ since
$A_k(\th_k) ,A_k(\th_k - \mathbbm{1}) \mapsto 0$, $k=1,\dots,n$, so
$$\sigma_j = \sum_{l=1}^n \left( \frac{A_l(\th_l)}{\chi_l} - \frac{A_l(\th_l-\delta_{lj})}{\chi_l[-\varepsilon_j]} \right) \mapsto 0 \quad , \quad j = 1,\dots,n \ .$$

\chapter{ Rings $\Diff(n,N)$}\label{Chaptergenesevcop}

\vskip -1cm
In this chapter we generalize the results of \textbf{Chapter} 
\textbf{\ref{GeneralizedDiff}} to several copies of the $\h$-deformed coordinate rings.  
In \textbf{Section \ref{DefRelaDiffSigma}} we prove that for the ring which is built on the $\h$-deformed coordinate rings $V(n,N)$ and $V^\star(n,N')$ with $NN'>1$, the PBW property imposes severe restrictions on the constant terms in the equation \eqref{relaxndn3}. Essentially only one possibility is left: the ring $\Diff(n,N)$ consisting of $N$ copies of the ring $\Diff(n)$.
In \textbf{Section \ref{dingdifnN}} we define a $\K$-subring of $\Diff(n,N)$ which is a homomorphic image of $\U(\gl_N)$. With the help of this subring, we 
describe the subspace of quadratic central elements of $\Diff(n,N)$. In contrast to the ring $\Diff(n)$, it is one-dimensional for $N>1$. 
In \textbf{Section \ref{ActionSn}} we present an action of the symmetric group $\mathbb{S}_n$ on the ring $\Diff(n,N)$ and on the diagonal reduction algebra $\DRn$, generalizing the formulas for $\Diff(n)$. 

\section{Several copies}\label{DefRelaDiffSigma}
\vskip -.2cm
Let $\mathcal{L}$ be the ring with the generators $\Z^{i\alpha}$, $i=1,\dots,n$, $\alpha=1,\dots,N'$, and $\der_{j\beta}$, $j=1,\dots,n$, $\beta=1,\dots,N$ subject to the following defining relations. 
The $\h(n)$-weights of the generators are given by (\ref{hdesp1nn}) and (\ref{hdesp3}). The generators $\Z^{i\alpha}$ satisfy the relations (\ref{hdesp2}). The generators $\der_{j\beta}$ satisfy the relations (\ref{hdesp4}). We impose the general oscillator-like cross-commutation relations, compatible with the $\h(n)$-weights, between the generators $\Z^{i\alpha}$ and $\der_{j\beta}$:
\begin{equation}\label{secocrore}
\Z^{i\alpha}\der_{j\beta}=\sum_{k,l}\der_{k\beta} \RR_{lj}^{ki} \Z^{l\alpha}-\delta_{j}^i\sigma_{i\alpha\beta} \ ,\ 1\leq i,j\leq n\ ,\ 1\leq\alpha\leq N'\ ,\ 
1\leq\beta\leq N\ ,\ \sigma_{i\alpha\beta}\in\bar{\U}(n)\ .
\end{equation}

\subsection{PBW property}\label{SectionPBWDiffnN}

\begin{lemm}\label{seco}
Assume that at least one of the numbers $N$ and $N'$ is bigger than 1. Then the ring $\Diff(\Sigma)$ has the 
PBW is and only if 
\begin{equation}\label{posinb1}
\sigma_{i\alpha\beta} = \sigma^\alpha_\beta \ \ \text{for some}\ \ \sigma^\alpha_\beta \in \K \ .
\end{equation}
\end{lemm}
\noindent{\it Proof.} Assume that, say, $N>1$. Repeating the calculations \eqref{fiwa} and  \eqref{sewa} for one copy in \textbf{Section \ref{PBWPorpertySection}}, we find, for $i,j,k=1,\dots,n$, 
\begin{equation}\label{fiwanb1}\begin{array}{rcl}
\left( x^{i\alpha}\der_{j\beta}\right) \der_{k\gamma}\ 
\rule[-3.5mm]{.28mm}{7mm}_{\;\text{l.d.t.}}&=&\displaystyle{ 
\left( \sum_{u,v}\RR^{ui}_{vj}[\varepsilon_u]\der_{u\beta}x^{v\alpha}-\delta^i_j\sigma_{i\alpha\beta}\right)\der_{k\gamma}
\ \rule[-3.5mm]{.28mm}{7mm}_{\;\text{l.d.t.}} }\\[1.2em]
&=&\displaystyle{ -\sum_{u}\RR^{ui}_{kj}[\varepsilon_u]\der_{u\beta}\sigma_{k\alpha\gamma}-\delta^i_j\sigma_{i\alpha\beta}\der_{k\gamma}\ .}\end{array}\end{equation}
\begin{equation}\label{sewanb1}\begin{array}{rcl}
x^{i\alpha}\left( \der_{j\beta} \der_{k\gamma}\right)\rule[-3.5mm]{.28mm}{7mm}_{\;\text{l.d.t.}}&=&\displaystyle{ x^{i\alpha}\sum_{a,b}\RR^{ab}_{kj}
\der_{b\gamma}\der_{a\beta}\ \rule[-3.5mm]{.28mm}{7mm}_{\;\text{l.d.t.}} } \\[1.2em]
&=&\displaystyle{
\sum_{a,b}\RR^{ab}_{kj}[-\varepsilon_i]\left( \sum_{c,d}\RR^{ci}_{db}[\varepsilon_c]\der_{c\gamma}x^{d\alpha}-\delta^i_b\sigma_{i\alpha\gamma}\right)\der_{a\beta}
\ \rule[-3.5mm]{.28mm}{7mm}_{\;\text{l.d.t.}} }\\[1.2em]
&=&\displaystyle{ 
-\sum_{a,b,c}\RR^{ab}_{kj}[-\varepsilon_i] \RR^{ci}_{ab}[\varepsilon_c]\der_{c\gamma}\sigma_{a\alpha\beta}-\sum_{a}\RR^{ai}_{kj}[-\varepsilon_i] 
\sigma_{i\alpha\gamma}\der_{a\beta} }\ .\end{array}\end{equation}
Take $\beta \neq \gamma$.
Equating the coefficients in $\der_{u\beta}$, $u=1,\dots,n$, in \eqref{fiwanb1} and \eqref{sewanb1}, we find 
\begin{equation}\label{ysy1}\RR^{ui}_{kj}[\varepsilon_u]\sigma_{k\alpha\gamma}[\varepsilon_u]=\RR^{ui}_{kj}[-\varepsilon_i]\sigma_{i\alpha\gamma}
\ ,\ i,k,j,u=1,\dots,n\ .\end{equation}
Equating the coefficients in $\der_{u\gamma}$, $u=1,\dots,n$, in \eqref{fiwanb1} and \eqref{sewanb1}, we find    
\begin{equation}\label{ysy2}\delta^i_j\delta^u_k\sigma_{i\alpha\beta}=\sum_{a,b}\RR^{ab}_{kj}[-\varepsilon_i] \RR^{ui}_{ab}[\varepsilon_u]
\sigma_{a\alpha\beta}[\varepsilon_u]\ ,\ i,k,j,u=1,\dots,n\ .\end{equation}
Shifting by $-\varepsilon_u$ and using the property \eqref{weze}, we rewrite the equality \eqref{ysy1} in the form
\begin{equation}\label{kosha1}
\RR^{ui}_{kj}\left(\sigma_{k\alpha\gamma}-\sigma_{i\alpha\gamma}[-\varepsilon_u]\right)=0\ .
\end{equation}
Setting $u=k$ and $j=i$ (with arbitrary $i,k=1,\dots,n$) in \eqref{kosha1}, we obtain
$\sigma_{k\alpha\gamma}=\sigma_{i\alpha\gamma}[-\varepsilon_k]\ $
which implies the assertion \eqref{posinb1}. 

A direct calculation, with the help of the properties \eqref{weze}, \eqref{iceco} and \eqref{Q5do} of the operator $\RR$, shows that the condition \eqref{posinb1} implies the equalities \eqref{ysy1} and \eqref{ysy2} as well as all the remaining conditions for the flatness of the deformation.
\hfill$\square$
Making the redefinitions of the generators, $\Z^{i\alpha}\rightsquigarrow A^{\alpha}_{\alpha'}\Z^{i\alpha'}$ and $\der_{i\beta}\rightsquigarrow B_{\beta}^{\beta'}\der_{i\beta'}$ 
with some $A\in GL(N',\mathbb{K})$ and $B\in GL(N,\mathbb{K})$ we can transform the matrix $\sigma_{\alpha\beta}$ to the diagonal form, with the diagonal $(1,\dots,1,0,\dots,0)$. Therefore, the ring $\mathcal{L}$ is formed by several copies of the rings $\Diff(n)$, $\V(n)$ and $\V^*(n)$. 
The ring with $\sigma_{\alpha\beta}=\delta_{\alpha\beta}$ is formed by several copies  of the ring $\Diff(n)$ only. It plays an important role in the theory of diagonal reduction algebras \cite{KO5,KO6}. We denote this ring by $\Diff(n,N)$. 

\section{Ring $\Diff(n,N)$}\label{dingdifnN}
In components, the defining relations \eqref{hdesp2}, \eqref{hdesp4} and \eqref{secocrore} for the ring $\Diff(n,N)$ read 
\begin{eqnarray}
\label{relasevco1}
&\displaystyle{ \Z^{i,\alpha} \Z^{j,\beta} = \frac{1}{\th_{ij}} \Z^{i,\beta} \Z^{j,\alpha} + \frac{\th_{ij}^2-1}{\th_{ij}^2} \Z^{j,\beta} \Z^{i,\alpha}\ ,\ 
\Z^{j,\alpha} \Z^{i,\beta} = - \frac{1}{\th_{ij}} \Z^{j,\beta} \Z^{i,\alpha} + \Z^{i,\beta} \Z^{j,\alpha} \ ,}&\\
\label{relasevco2}
&\displaystyle{ \der_{i,\alpha} \der_{j,\beta} = -\frac{1}{\th_{ij}} \der_{i,\beta} \der_{j,\alpha} + \frac{\th_{ij}^2-1}{\th_{ij}^2} \der_{j,\beta} \der_{i,\alpha}\ ,\ 
\der_{j,\alpha} \der_{i,\beta} = \frac{1}{\th_{ij}} \der_{j,\beta} \der_{i,\alpha} + \der_{i,\beta} \der_{j,\alpha} \ ,}&\\
\label{relasevco3}
&\displaystyle{ \Z^{i,\alpha} \der_{j,\beta} = \der_{j,\beta} \Z^{i,\alpha}\ ,\ \Z^{j,\alpha} \der_{i,\beta} = \frac{\th_{ij}(\th_{ij}+2)}{(\th_{ij}+1)^2} \der_{i,\beta} \Z^{j,\alpha}}&
\end{eqnarray}
for all $1 \leq i < j \leq n$, and
\begin{equation}\label{relasevco4}
\Z^{i,\alpha} \der_{i,\beta} = \sum_{k=1}^n \frac{1}{1+\th_{ki}} \der_{k,\beta} \Z^{k,\alpha} - \delta^\alpha_\beta\ ,\ 1 \leq i \leq n\ ,\ \alpha,\beta=1,\dots,N
\end{equation}

\subsection{Quadratic central elements}

Propositions \ref{quceel} and \ref{isocenter} tell us that, in $\Diff(n)$, there is $n$ quadratic elements generating the center. For $N>1$, we prove that there is only one quadratic central element in $\Diff(n,N)$.

Let 
$$A^\alpha_\beta = \sum_{i} \der_{i\beta} \Z^{i\alpha} \quad , \quad \alpha, \beta = 1,\dots,N \ .$$

\begin{prop}\label{propcommael}
In $\Diff(n,N)$, we have the relations
\begin{equation}\label{relaAxd}
[ A_\alpha^\beta , \th_i ] = 0 \ , \ [ A_\alpha^\beta , \Z^{i\gamma} ] = \delta_\alpha^\gamma \, \Z^{i\beta} \ \text{and} \ [ A_\alpha^\beta , \der_{i\gamma} ] = - \delta_\gamma^\beta \, \der_{i\alpha}\ ,\ 1\leq i\leq n\ ,\ 1\leq \alpha, \beta, \gamma\leq N \ .
\end{equation}
\end{prop}
\begin{proof}
A direct calculation.
\end{proof}
It follows that $[A^\alpha_\beta , A^\gamma_\rho] = \delta^\alpha_\rho A^\gamma_\beta - \delta^\gamma_\beta A^\alpha_\rho$, so the elements $A^ \alpha_\beta$, 
generate a $\K$-subring of $\Diff(n,N)$ which is a homomorphic image of the enveloping algebra $\U(\gl_N)$.

\vskip .1cm
We define the elements
\begin{equation}\label{expressionM}
M^i_j = \sum_{\beta=1}^N \der_{j\beta} \Z^{i\beta} \quad , \quad i,j=1,\dots,n \ .
\end{equation}

\begin{lemm}\label{calculMx}
For any $i \neq j$ in $\{1,\dots,n\}$ and $\alpha \in \{1,\dots,N\}$, we have
$$[M^j_j , \Z^{i\alpha}] = \sum_{\beta=1}^N \der_{j\beta} \Z^{j\alpha} \Z^{i\beta} \frac{1}{\th_{ji}} \quad \text{and} \quad [M^i_i , \Z^{i\alpha}] = - \sum_{\beta=1}^N \sum_{j\neq i} \der_{j\beta} \Z^{j\alpha} \Z^{i\beta} \frac{1}{\th_{ji}} + \Z^{i\alpha} \ .$$
\end{lemm}
\begin{proof}
It is done by direct calculations using the relations \eqref{relasevco1}-\eqref{relasevco4}.
\end{proof}
\begin{prop}
Any quadratic central element of $\Diff(n,N)$, $N>1$, is proportional to $\sum A_\alpha^\alpha-H_1$.
\end{prop}
\begin{proof} If an element of the form $\sum \der_{j\alpha}\Z^{i\beta}f^{j\alpha}_{i\beta}+g$, $f^{j\alpha}_{i\beta},g\in\Un$, is central then by the weight considerations and Proposition \ref{propcommael}, it is of the form 
$$C = \sum_{i=1}  M^i_i f_i+ g\ ,\ f_i,g \in \Un\ .$$
Lemma \ref{calculMx} implies
\begin{equation}\label{calculcomCx}
[C,\Z^{i\alpha}] = \sum_{\beta} \sum_{j } \der_{j\beta} \left( \Z^{j\alpha}\Z^{i\beta} \frac{f_j-f_i}{\th_{ji}} + \Z^{j\beta} \left( \Delta_i f_j \right) \Z^{i\alpha} \right) + \left( \Delta_i g + f_i[-\varepsilon_i] \right) \Z^{i\alpha} \ .
\end{equation}
For $N>1$, the right hand side is zero iff all $f_i$ are equal, $f_i = f$, $\Delta_i f = 0$ and $\Delta_i g = -f$. This concludes the proof.
\end{proof}

\section{Action of symmetric group}\label{ActionSn}
\paragraph{Action of $\mathbb{S}_n$ on $\Diff(n,N)$.}
The symmetric group $\mathbb{S}_n$ acts on $\Diff(n)$, see  \eqref{acofsn1}.
We do not know an analogue of the isomorphism $\mu$, see \eqref{fromwtodiff}, for the ring $\Diff(n,N)$. However, straightforward analogues of the formulas \eqref{acofsn1} provide an action of $\mathbb{S}_n$ by automorphisms on the ring $\Diff(n,N)$.
\begin{prop}
The maps $s_i$, $i=1,\dots,n-1$, defined on the generators of $\Diff(n,N)$ by 
\begin{equation}\label{acsnN1}
\begin{array}{l}
{\displaystyle
s_i(\Z^{i\alpha}) = - \Z^{i+1,\alpha} \th_{i,i+1}\ , \ s_i(\Z^{i+1,\alpha}) = \Z^{i\alpha} \frac{1}{\th_{i,i+1}}\ , \ s_i(\Z^{j\alpha}) = \Z^{j\alpha}\ 
\text{for}\ j\neq i,i+1\ ,}\\[.1em]
{\displaystyle 
s_i(\der_{i\alpha}) = - \frac{1}{\th_{i,i+1}} \der_{i+1,\alpha}\ , \ s_i(\der_{i+1,\alpha}) = \th_{i,i+1} \der_{i\alpha}\ ,\ 
s_i(\der_{j\alpha}) =\der_{j\alpha}\ \text{for}\ j\neq i,i+1\ ,}\\[.8em]
s_i(\th_j)=\th_{s_i(j)} \ ,
\end{array}
\end{equation}
extend to automorphisms of the ring $\Diff(n,N)$. Moreover, these automorphisms satisfy the Artin relations and therefore give the action of the  
symmetric group $\mathbb{S}_n$ by automorphisms.
\end{prop}
\begin{proof} After the formulas (\ref{acsnN1}) are written down, the verification is a direct calculation. 
\end{proof}
\noindent{\bf Action of $\mathbb{S}_n$ on $\mathbf{L}(n,N)$ and $\DRn$.}
The diagonal reduction algebra $\DRn$ is   
generated by the elements $\teL_i^j$, $i,j=1,\dots,n$. The defining relations 
of $\DRn$ are given by the reflection equation, see \cite{KO5}
$$\RR_{12}\teL_1\RR_{12}\teL_1-\teL_1\RR_{12}\teL_1\RR_{12}=\RR_{12}\teL_1-\teL_1\RR_{12}\ ,$$
where $\teL=\{\teL_i^j\}_{i,j=1}^n$ is the matrix of generators (we refer to \cite{C,Sk,RS,KS,IO,IOP,IMO1,IMO2} for various aspects and applications 
of the reflection equation). 

\vskip .1cm
For each $N$ there is a homomorphism (\cite{KO5}, Section 4.1)
$$\tau_N\colon \DRn\to\Diff(n,N)\ \ \text{defined by}\ \ 
\tau_N(\teL_i^j)=\sum_\alpha \Z^{j,\alpha}\der_{i,\alpha}\ .$$
Moreover $\tau_N$ is an embedding for $N\geq n$. 

\vskip .1cm
The formulas (\ref{acsnN1}) show that the image of $\tau_N$ is preserved by the automorphisms $s_i$. 

\vskip .1cm
The element $s_i(\tau_N(\teL^j_k))$ 
can be written by the same formula for all $N$. Since $\tau_N$ is injective for $N\geq n$ we conclude that the formulas (\ref{acsnN1})
induce the action of the symmetric group $\mathrm{S}_n$ on the diagonal reduction algebra $\DRn$ by automorphisms. 
The action on the generators $\teL_j^k$, $j,k=1,\dots,n$  is given by
\begin{equation}\label{defsymactL}
\begin{array}{l}
\displaystyle{ s_i( \teL^i_j) = -  \teL^{i+1}_j\th_{i,i+1} \ ,\ s_i( \teL^{i+1}_j) =  \teL^{i}_j\frac{1}{\th_{i,i+1}} \ ,\ j\neq i,i+1 \ ,}\\[.3em]
\displaystyle{ s_i( \teL_i^j) = -\frac{1}{\th_{i,i+1}}  \teL_{i+1}^j\ ,\ s_i( \teL_{i+1}^j) = \th_{i,i+1}  \teL^{i}_j\ ,\ j\neq i,i+1\ ,}\\[1.2em]
\displaystyle{ s_i( \teL^i_i) =  \teL^{i+1}_{i+1}\ ,\ s_i( \teL^i_{i+1}) = - \teL^{i+1}_{i}(\th_{i,i+1}-1)^2\ ,}\\[.6em]
\displaystyle{ s_i( \teL^{i+1}_i) = - \teL^{i}_{i+1}\frac{1}{(\th_{i,i+1}+1)^2}\ ,\ s_i( \teL^{i+1}_{i+1}) =  \teL^{i}_{i}\ ,} \\[1.2em]
\displaystyle{ s_i( \teL^k_j) =  \teL^k_j\ ,\ k\neq i,i+1\ \text{and}\ j\neq i,i+1\ .}
\end{array}
\end{equation}

\newpage

\centerline{\bf Abstract}

\vskip .2cm
The ring $\Diff(n)$ of $\h$-deformed differential operators appears in the theory of reduction algebras. In this thesis, we construct the rings of generalized differential operators on the $\h$-deformed vector spaces of $\gl$-type. In contrast to the $q$-deformed vector spaces for which the ring of differential operators is unique up to an isomorphism, the general ring of $\h$-deformed differential operators $\Diffs(n)$ is labeled by a rational function $\sigma$ in $n$ variables, satisfying an over-determined system of finite-difference equations. We obtain the general solution of the system. We show that the center of $\Diffs(n)$ is a ring of polynomials  in $n$ variables. We construct an isomorphism between certain localizations of $\Diffs(n)$ and the Weyl algebra $\W_n$ extended by $n$ indeterminates. We present some conditions for the irreducibility of the finite dimensional $\Diffs(n)$-modules. Finally, we discuss difficulties for finding analogous constructions for the ring $\Diff(n, N)$ formed by several copies of $\Diff(n)$.

\vskip .2cm\noindent
{\bf Key words:} differential operators, Yang-Baxter equation, reduction algebras, universal enveloping
algebra, representation theory, Poincar\'e--Birkhoff--Witt property, rings of fractions.

\vspace{10ex}

\centerline{\bf R\'esum\'e}

\vskip .2cm
L'anneau $\Diff(n)$ des op\'erateurs diff\'erentiels $\h$-deform\'es appara\^it dans la th\'eorie des alg\`ebres de r\'eduction.
Dans cette th\`ese, nous construisons les anneaux des op\'erateurs diff\'erentiels g\'en\'eralis\'es sur les espaces vectoriels $\h$-deform\'es de type $\gl$. Contrairement aux espaces vectoriels $q$-deform\'es pour lesquel l'anneau des op\'erateurs diff\'erentiels est unique \`a isomorphisme pr\`es, l'anneau g\'en\'eralis\'e des op\'erateurs diff\'erentiels $\h$-deform\'es $\Diffs(n)$ est index\'e par une fonction rationnelle $\sigma$ en $n$ variables, solution d'un syst\`eme d\'eg\'en\'er\'e d'\'equations aux diff\'erences finies. Nous obtenons la solution g\'en\'erale de ce syst\`eme. Nous montrons que le centre de $\Diffs(n)$ est un anneau des polyn\^omes en $n$ variables. Nous construisons un isomorphisme entre des localisations de l'anneau $\Diffs(n)$ et de l'alg\`ebre de Weyl $\text{W}_n$ \'etendue par $n$ ind\'etermin\'es. Nous pr\'esentons des conditions d'irr\'eductibilit\'e des modules de dimension fini de $\Diffs(n)$. Finalement, nous discutons des difficult\'es \`a trouver les constructions analogues pour l'anneau $\Diff(n,N)$ correspondant \`a $N$ copies de $\Diff(n)$.

\vskip .2cm\noindent
{\bf Key words:} op\'erateurs diff\'erentiels, \'equation de Yang-Baxter, alg\`ebres de r\'eduction, alg\`ebre enveloppante universelle, th\'eorie des repr\'esentations, propri\'et\'e de Poincar\'e--Birkhoff--Witt, corps des fractions.

\end{document}